%% file: main.tex
\documentclass[11pt,a4paper]{article}
\pdfoutput=1

\usepackage{jheppub}
\makeatletter
\gdef\@fpheader{}
\makeatother

\usepackage{amsthm}
\usepackage{mathrsfs}

\usepackage{booktabs}
\usepackage{array}
\usepackage{enumitem}

\usepackage{orcidlink}

\theoremstyle{definition}
\newtheorem{theorem}{Theorem}[section]
\newtheorem{proposition}[theorem]{Proposition}
\newtheorem{corollary}[theorem]{Corollary}
\newtheorem{definition}[theorem]{Definition}
\newtheorem{assumption}[theorem]{Assumption}

\input{macros}

\input{frontmatter-jhep}

\hypersetup{
  pdftitle={Coarse-Graining and Long-Range Response in Generated Systems},
  pdfauthor={Jinku Guo},
  pdfsubject={Coarse-graining, preservation of linear response, and isolated spectral poles},
  pdfkeywords={Coarse-graining; Linear response; Operator approximation; Schur complements; Memory kernels; Analytic Fredholm theory}
}

\begin{document}

\maketitle\flushbottom

\input{sections-en/sec01_introduction}
\input{sections-en/sec02_three_criteria}
\input{sections-en/sec03_operator_classification}
\input{sections-en/sec04_irreversibility}
\input{sections-en/sec05_synergy}
\input{sections-en/sec06_summary}

\appendix
\input{sections-en/app_kl_representation}
\input{sections-en/app_operator_hierarchy}
\input{sections-en/app_kernel_proof}

\input{sections-en/declarations}

\bibliographystyle{JHEP}
\bibliography{references}

\end{document}

%% file: macros.tex
\newcommand{\diag}{\operatorname{diag}}

\newcommand{\spec}{\operatorname{spec}}

\newcommand{\Ran}{\operatorname{Ran}}
\newcommand{\Ker}{\operatorname{Ker}}

\newcommand{\TT}{\mathrm{TT}}

\newcommand{\iu}{\mathrm{i}\mkern1mu}

\newcommand{\cO}{\mathcal{O}}
\newcommand{\cK}{\mathcal{K}}
\newcommand{\cB}{\mathcal{B}}
\newcommand{\cG}{\mathcal{G}}
\newcommand{\cH}{\mathcal{H}}
\newcommand{\cC}{\mathfrak{C}}
\newcommand{\cI}{\mathfrak{I}}

%% file: frontmatter-jhep.tex
\title{Coarse-Graining and Long-Range Response in Generated Systems}

\author[a]{Jinku Guo\,\orcidlink{0009-0000-6600-6171}}
\affiliation[a]{Northwestern Polytechnical University, Xi'an 710072, China}
\emailAdd{guojk@nwpu.edu.cn}

\abstract{%
We study when a reduced description preserves the response to an applied
source.
For systems defined by continuously differentiable state updates and
observations a bounded surjective linear coarse map gives differentiable
descent when the projected maps are constant on each fibre.
Their common source response is then preserved on the invertible domain.
For approximate descent an exact error identity separates defects in
the update and source coupling from defects in the readout and direct term.
Resolvent and contour estimates carry these errors and the displacement
of the reference branch to visible poles and their residues at shifted
critical roots.
In a reciprocal Hilbert-space realisation an isolated finite-dimensional
critical subspace with a positive crossing matrix has a
positive-semidefinite finite-rank residue when the endpoint operator
is self-adjoint and the complementary inverse is bounded.
Source overlap determines its visible directions.
Arithmetic and reflecting-heat examples separate stationary response
preservation from memoryless dynamical closure.
In the two-mode example pole displacement decays geometrically with
the retained memory length.
The residue error includes an additional factor of that length.
For a matched vacuum covariance the residue is the zero-mass atom
of its positive spectral measure.
This identification uses the same physical projection and source pairing
together with the continuation and subtraction conventions and endpoint
matching.
}

\keywords{%
 Coarse-graining; Linear response; Operator approximation; Schur complements;
 Memory kernels; Analytic Fredholm theory
}

%% file: sections-en/sec01_introduction.tex
\section{Introduction}
\label{sec:introduction}

A macroscopic description replaces many microscopic variables by a few
collective quantities. Agreement on an equilibrium value leaves open
how those quantities respond to an applied disturbance.
Unresolved variables can return the disturbance to the retained state.
They can also alter its source coupling and observed output.
Near a collective critical mode these effects can move a response pole
or substantially change its strength. The physical question is which
microscopic distinctions can be discarded while preserving that response.

Two microscopic states may give the same present observation and respond
differently to a later intervention. A coarse state must therefore retain
the distinctions visible under the allowed future operations.
We formulate this requirement through a sourced state equation and an
observation map. Differentiating both follows a disturbance from its
entry into the system to its measured effect.
Joint differentiable descent preserves this derivative.
Approximate descent introduces errors whose amplification depends on
the collective modes.

For specified generating operations and source interventions two states
are equivalent when every allowed future experiment gives the same output.
This relation defines the coarsest observation-sufficient quotient.
Its universal property is set-theoretic. Differentiable response requires
additional topological structure. For a bounded surjective linear map
between Banach spaces constancy of the projected update and observation
on each fibre gives a unique differentiable coarse realisation.
The operations and observations determine which distinctions must be
retained. Regularity determines whether the resulting description supports
linear response.

The future-output equivalence extends Moore's state distinction for finite
sequential machines~\cite{Moore:1956}.
Ho--Kalman realisation reconstructs a minimal linear state space from
input--output data~\cite{HoKalman:1966}.
Here the full and coarse descriptions are specified state equations
that may be nonlinear on Banach spaces. Fibre consistency gives exact
descent. An identity for the response difference measures its failure
under approximation.

The derivative has four parts. The source map carries the disturbance
into the state. The inverse linearised state equation sums its repeated
transmission through coupled variables. The readout converts the state
change into an output. The direct term records the change in that output
at fixed state. Theorem~\ref{thm:visible-descent} proves preservation of
their combination under joint descent. Source-dependent coarse
coordinates contribute compensating derivatives to the source map and
direct term. The implicit state equation determines the same response
for any numerical iteration that solves it.

For approximate descent the response difference separates into one defect
for each of these four maps.
The comparison map may remove state coordinates.
The estimates include the displacement of the reference branch and
extend to contour integrals of the response.
Schur--Feshbach reduction~\cite{Feshbach:1958} gives the complementary
description by exact block elimination.
It changes the source and readout together with the direct term and
effective kernel. All four enter the observed residue.
In an evolution equation the eliminated coordinates carry an initial-state
contribution and memory of earlier states and sources
\cite{Mori:1965}.
Both the block inverse and the memory mechanism are classical.
The additional result here compares the complete response within one
source--observation system through irreversible reduction and
reference-branch displacement.
The four-defect identity of Theorem~\ref{thm:irreversible-error}
and the displaced-root estimates of
Proposition~\ref{prop:closure-stability} carry errors in the four maps
to visible poles and their residues.

Hydrodynamics illustrates the same physical distinction.
Conservation laws identify the slow sectors near local equilibrium.
Constitutive data determine the response inside those sectors.
Eliminated variables leave a delayed influence in the constitutive
relations. Scale separation can justify replacing this memory by a
local approximation.

Projection-based reduction of linear dynamical systems compares
input--output transfer functions through global error norms for
asymptotically stable systems
(ref.~\cite{BennerGugercinWillcox:2015} Sections~2.2--2.4).
Our comparison concerns a neighbourhood of an isolated critical response.
It controls the root and residue using local resolvent bounds.
The two sets of response maps may be derivatives of nonlinear state
equations at different reference states.
Their comparison retains that state displacement and the slope of the
critical zero.

Analytic Fredholm theory describes the isolated singularities of the
constructed response
(ref.~\cite{Kato:1995} Chapter~VII).
A simple crossing of a kernel eigenvalue through one gives a resolvent
pole. Source and readout overlaps determine whether the pole is observed.
In a reciprocal Hilbert-space realisation a positive crossing matrix
on an isolated critical subspace gives a positive finite-rank residue.
Norm and Riesz-projection estimates control its change under approximation
(ref.~\cite{Chatelin:1983} Chapter~5).
Resolvent bounds retain the sensitivity of nonnormal kernels
\cite{TrefethenEmbree:2005}.

An arithmetic merge from ref.~\cite{Guo:2026arithmetic} and a two-mode
reflecting-heat system distinguish static response preservation from
memoryless dynamical closure.
Their finite identities show how a hidden state can be removed from a
stationary susceptibility while remaining visible to future observations.
The reflecting-heat example gives explicit errors for finite memory.
Its pole displacement decays geometrically with the retained memory
length. The residue error contains an additional factor of that length.
This factor comes from differentiating the memory remainder.

Gaussian averaging supplies a continuous family of operations.
We prove its characterisation within the class of centred and
orthogonally invariant convolution semigroups that are weakly continuous
and have finite nonzero variance and square-root self-similarity.
This is a special case of stable-law theory~\cite{Sato:1999}.
External averaging changes the external momentum multiplier.
Averaging relative coordinates acts inside the composite response.
An information-losing quotient instead changes which state distinctions
are represented. The reflecting-window example relates the averaging law
to calibrated correlation endpoints and a specified boundary condition.

Three uniqueness statements occur in the construction.
The allowed future operations determine the observation-sufficient
quotient up to its universal property.
Composition and square-root self-similarity select the Gaussian law
within the stated class.
A positive spectral representation identifies the residue of the matched
Euclidean correlator with a unique measure atom.
Each statement concerns its own mathematical object.

Quantum field theory gives a physical realisation of the response.
Vacuum positivity and covariance assign spectral content to a centred
connected correlator.
Renormalised operator mixing and physical symmetry quotients specify
its channels. Ward identities and functional flows constrain and
compute their responses.
Matching the physical projection and source pairing together with the
continuation and subtraction conventions identifies a Fredholm residue
with a K\"all\'en--Lehmann atom.
The reciprocal response theorem already gives residue positivity at the
operator level. Spectral matching identifies the states that carry
this weight.
The complete covariant stress tensor provides a restrictive example.
Its null-shell TT atom is excluded by the local tensor-measure argument.
A directly reduced two-helicity covariance gives a positive null-shell
example on a flat background.
Its interpretation as gravity uses the matter couplings and constraints
of an effective action together with the scattering consistency conditions
\cite{Weinberg:1964,WeinbergWitten:1980}.

The tensor application in ref.~\cite{Guo2026EmergentGravity} treats
matter coupling and long-range exchange.
The compact construction of ref.~\cite{Guo2026CompactSpectrumI}
provides the gauge and full matter carriers with their common transport.
It also constructs a Lorentz event carrier and calibrated phase clocks.
Its compensation law transports nonlinear gauge--matter responses
across Gaussian resolutions. Its Yang--Mills response completion gives
a gapped colour-invariant example.
The companion application~\cite{Guo2026CompactSpectrumII}
evaluates effective couplings and mass scales with their matching
relations using the Coarse-Graining Genesis Framework
\cite{Guo2026CoarseGrainingFramework}.
Appendix~G.6 of that work gives a coupled matter--amplitude evolution
with a phase-clock readout and finite-interval propagation bounds.
The model illustrates regular auxiliary elimination along an evolving
background.
These applications connect response preservation to explicit source
algebras and observables.

For time-dependent generated spectra
ref.~\cite{Guo2026GeneratedDynamics} controls the discarded nonlinear
source through interaction estimates and stronger-norm bounds on a common
existence interval.
The resulting state estimates supply concrete data for the local response comparison developed here.

%% file: sections-en/sec02_three_criteria.tex
\section{Generated operations, sufficient coarse states, and response}
\label{sec:criteria}

The choice of a coarse state expresses which microscopic distinctions
matter for the specified physical questions.  Agreement of present
observations is sufficient only if subsequent allowed operations never
reveal a difference.  This requirement first determines a set quotient. Its regular realisation makes it possible to differentiate the response
to a source.  In a physical realisation observations are the specified
physical outputs of the system and need not involve a measuring apparatus.

\subsection{Operations and the observation-sufficient quotient}
\label{sec:generated-operations}

Let $\mathsf X$ be a common invariant state domain and let
$F_a:\mathsf X\to\mathsf X$ be the allowed deterministic operations.
The label $a$ includes each allowed source intervention not just the
unforced update.  A finite word $w$ specifies an ordered composition $F_w$. The empty word acts as the identity.  Let $\mathscr O$ be the declared
observation family.  Its members may have different output spaces but
equality of their outputs is fixed before taking a quotient.  These data
come from a specified protocol.  For finite sequential machines
distinguishing states by their future outputs is the viewpoint of
Moore's experiments~\cite{Moore:1956}.  The same construction applies to infinite state and output sets. Differentiating the response then requires a regular realisation of the quotient.

In a generated physical realisation the operations and allowed sources together with the observations belong to one fixed construction.  The arithmetic operations
of Ref.~\cite{Guo:2026arithmetic} and the spectral products and covariant
transport of Ref.~\cite{Guo2026CompactSpectrumI} give companion
realisations.  Their continuous extensions use the specified common domains
and convergence norms.  For time evolution
Ref.~\cite{Guo2026GeneratedDynamics} gives model-specific interaction
hypotheses and uniform bounds under which finite constructions and their
controlled limits define the same observation.  The quotient applies to these constructions through their fixed operations and observations.

For example observe only $x$ in a state $(x,u)$ whose next observed
coordinate is $x'=ax+bu$ with $b\ne0$.  States with the same $x$ and
different $u$ agree at present but give different outputs after one update.
Observation-sufficient equivalence tests whether such a distinction can appear after any allowed sequence of operations.

\begin{theorem}[Coarsest observation-sufficient quotient]
\label{thm:canonical-quotient}
Define
\begin{equation}
 x\sim y\quad\Longleftrightarrow\quad
 O(F_wx)=O(F_wy)\quad\hbox{for every }O\in\mathscr O\hbox{ and }w.
 \label{eq:observational-equivalence}
\end{equation}
Then $\sim$ is an equivalence relation and every allowed operation and
observation descends to $q:\mathsf X\to\mathsf X/{\sim}$.  If another
map $C$ admits exact descended operations and observations on its image
there is a unique surjection $\pi:C(\mathsf X)\to\mathsf X/{\sim}$ with
$q=\pi C$.  It intertwines the descended operations and observations.
Thus the quotient is unique up to its canonical bijection among coarsest
exact observation-sufficient descriptions.
\end{theorem}

\begin{proof}
Equality of all specified outputs is reflexive and symmetric and transitive.
If $x\sim y$ append $F_a$ at the input of any word to obtain
$F_ax\sim F_ay$.  The empty word gives $O(x)=O(y)$ so both operations
and observations descend.  If $Cx=Cy$ induction along a word using the
descended $C$-operations gives $CF_wx=CF_wy$. The descended observations
then give $x\sim y$.  Therefore $\pi(Cx)=[x]$ is well defined and
surjective.  This formula proves uniqueness and both intertwining claims.
\end{proof}

The quotient is relative to the chosen operations and observations.
Enlarging either family can distinguish previously equivalent states and
therefore refine the quotient.  The theorem
concerns deterministic operations on a common invariant domain. Its
set-theoretic conclusion leaves the dimension and regularity of the
quotient to be determined from the particular system.

A coarse state therefore encodes future-response sufficiency.  Two states
are identified only when no allowed sequence can
bring their difference back to an observed output.  The quotient is thus
the smallest state description that remains sufficient for the declared
source interventions rather than a compression chosen solely from the
present configuration.

\subsection{Regular descent and its limits}
\label{sec:regular-descent}

For a fixed coarse map $C$ exact descent on its actual image has a simple
test.  No linear structure is needed for the equivalence
\begin{equation}
 CF=\bar F C\text{ for a unique }\bar F:C(\mathsf X)\to C(\mathsf X)
 \quad\Longleftrightarrow\quad
 Cx=Cy\ \Longrightarrow\ CF(x)=CF(y).
 \label{eq:fibre-consistency}
\end{equation}
Indeed the only possible definition is $\bar F(Cx)=CF(x)$ and the
right-hand condition is exactly its well-definedness.  The fixed maps $C$ and $F$ determine the descended map uniquely on the actual image $C(\mathsf X)$.

Fibre consistency makes the coarse output a function of the coarse input. Linear response requires differentiability of this function. The Banach-space statement fixes the regularity needed for this derivative.

\paragraph{Regular Banach descent.}
Let $\mathcal X,\mathcal Y,\mathcal J,\mathcal Z$ be Banach spaces and let $C:\mathcal X\to\mathcal Y$ be a bounded linear surjection. Let $U\subset\mathcal X$ and $V\subset\mathcal J$ be open and let $F:U\times V\to\mathcal X$ be continuously Fr\'echet differentiable.
Assume fibre consistency for every $h\in V$.  Then the unique map
$\bar F:C(U)\times V\to\mathcal Y$ is $C^1$ on the open set $C(U)\times V$ and
\begin{equation}
 \bar F(Cx,h)=CF(x,h),\qquad
 D_y\bar F(Cx,h)C=CD_xF(x,h),\qquad
 D_h\bar F(Cx,h)=CD_hF(x,h).
 \label{eq:banach-source-descent}
\end{equation}
A $C^1$ observation $O:U\times V\to\mathcal Z$ descends with the same
regularity if it is constant on each $C$-fibre at every $h$.

The quotient-norm proof is given in Appendix~\ref{app:kernel-proof}.
Bounded surjectivity makes $\Ker C$ closed and identifies
$\mathcal X/\Ker C$ with $\mathcal Y$ by a bounded isomorphism.
The quotient argument applies without a bounded linear section or a
complemented kernel.  Repeated iteration additionally requires the relevant
invariant domains.  A full-system
fixed point projects to a coarse fixed point. A coarse fixed point permits a nonzero displacement $F(x)-x$ in $\Ker C$ since it imposes only $CF(x)=Cx$.

A Gaussian heat map $\cH_\tau$ on $L^2(\mathbb R^d)$ has the strictly
positive Fourier multiplier $e^{-\tau|q|^2}$.  It is injective but its
multiplier tends to zero at large $|q|$.  Unit vectors supported at
successively higher frequencies therefore have images tending to zero. The inverse on the range is unbounded and that range is nonclosed for
$\tau>0$ in the inherited $L^2$ norm.  It is not a
surjection onto that $L^2$ space.  Giving its range the norm
$\|\cH_\tau f\|_{\rm tr}=\|f\|_2$ makes it a Banach isomorphic copy of
the original space rather than an information-losing quotient.  A genuine
projection or observation quotient must still satisfy
Eq.~\eqref{eq:fibre-consistency}.  Nonlinear descent can fail even after heat smoothing. On $[0,\pi]$ let $C=P_0$ be the
constant-mode projection and $F(u)=H_t(u^2)$ on the finite cosine core.
Then $C(a\cos nx)=C0=0$ whereas $CF(a\cos nx)=a^2/2\ne CF(0)$ for $a\ne0$ and $n\geq1$. Here $a$ is a function amplitude.

\subsection{Linear response of the sourced state equation}
\label{sec:same-rule-response}

An applied source can change both the state and the observation at fixed
state.  These two contributions explain the direct term in the response.
Likewise eliminating a state can change how the source enters and how the
remaining state is observed.  Differentiating the update and observation
together keeps all these contributions in the same calculation.

On a regular state space specify an actual update $F(X,h;z)$ and
observation $O(X,h;z)$ with finite-dimensional source and output spaces.
The parameter $z$ is an analytic response coordinate. $\eta$ denotes an
independent control parameter.  Neither is intrinsically a momentum or a
clock.  The symbols $B,S,R,D$ denote the four derivatives of this rule
at a declared fixed-point branch $X_*(z)$ and $h=0$
\begin{equation}
 B=D_XF,\qquad S=D_hF,\qquad R=D_XO,\qquad D=D_hO.
 \label{eq:native-four-derivatives}
\end{equation}
Here $D$ acts directly between the source and output spaces.  We
write $\mathscr C$ for the measured response and reserve $C$ in the general
descent results for the coarse map.  The QFT realisation writes
the measured matrix as $C(Q^2)$ and its kernel as $\cB$.

\paragraph{Same-rule response.}
Suppose $F,O$ are $C^1$ near the reference solution and $I-B$ is boundedly
invertible.  The local fixed-point branch is differentiable in $h$ and
\begin{equation}
 \mathscr C=D_h[O(X(h),h)]_{h=0}=D+R(I-B)^{-1}S.
 \label{eq:native-response}
\end{equation}
More generally if a $C^1$ residual equation $\mathcal E(X,h;z)=0$ defines
the branch and $L=D_X\mathcal E$ is boundedly invertible then
\begin{equation}
 \mathscr C=O_h-O_XL^{-1}\mathcal E_h.
 \label{eq:implicit-native-response}
\end{equation}
The Banach implicit-function theorem gives the local branch.
Differentiating its defining equation gives
$(I-B)X_h=S$ or $LX_h=-\mathcal E_h$.  Differentiating the observation
gives $\mathscr C=D+RX_h$ proving both formulae.

The linearised equation $\delta X=B\delta X+S\delta h$ describes the
source-induced change and its internal feedback.  When $\|B\|<1$
successive substitution converges in operator norm and gives
\begin{equation}
 \mathscr C=D+\sum_{n=0}^{\infty}RB^nS
            =D+RS+RBS+RB^2S+\cdots.
 \label{eq:feedback-path-expansion}
\end{equation}
Here $S$ injects the disturbance. Each $B$ transmits it through the coupled state and $R$ reads out the result.  The direct term $D$ records
the source dependence of the observation at fixed state.  In a stationary
problem these are repeated feedback contributions. They represent
successive evolution steps only when the update has that physical
meaning.  Away from the norm-convergent regime
Eq.~\eqref{eq:native-response} still defines the response wherever the
inverse exists.  Eliminating variables must preserve the source and
readout together with the internal feedback to preserve this sum.

For the residual convention put $S=-\mathcal E_h$.  The displayed formula
uses the residual derivatives and makes no choice of numerical iteration.
Convergence of a solver and stability of the actual update are thus
separate questions.  Their common source rule gives mutually consistent derivatives. Reciprocity and positivity are properties of the resulting response representation.

\begin{theorem}[Preservation of the complete visible response]
\label{thm:visible-descent}
Suppose $F,O,\bar F,\bar O$ are jointly $C^1$ in state and source and
$CF(X,h)=\bar F(CX,h)$ and $O(X,h)=\bar O(CX,h)$ hold in compatible open
neighbourhoods with fixed bounded linear $C$ and corresponding reference
branches.  If both $I-B$ and $I-\bar B$ are invertible then
\begin{equation}
 CB=\bar BC,\quad \bar S=CS,\quad R=\bar RC,\quad D=\bar D,
 \qquad \mathscr C=\bar{\mathscr C}.
 \label{eq:visible-descent}
\end{equation}
If both responses have meromorphic extensions to the same connected
domain from a common nonempty open invertible set their visible poles and
output Laurent coefficients agree there.
\end{theorem}
\begin{proof}
Differentiate the two neighbourhood identities.  Multiplying
$(I-\bar B)C=C(I-B)$ by the inverses gives
$(I-\bar B)^{-1}C=C(I-B)^{-1}$.  Substitution in
Eq.~\eqref{eq:native-response} gives equality. The identity theorem gives
the meromorphic assertion.
\end{proof}

Joint descent preserves the source-to-observation response. Its dependence on the readout is visible in a two-component example.  For $B(z)=\diag(0,1-z)$ and $S=(0,1)^{\mathsf T}$ with $R=(0,1)$ and $C=(1,0)$ the update descends to $\bar B=0$ but the full response is
$1/z$.  It cannot be recovered from the retained coordinate because the
readout does not descend.

\paragraph{State- and source-dependent coarse maps.}
Let $F,O,C=C(X,h)$ and the descended maps $\bar F,\bar O$ be jointly
$C^1$ in state and source on compatible open neighbourhoods of the
reference states and source with the following compositions defined.
Suppose
$C(F(X,h),h)=\bar F(C(X,h),h)$ and
$O(X,h)=\bar O(C(X,h),h)$ hold near an actual sourced fixed point.
Set $T=D_XC$ and $C_h=D_hC$ there.  Then
\begin{equation}
 \begin{aligned}
 TB&=\bar BT,& \bar S&=TS+(I-\bar B)C_h,\\
 R&=\bar RT,& D&=\bar D+\bar RC_h.
 \end{aligned}
 \label{eq:moving-coarse-derivatives}
\end{equation}
On the common invertible domain the full responses still agree.
At a fixed point the derivatives of $C$ on both sides are evaluated at
the same full state.  The state derivative gives the first identity. The source derivative gives $TS+C_h=\bar BC_h+\bar S$.
The observation chain rule gives the second line.  Inserting all four
identities into Eq.~\eqref{eq:native-response} cancels the $C_h$ terms.

These formulae concern fixed-source stationary response.  Time-varying
windows or sources require the actual adjacent-time coarse maps and their
history.  Discarding $C_h$ would already change the static response.

This also explains the distinction between static elimination and dynamical
reduction.  In a stationary calculation an unresolved variable can adjust
to the current source and appear as an effective stiffness.  During an
actual update it carries its own initial data and relaxation law. After it
is removed the retained variables remember the past through the induced
kernel.

\subsection{Gaussian operations and their representations}
Gaussian averaging is one of the allowed operations when its spatial
domain has been specified.  It is distinct from the sufficient quotient.
It changes the resolution by giving less weight to rapidly varying components.  Interaction-induced collective
feedback instead enters through the sourced state equation and its
derivative.  Keeping these roles explicit separates attenuation of an
existing response from a change in the dynamics that produces it.
The Euclidean Fourier representation separates external averaging from averaging of internal relative coordinates.
For $\tau>0$ define the normalised heat kernel and its action by
\begin{equation}
 g_\tau(x)=\frac{1}{(4\pi\tau)^{d/2}}
             e^{-|x|^2/(4\tau)},
 \qquad
 (\cH_\tau f)(x)=(g_\tau*f)(x).
 \label{eq:heat-kernel}
\end{equation}
We use $\widehat f(q)=\int e^{\iu q\cdot x}f(x)\,d^dx$ so that
$\widehat{\partial_\mu f}(q)=-\iu q_\mu\widehat f(q)$.  In this convention
\begin{equation}
 \widehat{\cH_\tau f}(q)=e^{-\tau|q|^2}\widehat f(q),
 \qquad
 \cH_{\tau_1}\cH_{\tau_2}
 =\cH_{\tau_1+\tau_2}.
 \label{eq:heat-multiplier}
\end{equation}
For a translationally invariant two-point function smearing each insertion
with $\cH_{\tau_Q/2}$ gives the external representation
\begin{equation}
 \cG^{\rm ext}_{\tau_Q}C(Q^2)
 :=e^{-\tau_QQ^2}C(Q^2).
 \label{eq:external-gaussian-representation}
\end{equation}
If a Gaussian is instead written as
$\exp[-|x|^2/(2\sigma^2)]$ its Fourier multiplier is
$\exp[-\sigma^2|q|^2/2]$ so that $\tau=\sigma^2/2$.
Equivalently the heat-kernel parameter used here is one half of the
one-coordinate variance.
This convention is stated once because one-insertion and two-insertion averages together with relative-coordinate averaging otherwise differ by factors of two.

Independent averaging steps add their displacements. Their laws therefore convolve.  Centring and
orthogonal invariance express the absence of drift and a preferred
direction. Finite nonzero variance sets a resolution length.  We require
the full displacement law to retain its shape when that length grows as
the square root of the scale parameter.  Within this composition class
the following proposition determines the Gaussian semigroup.

\begin{proposition}[Conditional uniqueness of the Gaussian semigroup]
\label{prop:gaussian-unique}
Let $\{\mu_\tau\}_{\tau\geq0}$ be a weakly continuous family of probability measures on $\mathbb R^d$ with $\mu_0=\delta_0$ and convolution law $\mu_{\tau+s}=\mu_\tau*\mu_s$. Assume the measures are centred and orthogonally invariant with finite nonzero second moment at $\tau=1$. Assume exact square-root self-similarity $(D_{\sqrt\tau})_\#\mu_1=\mu_\tau$.
Then there is a constant $D>0$ such that
\begin{equation}
 \widehat\mu_\tau(q)=e^{-D\tau|q|^2}.
 \label{eq:gaussian-unique}
\end{equation}
Thus the semigroup is Gaussian up to the choice of the scale unit $D$.
\end{proposition}

The self-contained proof in Appendix~\ref{app:kernel-proof} uses the
characteristic semigroup and its quadratic homogeneity.  For the stable-law
background see Chapter~3 of Ref.~\cite{Sato:1999}. A related Gaussian
derivation from different axioms appears in Section~2 of
Ref.~\cite{Guo:2026arithmetic}.
The conclusion belongs to this probabilistic Euclidean class.  An FRG
regulator is a separate scale-dependent modification of the inverse
propagator.

There is a complementary entropy characterisation at fixed covariance.
If a density $f$ has the same mean and positive covariance as the Gaussian
$g$ and the entropy integrals exist then
$\int f\log(f/g)\geq0$ and $\int f\log g=\int g\log g$ imply
$-\int f\log f\leq-\int g\log g$ with equality only for $f=g$ almost
everywhere.  This entropy characterisation concerns densities at fixed moments. Proposition~\ref{prop:gaussian-unique} instead characterises the whole semigroup through composition and exact self-similarity.

\begin{theorem}[Tensor lift and momentum-routing form]
\label{thm:gaussian-tensor-lift}
Let $F\in\mathcal S'((\mathbb R^d)^n)$ and apply $\cH_{a_i\tau}$ to its
$i$th coordinate where $a_i>0$.  In Fourier space the tensor-product action
is
\begin{equation}
 \widehat F_\tau(p_1,\ldots,p_n)
 =\exp\!\left[-\tau\sum_{i=1}^na_i|p_i|^2\right]
   \widehat F(p_1,\ldots,p_n).
 \label{eq:gaussian-tensor-lift}
\end{equation}
Suppose in addition that translation invariance has been used to factor out
the single overall momentum-conservation distribution.  Let $\mathsf E=(\mathbb R^d)^n$ and define $\Pi:\mathsf E\to\mathbb R^d$ by $\Pi(p_1,\ldots,p_n)=\sum_i p_i$ and set
$r=d(n-1)$.  Choose linear maps
\begin{equation}
 A:\mathbb R^d\to\mathsf E,
 \qquad L:\mathbb R^r\to\mathsf E,
 \qquad \Pi A=I_d,
 \qquad \operatorname{Im}L=\Ker\Pi,
 \label{eq:adapted-routing-coordinates}
\end{equation}
with $L$ injective.  Then
$T(P,\boldsymbol\xi)=AP+L\boldsymbol\xi$ is a linear coordinate
isomorphism from $\mathbb R^d\times\mathbb R^r$ onto $\mathsf E$ and
$L\boldsymbol\xi$ parametrises the entire momentum-conserving subspace.
Let $\widehat F_{\rm full}$ denote the coordinate representation of
$\widehat F$ under $T$ defined for test functions $\Phi$ by
\begin{equation}
 \langle\widehat F_{\rm full},\Phi\rangle
 :=|\det T|^{-1}
 \langle\widehat F,\Phi\circ T^{-1}\rangle.
 \label{eq:adapted-distribution-coordinate}
\end{equation}
In these adapted variables assume
\begin{equation}
 \widehat F_{\rm full}(P,\boldsymbol\xi)
 =(2\pi)^d\delta^{(d)}(P)\widehat F_{\rm red}(\boldsymbol\xi),
 \qquad \widehat F_{\rm red}\in\mathcal S'(\mathbb R^r),
 \label{eq:reduced-momentum-distribution}
\end{equation}
in the distributional sense that
\begin{equation}
 \langle\widehat F_{\rm full},\Phi\rangle
 =(2\pi)^d\langle\widehat F_{\rm red},\Phi(0,\cdot)\rangle
 \quad
 \text{for every }\Phi\in\mathcal S(\mathbb R^d\times\mathbb R^r),
 \label{eq:reduced-momentum-pairing}
\end{equation}
where $\delta^{(d)}$ is normalised by
$\int\delta^{(d)}(P)f(P)\,dP=f(0)$.  This convention includes the
Jacobian associated with the declared adapted coordinates. The equality then
fixes the normalisation of $\widehat F_{\rm red}$.  The induced action on the
reduced distribution is then
\begin{equation}
 \widehat F_{{\rm red},\tau}(\boldsymbol\xi)
 =\exp[-\tau\boldsymbol\xi^{\mathsf T}M_L\boldsymbol\xi]\,
  \widehat F_{\rm red}(\boldsymbol\xi),
 \qquad
 M_L=L^{\mathsf T}\diag(a_1I_d,\ldots,a_nI_d)L\succeq0.
 \label{eq:routed-gaussian-form}
\end{equation}
If $V$ is a routing subspace on which $M_L\succeq cI$ for some $c>0$ then
the multiplier is at most $e^{-\tau c q_0^2}$ on
$\{\boldsymbol\xi\in V:\|\boldsymbol\xi\|\geq q_0\}$.
Directions in $\Ker M_L$ are not suppressed by this statement.
\end{theorem}

The translation-invariant $n$-point convention removes the overall momentum delta before introducing independent routing variables. The resulting tensor-product distribution defines Eq.~\eqref{eq:routed-gaussian-form} through its test-function pairing.  Its test-function proof
including the normalisation in Eq.~\eqref{eq:reduced-momentum-pairing} is
given in Appendix~\ref{app:kernel-proof}.

For two equally weighted coordinates the routing can be displayed without
matrix notation.  With
\begin{equation}
 P_{\rm c}:=p_1+p_2,
 \qquad q_{\rm r}:=p_1-p_2,
 \qquad
 p_{1,2}=\frac12(P_{\rm c}\mathbin{\pm}q_{\rm r}),
 \label{eq:center-relative-routing}
\end{equation}
one has the exact identity
\begin{equation}
 |p_1|^2+|p_2|^2
 =\frac12\bigl(|P_{\rm c}|^2+|q_{\rm r}|^2\bigr).
 \label{eq:center-relative-gaussian-energy}
\end{equation}
Thus an equal-weight tensor Gaussian remains diagonal in centre and relative
momentum with multiplier
$\exp[-a\tau(|P_{\rm c}|^2+|q_{\rm r}|^2)/2]$.
For unequal weights the same substitution also produces the mixed term
$(a_1-a_2)P_{\rm c}\cdot q_{\rm r}/2$. It cannot be discarded unless the
weights or the kinematics remove it.

For a two-point function with $p_1=Q$ and $p_2=-Q$ and
$a_1=a_2=1/2$ Eq.~\eqref{eq:gaussian-tensor-lift} reduces to
Eq.~\eqref{eq:external-gaussian-representation}.  For an internal transfer
Eq.~\eqref{eq:routed-gaussian-form} generally contains a different quadratic
form and may contain mixed momentum terms.  We denote that separately
specified action by
\begin{equation}
 \cG^{\rm rel}_{\tau_q}A(\boldsymbol\xi)
 :=e^{-\tau_q\boldsymbol\xi^{\mathsf T}M_L\boldsymbol\xi}
 A(\boldsymbol\xi).
 \label{eq:relative-gaussian-representation}
\end{equation}
Equation~\eqref{eq:relative-gaussian-representation} for an internal loop
momentum follows from the declared coordinate action and routing.  These data
must be supplied separately from any external smearing of the composite
insertion.

\subsection{Gaussian momentum-transfer criterion}
\label{sec:gaussian-criterion}

Let $A(q)$ denote a projected contribution as a function or distribution of
a specified set of transferred momenta.  Once the coordinate action and the
routing have supplied a positive-semidefinite matrix $M$ set
\begin{equation}
 A_\tau(q):=\cG^{\rm rel}_\tau A(q)
 =e^{-\tau q^{\mathsf T}Mq}A(q).
 \label{eq:averaged-amplitude}
\end{equation}
This definition quantifies the Gaussian action on a specified contribution.
The corresponding graph and the scale dependence of the unfiltered amplitude
remain part of the effective dynamics.

\begin{theorem}[Fixed-transfer suppression]
\label{thm:fixed-transfer}
For $\delta>0$ let
$\Omega_\delta=\{q:q^{\mathsf T}Mq\geq\delta^2\}$.
The following statements hold.
\begin{enumerate}[label=(\roman*)]
 \item If $A\in L^\infty(\Omega_\delta)$ then
 \begin{equation}
  \|A_\tau\|_{L^\infty(\Omega_\delta)}
  \leq e^{-\tau\delta^2}
       \|A\|_{L^\infty(\Omega_\delta)}
  \xrightarrow[\tau\to\infty]{}0.
  \label{eq:uniform-suppression}
 \end{equation}
 \item If $A$ is a tempered distribution and
 $\chi\in C_c^\infty(\mathbb R^d)$ has support in $\Omega_\delta$ then
 \begin{equation}
   \langle A_\tau,\chi\rangle
   =\langle A,e^{-\tau q^{\mathsf T}Mq}\chi\rangle
   \xrightarrow[\tau\to\infty]{}0.
   \label{eq:distribution-suppression}
  \end{equation}
 \item Suppose a scale-dependent contribution has the form
 $A_\tau^{\rm full}=e^{-\tau q^{\mathsf T}Mq}\widetilde A_\tau$ and
 \begin{equation}
  \|\widetilde A_\tau\|_{L^\infty(\Omega_\delta)}
  \leq C e^{a\tau}
  \quad\hbox{with}\quad a<\delta^2.
  \label{eq:amplitude-growth-condition}
 \end{equation}
 Then $A_\tau^{\rm full}\to0$ in $L^\infty(\Omega_\delta)$.
\end{enumerate}
\end{theorem}

The proof is given in Appendix~\ref{app:kernel-proof}.
The $L^\infty$ bounds use essential suprema.  For bounded continuous
representatives they also give ordinary uniform bounds on the indicated
region.
Part~(ii) is the appropriate formulation for amplitudes containing threshold distributions or other non-function singularities.
Part~(iii) records the additional hypothesis required when the amplitude
itself runs with the averaging scale.  Without such a bound growth of
$\widetilde A_\tau$ can offset the Gaussian multiplier.
The convergence is not uniform in a shrinking neighbourhood of the origin.
Indeed for $q=u/\sqrt\tau$
\begin{equation}
 e^{-\tau q^{\mathsf T}Mq}=e^{-u^{\mathsf T}Mu},
 \label{eq:soft-scaling}
\end{equation}
so the soft scaling region together with $\Ker M$ remains in the problem.
Its integrated weight also depends on the volume of that region.  For
example if $M\succ0$ and $A$ is bounded and continuous at the origin
the substitution $q=u/\sqrt\tau$ gives
\begin{equation}
 \begin{aligned}
 I_\tau&:=\int_{\mathbb R^d}e^{-\tau q^{\mathsf T}Mq}A(q)\,d^dq\\
 &=\tau^{-d/2}\int_{\mathbb R^d}
       e^{-u^{\mathsf T}Mu}A(u/\sqrt\tau)\,d^du,\\
 \lim_{\tau\to\infty}\tau^{d/2}I_\tau
 &=\frac{\pi^{d/2}}{\sqrt{\det M}}A(0).
 \end{aligned}
 \label{eq:gaussian-soft-volume}
\end{equation}
The Gaussian supplies an integrable majorant so the last step follows
by dominated convergence.  A regular soft amplitude thus carries a
phase-volume factor $\tau^{-d/2}$ even though the multiplier equals one
at zero.  Singular amplitudes and undamped directions in $\Ker M$ retain
their own scaling and integration measures.  Integrating the damped
variables into an effective action can also leave local operators and derivative corrections together with running couplings in the retained variables.

The multiplier describes dephasing in the transferred momentum.  Contributions with a resolved
nonzero transfer acquire rapidly varying phases under spatial averaging
whereas back-to-back internal momenta with zero net transfer do not
dephase.  This is a kinematic retention statement.  The size and topology
of the contribution that survives in the soft region are still determined
by the response dynamics.

Momentum routing specifies which combinations of momenta are damped. Graph topology specifies how interactions connect the internal lines.
Bubbles and crossed ladders together with vertex corrections can all contribute in
the soft region.  Their relative weights follow from the interaction
vertices and propagators of the response equation.

\paragraph{Preservation of an isolated pole.}
Let a projected Euclidean form factor be meromorphic near $Q^2=0$ and let $w(Q^2)$ be analytic there with $w(0)=1$ and no zero in a neighbourhood of the origin.
Then multiplication by $w$ preserves both the order and residue of a possible simple pole at $Q^2=0$ and cannot create such a pole from a regular form factor.
In particular $w(Q^2)=e^{-\tau_Q Q^2}$ preserves
\begin{equation}
 Z_0=\lim_{Q^2\to0}Q^2 C(Q^2).
 \label{eq:pole-residue}
\end{equation}

This statement concerns $\cG^{\rm ext}_{\tau_Q}$ acting on a fixed
correlator.
Interactions integrated along an RG flow can change the correlator through the effective action. The derivative of the sourced effective action measures this dynamical change.
The positive spectral measure continues to be defined from the original
Wightman correlator.  In particular multiplication turns a
contact polynomial into an analytic non-polynomial short-range term.  The
physical spectral decomposition and its contact subtraction are therefore
defined before the diagnostic filter is applied.

\subsection{An additive-source realisation}
\label{sec:source-realisation}

The response operator should be identified from the linearisation of the
self-consistent dynamics rather than from the topology of a selected set of
graphs.  An additive-source specialisation of
Eq.~\eqref{eq:native-response} uses a Banach response space
$\mathcal X_{\cC}$ and maps $\mathsf S_X,\mathsf R_X$ of the types in
Eq.~\eqref{eq:native-four-derivatives}.  Let
$\mathfrak F_{\tau_q}(\,\cdot\,;z,\eta):\mathcal X_{\cC}\to
\mathcal X_{\cC}$ and consider near a reference solution the sourced
equation
\begin{equation}
 X=\mathfrak F_{\tau_q}(X;z,\eta)
   +\mathsf S_{X,\tau_Q}(z)J.
 \label{eq:sourced-self-consistency}
\end{equation}
Here $J\in\mathbb C^m$ while $\mathsf S_{X,\tau_Q}:\mathbb C^m\to
\mathcal X_{\cC}$ and $\mathsf R_{X,\tau_Q}:\mathcal X_{\cC}\to
\mathbb C^m$ may include external Gaussian factors.  The relative
scale $\tau_q$ if present belongs to the self-consistent map itself.  The
analytic term $C_{\rm reg}$ is understood in the same external convention as
the source and readout.

For a local solution branch $X(J;z,\eta)$ its measured output is a
map into the external operator coordinates
\begin{equation}
 \begin{aligned}
  \mathcal Y(J;z,\eta)
  &=\mathcal Y_{\rm reg}(J;z,\eta)\\
  &\quad+\mathsf R_{X,\tau_Q}(z)X(J;z,\eta),\\
  C_{\tau_Q,\tau_q}&:=D_J\mathcal Y\big|_{J=0},
  &C_{\rm reg}&:=D_J\mathcal Y_{\rm reg}\big|_{J=0}.
 \end{aligned}
 \label{eq:measured-response-output}
\end{equation}
Thus the correlator or susceptibility is the Jacobian from the declared
physical source coordinates to the declared readout coordinates.  A critical direction contributes when the source excites it and the readout detects it.
When the sources couple to physical insertions and the output is their
normalised expectation value differentiation gives the connected two-point
matrix with the specified source-pairing convention and local contact
terms.  The source perturbs the state equation. The state adjusts and the output records that adjustment together with its direct source dependence.

\begin{proposition}[Response factorisation]
\label{prop:response-factorisation}
Suppose $\mathfrak F_{\tau_q}$ is continuously Fr\'echet differentiable in
$X$ on a neighbourhood of a solution $X_*$ of
Eq.~\eqref{eq:sourced-self-consistency} with $J=0$ and suppose
$\mathcal Y_{\rm reg}$ is Fr\'echet differentiable in $J$ at the origin.  Set
\begin{equation}
 \cB_{\tau_q}(z;\eta)
 :=D_X\mathfrak F_{\tau_q}(X_*;z,\eta).
 \label{eq:response-jacobian}
\end{equation}
If $I-\cB_{\tau_q}$ is invertible the Banach-space implicit-function theorem
gives a differentiable local solution branch and its measured response has
the form
\begin{equation}
 C_{\tau_Q,\tau_q}(z;\eta)
 =C_{\rm reg}(z;\eta)
 +\mathsf R_{X,\tau_Q}(z)
  [I-\cB_{\tau_q}(z;\eta)]^{-1}
  \mathsf S_{X,\tau_Q}(z).
 \label{eq:coarse-grained-response-factorisation}
\end{equation}
At a critical point the same identity holds as a meromorphic continuation
on a domain satisfying the analytic Fredholm hypotheses.
\end{proposition}

Thus an external Gaussian changes the source and readout maps but not the
spectrum of the fixed response operator in Eq.~\eqref{eq:response-jacobian}.
For a source--source matrix with the same scalar filter on both insertions
one may write on the Euclidean axis
\begin{equation}
 \begin{aligned}
  \mathsf S_{X,\tau_Q}(z)&=\mathsf S_{X,0}(z)W_{\tau_Q}(z),\\
  \mathsf R_{X,\tau_Q}(z)&=W_{\tau_Q}(z)\mathsf R_{X,0}(z),\\
  C_{{\rm reg},\tau_Q}(z)
   &=W_{\tau_Q}(z)C_{{\rm reg},0}(z)W_{\tau_Q}(z),\\
  W_{\tau_Q}(z)&=e^{-\tau_Q z/2}I_m.
 \end{aligned}
 \label{eq:external-source-readout-filter}
\end{equation}
Since $W_{\tau_Q}(0)=I_m$ this factorisation also makes the preservation of
an existing residue manifest.
Internal averaging affects that spectrum through a specified
$\tau_q$-dependence of $\mathfrak F_{\tau_q}$.  This dependence must be
derived from the declared Gaussian action.

%% file: sections-en/sec03_operator_classification.tex
\section{Channels, elimination, and visible critical structure}
\label{sec:classification}

The sourced state equation of Section~\ref{sec:criteria} supplies a
response space and the four maps $B,S,R,D$.  A channel is a
specified source-to-observation problem on this space.  A bounded projector
$P_X$ selects a retained subspace and $Q_X=I_X-P_X$ its
complement.  The retained sector is closed only when the update does not
send it into the complement.  We use $\cB=B$ and $\mathsf S_X=S$ together with $\mathsf R_X=R$ and $C_{\rm reg}=D$ to keep the notation consistent with the QFT realisation.

\subsection{Reference-state symmetry and exact elimination}
\label{sec:general-channels}

The relevant symmetry is the symmetry of the reference state as well as
of the update.  A transformation that moves the reference state relates
two different linear-response problems. It need not split either one
into invariant channels.

\paragraph{Equivariant linearisation at the reference state.}
Let a group act by bounded linear isomorphisms $U_g$ on the state space
and suppose $F(U_gX)=U_gF(X)$.  At $X_*$
\begin{equation}
 DF(U_gX_*)U_g=U_gDF(X_*).
 \label{eq:equivariant-background}
\end{equation}
Consequently $DF(X_*)$ commutes with $U_g$ for the stabilizer
$G_{X_*}=\{g:U_gX_*=X_*\}$.  The same fixed-background restriction applies
to source and observation intertwiners when their equivariance holds.
Differentiate the equivariance identity with respect to $X$.  For
$g\in G_{X_*}$ the two Jacobians have the same base point.  Differentiating
the sourced equivariance and observation identities gives their
intertwining relations.  A broken symmetry instead relates different
backgrounds and need not block-diagonalise one of their Jacobians.

Symmetry provides invariant sectors only after these fixed-background
conditions and the relevant bounded projectors have been established.
The source and readout maps determine which of these sectors is observable.
Exact finite systems and finite approximations to infinite systems both
admit block elimination but only the latter need cutoff convergence to
identify an infinite target.

\begin{proposition}[Response-closed channel]
\label{prop:response-closed-channel}
The spectrum of the restriction $\cB|_{\Ran P_X}$ is represented by
$P_X\cB P_X$ only if
\begin{equation}
 Q_X\cB P_X=0.
 \label{eq:response-channel-invariance}
\end{equation}
If Eq.~\eqref{eq:response-channel-invariance} fails and
$I_X-Q_X\cB Q_X$ is invertible in the domain under consideration elimination of
the complementary sector gives the exact effective response operator
\begin{equation}
 \cB_{\rm eff}
 =P_X\cB P_X
  +P_X\cB Q_X(I_X-Q_X\cB Q_X)^{-1}Q_X\cB P_X.
 \label{eq:feshbach-response}
\end{equation}
 Writing $D_Q:=I_X-Q_X\cB Q_X$ the associated source and readout maps are
\begin{align}
 \mathsf S_{X,{\rm eff}}
 &=P_X\mathsf S_X+P_X\cB Q_XD_Q^{-1}Q_X\mathsf S_X,
 \label{eq:feshbach-source}\\
 \mathsf R_{X,{\rm eff}}
 &=\mathsf R_XP_X+\mathsf R_XQ_XD_Q^{-1}Q_X\cB P_X.
 \label{eq:feshbach-readout}
\end{align}
If the unreduced response contains the analytic term $C_{\rm reg}$ of
Eq.~\eqref{eq:coarse-grained-response-factorisation} its reduced value is
\begin{equation}
 C_{{\rm reg},{\rm eff}}
 =C_{\rm reg}+\mathsf R_XQ_XD_Q^{-1}Q_X\mathsf S_X.
 \label{eq:feshbach-direct-response}
\end{equation}
The full response is therefore
$C_{{\rm reg},{\rm eff}}
 +\mathsf R_{X,{\rm eff}}(I_X-\cB_{\rm eff})^{-1}
  \mathsf S_{X,{\rm eff}}$.
Consequently the spectrum of $P_X\cB P_X$ alone has no channel interpretation
when the omitted sector remains coupled.
If $\cB(z)$ is a compact analytic family and $D_Q(z)^{-1}$ exists throughout
a neighbourhood then $\cB_{\rm eff}(z)$ is a compact analytic family on
$\Ran P_X$.  On a connected neighbourhood where $I_X-\cB_{\rm eff}$
is invertible at one point analytic Fredholm theory gives a meromorphic
inverse.  With analytic source and readout maps together with the background term and no
continuum threshold in that neighbourhood the effective response is
meromorphic as well.
\end{proposition}

\begin{proof}
Write $I-\cB$ in $P\oplus Q$ block form and solve the $Q$ equation before the
$P$ equation.  This gives Eqs.~\eqref{eq:feshbach-response}--
\eqref{eq:feshbach-direct-response}.  Compactness of $\cB_{\rm eff}$ follows
because bounded left and right multiplication preserve compactness while
analyticity follows from analyticity of the inverse map on the open set of
bounded invertible operators.
\end{proof}

This is the Feshbach--Schur block elimination adapted to response maps between
the stated spaces~\cite{Feshbach:1958}.  A channel that fails
Eq.~\eqref{eq:response-channel-invariance} can be enlarged to include
the coupled states or described by the energy-dependent operator
Eq.~\eqref{eq:feshbach-response}.
The second term accounts for excursions into the eliminated sector and
their return to the retained one. The source and readout maps together with the direct term account for the other ways that sector contributes to the observation.
Thus a retained channel can acquire collective feedback through other
channels even when its restricted update is noncritical.  Whether the
feedback reaches criticality is determined by the full effective
operator including these excursions.

If the source lies in the retained sector and the readout annihilates
its complement then $Q_X\mathsf S_X=0$ and $\mathsf R_XQ_X=0$.
The general block formulae then reduce to
\begin{equation}
 \mathsf S_{X,{\rm eff}}=\mathsf S_X,
 \qquad
 \mathsf R_{X,{\rm eff}}=\mathsf R_X,
 \qquad
 C_{{\rm reg},{\rm eff}}=C_{\rm reg}.
 \label{eq:feshbach-compatible-simplification}
\end{equation}
These simplifications require both compatibility identities.  A general
source and observation can couple to $Q_X$ in which case all four
unsimplified terms must be retained.

For differentiation with respect to a parameter $a$ the inverse identity is
$\partial_aD_Q^{-1}=D_Q^{-1}Q_X(\partial_a\cB)Q_XD_Q^{-1}$ for fixed
projectors.  The product rule then contributes derivatives of both off-diagonal blocks and of the source and readout maps together with the direct term.  Derivatives along a reference
branch include its state displacement.  A moving projector additionally
contributes its own derivatives.  Appendix~\ref{app:kernel-proof} records
the fixed-projector formula. Frequency dependence generated by this
elimination carries the delayed influence of complementary states.

An eliminated dynamical variable with independent initial data generally
leaves both memory and an initial-state contribution. A local auxiliary
field instead obeys an algebraic equation. If its derivative with respect
to that field is invertible the implicit-function theorem determines
the field locally from the retained fields and their derivatives together with the sources.
Substitution into the common quadratic action gives the corresponding
Schur correction. In the auxiliary-radius realisation of
ref.~\cite{Guo2026CompactSpectrumII} Appendix~G.6 this operation is
performed along a coupled matter--amplitude FLRW solution. The radius is
first eliminated from the common action and the lapse and shift
constraints are then removed from its quadratic variation. The resulting
physical kinetic and gradient matrices are positive on a certified finite
interval. This example supplies an actual algebraic complementary block. The time-dependent background equations provide the propagation and
constraint conditions used to interpret its reduced response.

\subsection{Isolated Fredholm response}
\label{sec:kernel-criterion}

Near an isolated critical mode an internally transmitted disturbance
nearly reproduces itself.  An eigenvalue approaching one measures this
feedback in the linearised equation.  For it to appear in an observation
the source must excite the corresponding direction and the readout must
detect it.  Feedback criticality and visibility enter the residue through separate factors. The reciprocal and spectral realisations then determine the sign and
physical interpretation of the residue.

The operator $I-\cB$ is the restoring operator of the coupled channel
and its inverse is the corresponding susceptibility.  When an eigenvalue of $\cB$ reaches one a collective direction reproduces its disturbance
after a complete feedback cycle so the restoring coefficient along that
direction vanishes.  The condition $\partial_z\lambda\ne0$ specifies how
this softening is crossed in the declared response coordinate. The source
and readout overlaps then decide whether the softened direction is visible.

\begin{theorem}[Kernel criterion for an isolated algebraic pole]
\label{thm:kernel-pole}
Fix $\eta=\eta_c$ and let $U$ be a connected neighbourhood of $z=0$.
Suppose that $\cB(z;\eta_c)$ is a compact analytic operator family on
$U$ that $I-\cB(z;\eta_c)$ is invertible at least at one point of $U$
and that $U$ contains no continuum threshold or other non-meromorphic
singularity.  Suppose also that the response-space source $\mathsf S_X(z)$ and readout $\mathsf R_X(z)$ together with the matrix background
$C_{\rm reg}(z)$ in Eq.~\eqref{eq:coarse-grained-response-factorisation}
are analytic on $U$.  Suppress the fixed parameter $\eta_c$ and consider
that matrix response.  Then
\begin{enumerate}[label=(\roman*)]
 \item if $1\notin\spec\cB(0)$ the response has no pole at $z=0$ generated by this kernel.
 \item if $\cB(z)$ has an isolated algebraically simple eigenvalue
 $\lambda(z)$ with
 \begin{equation}
  \lambda(0)=1,
  \qquad \partial_z\lambda(0)\neq0,
  \label{eq:simple-spectral-zero}
 \end{equation}
 then the resolvent has a simple pole.  The measured response contains that
 pole precisely when both
 $\mathsf R_X(0)|r\rangle\neq0$ and
 $\langle\ell|\mathsf S_X(0)\neq0$.
 \item with $\langle\ell|r\rangle=1$ its algebraic residue matrix is
 \begin{equation}
  Z_{\rm alg}
  =-\frac{\mathsf R_X(0)|r\rangle
            \langle\ell|\mathsf S_X(0)}
           {\partial_z\lambda(0)}.
  \label{eq:matrix-algebraic-residue}
  \end{equation}
 The algebraic residue includes the sign of the crossing and both overlaps. Matching the response to a positive physical spectral representation gives its interpretation as a Lorentzian mass atom.
\end{enumerate}
\end{theorem}

The proof in Appendix~\ref{app:kernel-proof} is an application of the analytic Fredholm theorem and the Riesz projection onto the critical eigenspace.
On the simple branch choose local analytic eigenvectors with
$\langle\ell(z)|r(z)\rangle=1$.  The singular contribution is
\[
 \frac{\mathsf R_X(z)|r(z)\rangle
       \langle\ell(z)|\mathsf S_X(z)}{1-\lambda(z)}.
\]
The two numerator factors measure excitation and detection of the
collective direction.  Spectral isolation permits the response to be
resolved into this collective mode and an analytic complementary part.
The bounded complementary inverse controls the response of the remaining
directions while the critical mode becomes soft.  Expanding
$1-\lambda(z)=-\lambda'(0)z+O(z^2)$ gives the residue in
Eq.~\eqref{eq:matrix-algebraic-residue}.  Its magnitude therefore depends
both on the coupling to the mode and on how rapidly its inverse response
vanishes.  This is a local statement in $z$. Its interpretation as a
frequency or momentum singularity uses the specified physical realisation.
Condition~\eqref{eq:simple-spectral-zero} is transversality in the declared
complex response coordinate $z$.  For a real eigenvalue branch along a
real control parameter $\partial_\eta\lambda(0;\eta_c)\neq0$ certifies
a first-order transverse crossing of one.  Higher-order control crossings are possible. For example $\lambda(z;\eta)=1-z+\eta^3$ crosses at $\eta=0$ despite a
vanishing first control derivative.  Higher-order crossings require their
own analysis and do not change the simple-pole criterion in $z$.
When $\eta=k$ is an RG
scale a crossing at intermediate $k$ is a flow diagnostic. A physical-state
claim requires the pole conditions for the exact endpoint response.
If a continuum threshold starts at $z=0$ the meromorphic-neighbourhood
hypothesis of Theorem~\ref{thm:kernel-pole} fails in general.
A limiting-absorption estimate or direct spectral analysis treats this threshold.

\subsection{Reciprocal response and a finite-rank positive residue}
\label{sec:native-positivity}

The Fredholm result accommodates nonnormal kernels and independent source
and readout maps.  A reciprocal Hilbert-space realisation permits an
additional positive statement before any Lorentzian representation.
The critical subspace determines the singular part. The invertible complement contributes analytic terms while the derivative of the
response operator fixes the sign of the residue.  Source overlap then
determines which critical directions are visible.

\begin{theorem}[Positive residue for reciprocal response]
\label{thm:native-positive-residue}
Let $L(z)$ be analytic in operator norm near zero on a Hilbert space
$\mathcal H$.  Suppose $L(0)$ is self-adjoint with a nonzero finite-dimensional kernel and orthogonal projector $P$. Suppose
$QL(0)Q|_{Q\mathcal H}$ has a bounded inverse where $Q=I-P$.
If
\begin{equation}
 G_0=PL'(0)P|_{P\mathcal H}\succ0,
 \label{eq:native-crossing-matrix}
\end{equation}
then on a sufficiently small punctured neighbourhood
\begin{equation}
 L(z)^{-1}=\frac{PG_0^{-1}P}{z}+A(z),\qquad A\text{ analytic at }0.
 \label{eq:native-inverse-residue}
\end{equation}
Let $S(z):\mathbb C^m\to\mathcal H$ and $R(z):\mathcal H\to\mathbb C^m$ together with $D(z)$ be analytic with $R(0)=S(0)^*$ in the declared source--output
pairing.  The response $D+RL^{-1}S$ has residue
\begin{equation}
 Z_{\rm native}=S(0)^*PG_0^{-1}PS(0)\succeq0,
 \qquad \operatorname{rank}Z_{\rm native}
       =\operatorname{rank}(PS(0)).
 \label{eq:native-positive-residue}
\end{equation}
A pole is visible precisely when $PS(0)\ne0$.
\end{theorem}

\begin{proof}
Self-adjointness makes both off-diagonal blocks of $L(0)$ zero.
The inverse of the $Q$ block persists analytically.  The Schur complement
on $P\mathcal H$ is $zG_0+O(z^2)$ since the product of the off-diagonal
blocks is $O(z^2)$.  Its inverse is $z^{-1}G_0^{-1}$ plus an analytic
remainder.  Block inversion proves Eq.~\eqref{eq:native-inverse-residue}. The off-diagonal factors cancel the possible $z^{-1}$ singularity outside
the $P$ block.  Taking the residue gives
Eq.~\eqref{eq:native-positive-residue}.  Its factorisation as
$(G_0^{-1/2}PS(0))^*(G_0^{-1/2}PS(0))$ proves positivity and rank.
\end{proof}

For a source direction $j$ the residue obeys
$j^*Z_{\rm native}j=\|G_0^{-1/2}PS(0)j\|^2$.
It is positive exactly when that source excites the critical subspace.
Reciprocity makes readout the adjoint of excitation in the same pairing. The two mode overlaps therefore combine into a squared norm.  The positive
crossing matrix supplies the weight with which that norm enters the pole.
Dynamical stability is determined separately by the actual evolution. The sign of an interaction energy also depends on how physical sources
couple to this response.

The theorem applies to the specified Hilbert pairing with reciprocity and an isolated critical subspace whose crossing matrix is positive.  For $L=I-\cB$ its crossing matrix is $-P\cB'(0)P$ in the response variable $z$.  The complementary inverse controls spectral accumulation. For example $\diag(0,1,1/2,\ldots)$ has no bounded complementary inverse and adding
$zI$ produces singularities accumulating at zero.  For unbounded operators the result applies after a justified common-domain or closed-form analytic reduction. Nonreciprocal updates retain independent source and readout maps.

For a finite-dimensional real affine response $L(z)=H+zK$ under the
hypotheses of the theorem with $H=H^{\mathsf T}$ and reciprocal source
pairing a basis of the critical space gives a direct residue formula.
For rational matrices this basis can be obtained by exact elimination.
If $N$ is a full kernel basis of $H$
and $N^{\mathsf T}KN\succ0$ the residue is
\begin{equation}
 Z=S^{\mathsf T}N(N^{\mathsf T}KN)^{-1}N^{\mathsf T}S.
 \label{eq:native-rational-kernel-residue}
\end{equation}
To see this without orthogonalising $N$ complete it to an invertible
basis and apply the same Schur argument by congruence. This expression
is invariant under $N\mapsto NT$ for invertible $T$ and has rank
$\operatorname{rank}(N^{\mathsf T}S)$. The software of
Ref.~\cite{Guo2026CoarseGrainingFramework} provides supplementary rational finite-matrix
evaluations of this formula.  A small positive eigenvalue instead
gives a regular response at zero since
$\lim_{z\downarrow0}z/(\epsilon+z)=0$ for every $\epsilon>0$.
The residue here refers to the specified analytic parameter $z$.
An interpretation as a spatial propagation pole or a vacuum spectral
atom additionally requires a representation of that parameter and of the
same source-contracted response.  Iteration stability is determined by
the update spectrum independently of the sign of this residue.

\subsection{Quantum-field-theory realisation}
\label{sec:channel-data}
For the local-operator realisation work with a local unitary
Poincar\'e-invariant quantum field theory admitting a Euclidean
continuation for the physical correlation functions considered below.
Definition~\ref{def:physical-spectral-realisation} also specifies the
directly reduced helicity realisation whose physical Hilbert space and
source covariance enter the same response factorisation. The
invariant-mass spectral assertions use a
Poincar\'e-invariant positive-energy vacuum unique in the chosen physical
representation as specified in Appendix~\ref{app:kl}.  Whenever a response
is identified with a physical two-point function its insertions are centred
$\Theta_i=O_i-\langle O_i\rangle I$ and its correlation matrix is connected.
Subtracting these one-point products precedes and is distinct from the
subtraction of local contact terms.
Gauge theories are understood in a renormalised BRST formulation without an uncancelled gauge anomaly.
Let $\mathcal V_{\rm phys}$ be the renormalised operator space after passage
to the physical quotient described in Appendix~\ref{app:operator-mixing}.
 The quotient is algebraic.  There is no canonical Hilbert norm on a space of
 local composite operators so every infinite-dimensional boundedness
 statement below is made relative to a specified channel norm for which
$\mathcal V_{\rm phys}$ has been completed to a Banach space and all stated
maps extend continuously. We use the same symbol for that completion.  All boundedness statements refer to this channel norm.  A bounded projector
$\mathcal P_{\cC}$ selects a Lorentz and internal-symmetry sector with the
closed Banach subspace $\mathcal H_{\cC}=\Ran\mathcal P_{\cC}$.  For a finite
operator block the choice of norm is immaterial because all norms are
equivalent.  A finite block approximates this exact space through specified projection and comparison maps.

The operator space specifies which insertions are probed. The response
space carries the variables used to solve their sourced equations.  The
realisation maps connect these descriptions while the regulator and
truncation data record how a concrete calculation is performed.

\begin{definition}[Physical channel]
\label{def:physical-channel}
A physical operator channel is the tuple
\begin{equation}
 \cC_{\rm phys}
 =\bigl(\mathsf T,|\Omega\rangle,\mathcal V_{\rm phys},
         \mathcal P_{\cC},\mathcal S,\mathcal E,
         \mathsf S_{\cO},\mathsf R_{\cO}\bigr),
 \label{eq:physical-channel}
\end{equation}
Here $\mathsf T$ specifies the theory and $|\Omega\rangle$ the state or vacuum. The data $\mathcal S$ specify the exact symmetries and anomalies and $\mathcal E$ the external kinematics.
For a general state the response variables must respect its unbroken
symmetries. The vacuum identification $z=Q^2$ and the invariant-mass atom
test do not extend to an ensemble without a separate spectral argument.
The bounded maps
\begin{equation}
 \mathsf S_{\cO}(z):\mathbb C^m\longrightarrow\mathcal H_{\cC},
 \qquad
 \mathsf R_{\cO}(z):\mathcal H_{\cC}\longrightarrow\mathbb C^m
 \label{eq:operator-source-readout-maps}
\end{equation}
specify the physical operator insertions and their output coordinates.
To use these insertions in a response equation one must also specify how they
are realised on the response space.
Let $\mathcal X_{\cC}$ be a Banach response space and let $P_X\in\mathcal L(\mathcal X_{\cC})$ be a bounded projector and let
\begin{equation}
 \iota_{\cC}(z):\mathcal H_{\cC}\longrightarrow\mathcal X_{\cC},
 \qquad
 \pi_{\cC}(z):\mathcal X_{\cC}\longrightarrow\mathcal H_{\cC}
 \label{eq:channel-realisation-maps}
\end{equation}
be bounded source-realisation and readout-realisation maps.  They need not be
mutual inverses.  Compatibility with the selected response block means
\begin{equation}
 P_X\iota_{\cC}=\iota_{\cC},
 \qquad
 \pi_{\cC}P_X=\pi_{\cC}.
 \label{eq:channel-realisation-compatibility}
\end{equation}
The maps that enter a response equation are then
\begin{equation}
 \mathsf S_X:=\iota_{\cC}\mathsf S_{\cO}:
 \mathbb C^m\longrightarrow\mathcal X_{\cC},
 \qquad
 \mathsf R_X:=\mathsf R_{\cO}\pi_{\cC}:
 \mathcal X_{\cC}\longrightarrow\mathbb C^m.
 \label{eq:response-source-readout}
\end{equation}
Their construction specifies how the physical model excites and observes the response state.  In particular a critical response eigenspace can fail
to couple to either map.

The operator part of the tuple is an exact channel when
$\Ran\mathcal P_{\cC}$ is stable under the renormalisation map and the
unbroken symmetries recorded in $\mathcal S$.  A finite-dimensional
projected space $P_N\mathcal V_{\rm phys}$ is a computational block.
It is exact only when closure under renormalisation and symmetry together with dynamical response
has been proved. Otherwise conclusions drawn from it carry the truncation
label.  The classification applies to the full block and its correlation matrix.
Here boundedness is taken in the Banach completion and channel norm declared
as part of the response data below.  The algebraic BRST/EOM quotient precedes
this completion. Choosing a completion does not alter which representatives
define physical cohomology classes.

For later reference the response and approximation data are recorded as
\begin{equation}
 \cI=\bigl(\mathcal X_{\cC},P_X,\iota_{\cC},\pi_{\cC},
            \cH_\tau,\cG^{\rm ext}_{\tau_Q},
            \cG^{\rm rel}_{\tau_q},\mathfrak F_{\tau_q},\cB(z;\eta),
            \mathcal R_{\rm ren},\mathcal T_N\bigr),
 \label{eq:response-data}
\end{equation}
Here $\cH_\tau$ is the one-coordinate heat semigroup. The maps $\cG^{\rm ext}$ and $\cG^{\rm rel}$ are its induced representations from Section~\ref{sec:criteria}. The map $\mathfrak F_{\tau_q}$ specifies the self-consistent response and $\cB$ its linearisation. The data $\mathcal R_{\rm ren}$ fix renormalisation and the regulator. The truncation $\mathcal T_N$ approximates an infinite target with $N=\infty$ denoting the exact hierarchy.  A finite system defined by an exact
operation protocol has no such limiting requirement.
\end{definition}

The spectral and response variables have distinct roles.
The invariant $s\geq0$ labels Lorentzian spectral mass.
The variable $Q^2\geq0$ is the squared external Euclidean momentum of a projected source--source correlator.
The scale $k$ labels a Wilsonian or functional RG trajectory.
Within a diagram or kernel $q\in\mathbb R^d$ denotes the Euclidean momentum transferred between a fast subgraph and the external channel.
In this QFT realisation $z$ is the complexified external invariant with $z=Q^2$ on the Euclidean axis and $z=-p^2-\iu0$ at the Lorentzian physical
boundary.  Finally $\eta$ denotes a control parameter such as a coupling or background or RG scale.  A derivative with respect to $z$ is therefore not
a flow derivative with respect to $\eta$ or $k$.

\subsection{Operator blocks and the physical quotient}
\label{sec:operator-blocks}

Begin with the renormalised local operators of fixed Lorentz representation with specified internal quantum numbers and ghost number.  Their exact completion can be
infinite dimensional.  A finite set selected by a canonical-dimension bound
or a numerical basis is denoted $\mathcal V_N$ and is treated as a
truncation unless a separate closure argument applies.
Renormalisation takes the matrix form
\begin{equation}
 [\cO_i]_R(\mu)=Z_{ij}(\mu,\mathcal R)\cO_j^{\rm bare}.
 \label{eq:operator-renormalisation}
\end{equation}
In a gauge theory this set may contain gauge-invariant representatives and BRST-exact operators together with equation-of-motion operators and total derivatives.
The triangular renormalisation system retains these directions before
passage to physical matrix elements~\cite{JoglekarLee:1976}.

Let $\overline{\mathcal V}=\mathcal V/\mathcal I_{\rm EOM}$ be the on-shell local operator space and assume that the renormalised BRST differential $\mathsf s_{\rm B}$ descends to it.
For physical separated-point matrix elements the relevant ghost-number-zero cohomology is
\begin{equation}
 \mathcal V_{\rm phys}
 =H^0(\mathsf s_{\rm B},\overline{\mathcal V})
 =\frac{\Ker\!\left(\mathsf s_{\rm B}:\overline{\mathcal V}^{0}
       \longrightarrow\overline{\mathcal V}^{1}\right)}
 {\Ran\!\left(\mathsf s_{\rm B}:\overline{\mathcal V}^{\,-1}
       \longrightarrow\overline{\mathcal V}^{0}\right)},
 \label{eq:physical-quotient}
\end{equation}
with contact terms retained separately.  The cohomological treatment of
local operator renormalisation is developed in Section~8.6 of
Ref.~\cite{Barnich:2000}.  The on-shell gluonic matrix elements at zero
momentum transfer studied in Ref.~\cite{CollinsScalise:1994} illustrate
the infrared care needed when taking exceptional limits.
Here $\mathcal I_{\rm EOM}$ is represented quantum mechanically by the
renormalised field-redefinition insertions
$\mathcal E_F^R$ of Eq.~\eqref{eq:quantum-eom-insertion} including the
measure divergence and operator-mixing counterterms.
Total derivatives require separate treatment.  They are redundant for a
constant-source integrated insertion or at zero insertion momentum only
when the boundary term vanishes and the indicated limit is well defined.
For a nonconstant source integration by parts retains a source-gradient
insertion. At nonzero insertion momentum the descendants must remain in
the mixing block.
Equation~\eqref{eq:physical-quotient} is the local cohomological construction reviewed in Appendix~\ref{app:operator-mixing}. Anomaly insertions and descendants together with contact contributions remain in the appropriate mixing and source terms.

For non-gauge theories the BRST quotient is absent but operator mixing and EOM redundancies remain.
For the scalar improvement
$T_{\mu\nu}\mapsto T_{\mu\nu}
+(\partial_\mu\partial_\nu-\eta_{\mu\nu}\Box)X$
with $X$ a renormalised Lorentz scalar the separated-point TT projection
is unchanged whereas the spin-zero response and the poles visible in an
individual representative can change.  If the improvement descendant is
retained in a closed operator block a finite nonsingular constant basis
change preserves the pole structure of the full correlation matrix and the rank of each isolated residue. The full matrix and its residue transform
by $C\mapsto MCM^\dagger$ with $M$ constant and nonsingular.

This definition accommodates mixing between scalar operators such as mass terms and interaction densities together with total derivatives and between tensor operators with the same physical spin.
It also prevents unlike objects from being entered into a common table.
A running coupling is a coordinate on theory space. A local composite operator is an insertion.
They are related through source derivatives and beta functions but are not the same classification object.

Renormalisation closure does not by itself imply closure under the dynamical
response.  The operator projector $\mathcal P_{\cC}$ acts on
$\mathcal V_{\rm phys}$ whereas the compatible response projector $P_X$
from Eq.~\eqref{eq:channel-realisation-compatibility} acts on
$\mathcal X_{\cC}$.  Put $Q_X=I_X-P_X$. The two projectors are related by the
realisation maps but are not identified.

\subsection{Symmetry-allowance criterion}
\label{sec:symmetry-criterion}

The original QED Ward identity and its Ward--Takahashi generalisation
relate the fermion propagator and current vertex~\cite{Ward:1950,Takahashi:1957}.
More generally a symmetry identity relates current insertions to
variations of the other fields.
For a current $J_A^\mu$ and renormalised insertions $\Phi_r$ the
position-space identity allowing for an anomaly has the schematic form
\begin{align}
 \partial_\mu^x
 \langle T J_A^\mu(x)\Phi_1(x_1)\cdots\Phi_n(x_n)\rangle
 ={}&-\iu\sum_{r=1}^n\delta^{(d)}(x-x_r)
 \langle T\Phi_1\cdots(\delta_A\Phi_r)\cdots\Phi_n\rangle
 \nonumber\\
 &+\langle T\mathcal A_A(x)\Phi_1\cdots\Phi_n\rangle,
 \label{eq:ward-contact}
\end{align}
where $\mathcal A_A$ is an anomaly insertion and the displayed delta functions are contact terms.
Translation Ward identities for $T_{\mu\nu}$ have the same logical structure with derivatives acting on the other insertions.
After Fourier transformation these identities fix divergences and selected normalisations. They do not determine the unconstrained transverse form factors.

For a gauge theory physical local insertions are represented by BRST cohomology classes rather than by the condition of BRST closure alone.
The corresponding constraints on gauge-fixed Green functions are the
Slavnov--Taylor identities~\cite{Slavnov:1972,Taylor:1971}.
Writing the BRST differential as $\mathsf s_{\rm B}$ ghost-number-zero closure
$\mathsf s_{\rm B}\cO=0$ is supplemented by the identification
$\cO\sim\cO+\mathsf s_{\rm B}\Psi$ after passage to the equation-of-motion quotient.
A total derivative $dY$ is also quotiented for a constant-source integrated
functional with vanishing boundary term or at zero insertion momentum when
that limit is well defined.  A nonconstant source leaves a source-gradient
insertion after integration by parts and descendants are retained at
nonzero insertion momentum.
A colour component $F^a_{\mu\nu}$ transforms covariantly and is not by itself a gauge-invariant scalar observable whereas $\operatorname{tr}F_{\mu\nu}F^{\mu\nu}$ defines a scalar gauge-invariant class.

\paragraph{Scope of a symmetry identity.}
Fix a renormalised physical operator block with its state and projector.
Assume the relevant Ward or Slavnov--Taylor identity is anomaly free in that block.
After the contact and anomaly terms have been kept explicitly the identity
has the following consequences
\begin{enumerate}[label=(\roman*)]
 \item set symmetry-forbidden projected components to zero up to specified contact terms.
 \item determine longitudinal components in terms of lower-point functions.
 \item fix a zero-transfer form-factor normalisation when stable external states carry a nonzero conserved charge and LSZ reduction applies and the diagonal matrix element has a continuous zero-transfer limit.
 \item leave the remaining physical transverse components symmetry allowed.
\end{enumerate}
The state and dynamical response determine vacuum spectral weight and the presence of an isolated pole.
After Fourier transformation and projection with LSZ amputation in part~(iii)
Eq.~\eqref{eq:ward-contact} together with the stated representation and
regularity assumptions gives (i)--(iii).  Part~(iv) is the complement of those
constraints in the physical operator block.

This leaves form factors not fixed by the specified symmetry
to the dynamical calculation.
Symmetry thus organises which disturbances can couple to which outputs
and fixes selected normalisations while the state and interactions
determine the remaining spectral weight.  Comparing channels requires
both parts of this information.
Composite scalar and fermion-bilinear channels may contain bound-state poles even though no conserved-current normalisation fixes their residues.

The Bethe--Salpeter construction describes bound states through a
two-particle kernel~\cite{BetheSalpeter:1951}.  In a sourced channel this
kernel enters the inhomogeneous response equation
\begin{equation}
 \Gamma_{\cC}(z)=\Gamma^{(0)}_{\cC}(z)
      +\cB_{\cC}(z;\eta)\Gamma_{\cC}(z),
 \qquad
 \cB_{\cC}(z;\eta)
 =\cK_{\rm 2PI,\cC}(z;\eta)G_{0,\cC}(z;\eta).
 \label{eq:inhomogeneous-bse}
\end{equation}
The dimensionless eigenvalue problem is the product of the irreducible kernel
$\cK_{\rm 2PI}$ with the two-particle propagation operator $G_0$.
For an eigenvalue-based solution of the homogeneous bound-state equation
see Ref.~\cite{HaradaYoshida:1996}.  A concrete application must show that
its 2PI Bethe--Salpeter or Dyson--Schwinger or FRG construction represents the
Jacobian in Eq.~\eqref{eq:response-jacobian} exactly or with a stated defect.

A vacuum spectral representation gives the residue a particle interpretation.
Resolving the centred insertions into physical intermediate states writes
their Wightman matrix as a positive measure of excitation amplitudes.
This is the intermediate-state construction of the
K\"all\'en--Lehmann representation~\cite{Kallen:1952,Lehmann:1954}.
Compression onto positive-norm polarisations preserves that positivity.
Euclidean continuation then expresses the same source response as a mass
spectral integral with local subtractions recorded separately.  Matching
these insertions and conventions to the response equation allows its
Laurent coefficient to be read as the weight of the zero-mass atom.
The matching conditions identify both coefficients through the same infrared limit.

\begin{theorem}[Fredholm--KL bridge]
\label{thm:fredholm-kl-bridge}
Let the hypotheses of Theorem~\ref{thm:kernel-pole}(ii)--(iii) hold and set
$z=Q^2$ on the Euclidean axis.  Assume in addition that
\begin{enumerate}[label=(\roman*)]
 \item at the physical endpoint and with $\tau_Q=0$ the matrix in
 Eq.~\eqref{eq:coarse-grained-response-factorisation} is exactly the
 projected unfiltered connected Euclidean covariance of a physical
 spectral realisation in
 Definition~\ref{def:physical-spectral-realisation} with the vacuum and
 spectral assumptions of Appendix~\ref{app:kl}.
 \item the tensor or spin part of $\mathcal P_{\cC}$ uses the
 physical-polarisation compression in
 Assumption~\ref{ass:physical-polarisation-compression} including its
 separately defined zero-mass helicity endpoint and the external basis and vacuum centring together with the source pairing and continuation convention and subtraction polynomial $P(Q^2)$ agree in the response and spectral constructions.
 \item after subtraction of $P(Q^2)$ the direct term and the complementary
 Fredholm resolvent contribution are analytic at $Q^2=0$.
 \item the continuum part of the physical matrix-valued spectral measure is locally finite and has no atom at zero and satisfies the ultraviolet
 integrability hypotheses of Proposition~\ref{prop:extract-atom}.
\end{enumerate}
Then the algebraic residue is the physical zero-mass atom matrix
\begin{equation}
 Z_{\rm alg}
 =\lim_{Q^2\downarrow0}Q^2[C(Q^2)-P(Q^2)]
 =Z_0\succeq0.
 \label{eq:fredholm-kl-identity}
\end{equation}
Consequently a nonzero observable Fredholm pole satisfying these matching
hypotheses is a physical zero-mass spectral atom.  Conversely if an exact
calculation that purports to satisfy (i)--(iv) produces a non-Hermitian or
indefinite $Z_{\rm alg}$ the response and spectral constructions violate at least one of the matching or positivity assumptions.
\end{theorem}

\begin{proof}
The Fredholm decomposition gives
$C(Q^2)=Z_{\rm alg}/Q^2+A(Q^2)$ with $A$ analytic at the origin under
assumption~(iii).  Hence
$\lim_{Q^2\downarrow0}Q^2[C(Q^2)-P(Q^2)]=Z_{\rm alg}$.
The same connected matrix is by assumptions~(i)--(ii) the correlator entering the
subtracted K\"all\'en--Lehmann representation.  Assumption~(iv) permits
Proposition~\ref{prop:extract-atom} to identify the same limit with $Z_0$.
Positivity of the Wightman measure gives $Z_0\succeq0$.
\end{proof}

The bridge applies to the isolated meromorphic endpoint specified in its
hypotheses.  At the level of the spectral measure a zero-mass atom may also
coexist with continuum weight accumulating at $s=0$.
Proposition~\ref{prop:extract-atom} still extracts the atomic coefficient
under its local-finiteness and integrability hypotheses while the threshold
density determines an additional non-meromorphic contribution and its own
long-distance tail.  At such an endpoint the atom is defined by the measure limit and the threshold contribution is treated through its boundary behaviour. The isolated Fredholm pole describes a meromorphic response. The spectral atom also has meaning at a continuum threshold.

\paragraph{Spatial response of the pole.}
The spatial meaning can be made explicit in a translation-invariant
realisation with three spatial dimensions.  Suppose a fixed physical
source contraction $j$ has static pole contribution
$Z_j/|\boldsymbol q|^2$ where $Z_j=j^*Z_0j>0$ is independent of
$\boldsymbol q$ and $Q^2=|\boldsymbol q|^2$ at zero frequency.
For the external Gaussian representation already defined put
$r=|\boldsymbol x|>0$ and $\tau_Q>0$.  The pole alone gives
\begin{align}
 K^{\rm pole}_{\tau_Q}(r)
 &=Z_j\int\frac{\mathrm d^3q}{(2\pi)^3}
     \frac{e^{-\mathrm i\boldsymbol q\cdot\boldsymbol x}
           e^{-\tau_Q|\boldsymbol q|^2}}{|\boldsymbol q|^2}
   =Z_j\int_{\tau_Q}^{\infty}
      \frac{e^{-r^2/(4u)}}{(4\pi u)^{3/2}}\,\mathrm du
   \nonumber\\
 &=\frac{Z_j}{4\pi r}
     \operatorname{erf}\!\left(\frac{r}{2\sqrt{\tau_Q}}\right).
 \label{eq:static-gaussian-pole-kernel}
\end{align}
Here $\operatorname{erf}(v)=(2/\sqrt\pi)\int_0^v e^{-t^2}\,\mathrm dt$.
The integral representation follows by writing
\[
 \frac{e^{-\tau_Q|\boldsymbol q|^2}}{|\boldsymbol q|^2}
 =\int_{\tau_Q}^{\infty}e^{-u|\boldsymbol q|^2}\,\mathrm du.
\]
The Gaussian Fourier integral and the substitution $r/(2\sqrt u)$
give the remaining equalities.  At distances large compared with
$\sqrt{\tau_Q}$ this contribution approaches $Z_j/(4\pi r)$.
External smoothing resolves its short-distance behaviour while preserving
its long-distance coefficient.  For a tensor source with a
momentum-dependent numerator that numerator remains inside the transform.
The complete spatial response also includes the transformed regular and
continuous-spectrum terms. Dominance of the pole requires bounds on
those terms.  Equation~\eqref{eq:static-gaussian-pole-kernel} is a response
kernel. A force interpretation requires the corresponding matter coupling
and interaction energy.

\paragraph{A realised helicity response.}
The bridge can be evaluated explicitly for the two physical helicities
of a free gauge tensor on Minkowski space. Let $\phi_\lambda$ with $\lambda=\pm2$ denote unit-normalised transverse traceless amplitudes
with reduced action
\[
 S_2=\frac Z2\sum_{\lambda=\pm2}\int
    [(\partial_t\phi_\lambda)^2-|\nabla\phi_\lambda|^2]\,\mathrm d^4x,
 \qquad Z>0.
\]
The one-particle space is
$L^2(H_0^+,\mathrm d\Omega_0)\otimes\mathbb C^2$ with the unitary massless
helicity representations $+2$ and $-2$. Its symmetric Fock space has the
Poincar\'e-invariant free vacuum and centred insertions
$\Theta_\lambda=g\phi_\lambda$ paired with physical helicity test sections.
Their positive-frequency measure and continued covariance are
\begin{equation}
 \begin{aligned}
 \widehat W_{\lambda\lambda'}(p)
  &=\frac{2\pi g^2}{Z}\theta(p^0)\delta(p^2)\delta_{\lambda\lambda'},\\
 C_{\lambda\lambda'}(Q^2)
  &=\frac{g^2}{ZQ^2}\delta_{\lambda\lambda'}.
 \end{aligned}
 \label{eq:realized-helicity-bridge}
\end{equation}
The Fourier convention absorbs the displayed $2\pi$ in the mass-shell
measure. At fixed $p=|\boldsymbol p|>0$ the reflected time covariance
is $g^2e^{-p(t_i+t_j)}/(2Zp)$ times the helicity identity and is a positive Gram matrix. This establishes the directly reduced realisation of
Definition~\ref{def:physical-spectral-realisation}(ii) including its
continuation and null-shell source pairing. The subtraction polynomial and direct term together with the complementary contribution all vanish. Thus the two
helicity channels separately satisfy the simple-eigenvalue and realisation
hypotheses of Theorem~\ref{thm:fredholm-kl-bridge}. Their direct sum gives
the displayed two-dimensional residue.
Equivalently Theorem~\ref{thm:native-positive-residue} has $L(z)=ZzI_2$ and $P=I_2$ with $S=gI_2$ giving the same matrix
$Z_0=(g^2/Z)I_2$. The two constructions use the same vacuum and source normalisation together with the helicity endpoint. This response is the physical
quotient of the quadratic macroscopic tensor action. Its covariant
Einstein--matter completion~\cite{Deser:1970} supplies a common matter source and
constraint propagation while the displayed translation-invariant
spectral measure describes the flat-background linear response.
The coupling of this tensor response to stationary matter and the
separate potential and force estimates form a companion application in
Ref.~\cite{Guo2026EmergentGravity}.  The response and
spectral-identification theorems above are proved independently.

\paragraph{Finite truncations.}
An exact finite operation protocol can satisfy the native response and
positive-residue theorems without an infinite-cutoff limit.  Its interpretation
as a QFT spectral atom still needs a corresponding physical representation.
By contrast a truncation approximating an infinite target must establish
that limiting identification.
A finite response truncation that satisfies the algebraic pole and overlap
tests produces a candidate residue $Z_{{\rm alg},N}$.  Even if this matrix is
Hermitian and positive semidefinite it represents a physical spectral atom
only after the realisation maps and kernels together with their $z$-derivatives converge along with the source and readout maps and subtracted projected correlators so that
the critical root remains isolated and simple and its slope stays bounded
away from zero and the hypotheses of Theorem~\ref{thm:fredholm-kl-bridge}
hold in the limit.  Proposition~\ref{prop:riesz-convergence} gives the
corresponding convergence result.  This distinction is necessary because positivity of a finite matrix
does not identify the limiting correlator or its spectral measure.

Ladder diagrams provide a possible approximation to $\cK_{\rm 2PI}$. A
complete response may also contain crossed ladders and vertex corrections together with ghost contributions and higher-particle sectors.
Likewise the sign of a closed fermion loop does not by itself determine the sign of a projected response eigenvalue. Vertices and group factors together with propagators and the kernel convention enter the result.

\subsection{Classification and scope of the conclusions}
\label{sec:execution}

The classification returns four complementary entries
\begin{equation}
 \mathfrak R(\cC_{\rm phys};\cI)
 =\bigl(R_{\rm kin},R_{\rm sym},R_{\rm dyn},R_{\rm scope}\bigr).
 \label{eq:classification-output}
\end{equation}
Here $R_{\rm kin}$ records momentum suppression and soft scaling. The entry $R_{\rm sym}$ records the components fixed or forbidden by an exact symmetry and any anomaly or breaking. The entry $R_{\rm dyn}$ records the spectral conclusion established for the response. The entry $R_{\rm scope}$ specifies its response-kernel or topological realisation. A symmetry-allowed channel can have Gaussian decay at fixed transfer while retaining soft spectral weight determined by its dynamics.

Averaging fixes the momentum regions that retain support. Symmetry restricts the corresponding source components. The full response kernel determines collective criticality and the physical realisation assigns its spectral meaning. These are complementary properties of the same channel.

For an isolated simple eigenvalue criticality means $1\in\spec\cB(0)$ and transversality means $\partial_z\lambda(0)\neq0$. Both nonzero source--readout overlaps make the resulting pole visible. Theorem~\ref{thm:native-positive-residue} gives a positive residue for the reciprocal Hilbert-space realisation. Theorem~\ref{thm:fredholm-kl-bridge} identifies the coefficient with a Lorentzian atom when the unfiltered covariance and spectral conventions match. A continuum threshold instead requires its boundary analysis. Noncriticality is established for the exact operator or through a controlled approximation to it. An exact finite operation protocol is treated directly as the specified system.

\paragraph{Compatibility with a physical realisation.}
For a QFT channel the operator quotient and projector $\mathcal P_{\cC}$
first determine the source and readout $\mathsf S_{\cO},\mathsf R_{\cO}$.
Their realisation on the Banach response space uses $P_X$ and
$\iota_{\cC},\pi_{\cC}$ subject to
Eqs.~\eqref{eq:channel-realisation-compatibility}--
\eqref{eq:response-source-readout}.  The kernel is the Jacobian of the
same sourced state equation.  An approximate coarse equation also requires
comparison at actual reference solutions.  If the coarse solution differs
from the image of the full solution its displacement must be bounded. State derivatives of the kernel and source and readout maps then convert that
bound into response errors.  Residues must be compared at the two actual
critical roots including the variation of their slopes.

If $\Ran P_X$ is not invariant the four Schur--Feshbach terms of
Proposition~\ref{prop:response-closed-channel} retain the omitted feedback.
On the resulting Fredholm domain the simple zero and the two overlaps
determine the observable pole.  For a finite approximation to an infinite
target operator-norm and derivative convergence together with Riesz-projection convergence transfer
this conclusion to the limiting response.

Its spectral interpretation concerns the same unfiltered connected
endpoint correlator.  Common vacuum assumptions and centring together with the physical projection and continuation and source pairing and subtraction polynomial are
precisely the matching hypotheses of
Theorem~\ref{thm:fredholm-kl-bridge}.  They identify the residue with a
physical atom.  Theorem~\ref{thm:native-positive-residue} instead establishes
positivity directly for the reciprocal Hilbert-space response at the level of the reciprocal response operator.

Two exclusions complement these positive conclusions.  External Gaussian
filtering preserves an existing zero-mass residue and cannot create a pole
from a regular form factor as shown around Eq.~\eqref{eq:pole-residue}.
If an exact symmetry sets the complete separated-point block to zero
uniqueness of its spectral decomposition gives $Z_0=0$. Contact polynomials contribute only local terms.  Neither statement evaluates
an unevaluated symmetry-allowed response kernel.
\subsection{Four-axis classification}
\label{sec:four-axis}

For a channel block $\cC_{\rm phys}$ and a computational realisation $\cI$ the result
$\mathfrak R(\cC_{\rm phys};\cI)$ of Eq.~\eqref{eq:classification-output} is read as follows.

Theorem~\ref{thm:fixed-transfer} gives decay on support separated from the zero set of the routed quadratic form. The soft region $q=O(\tau^{-1/2})$ and directions in $\Ker M$ retain their own scaling. The same channel can exhibit both behaviours in different momentum regions.

An exact selection rule sets a forbidden component to zero. A Ward or Slavnov--Taylor identity can instead fix a component or its normalisation while leaving other form factors dynamical. Anomalies and explicit breaking modify the operator identity. Breaking by the state changes its matrix elements.

The dynamical statement specifies the kernel and external variable whose crossing is tested. Transversality and source visibility determine the observable pole. Spectral matching determines its physical atomic weight. Estimates at a continuum threshold use the measure limit or a limiting-absorption construction.

The response-kernel mechanism uses the Jacobian of the sourced equations. An internal topological zero mode is fixed instead by an index or cohomological construction. The native positive-residue theorem also applies to an isolated finite-dimensional critical subspace of an analytic Hilbert-space response operator.

These entries answer different questions.
In particular neither ``allowed'' nor ``constrained'' entails spectral
criticality and fixed-transfer suppression is compatible with a nontrivial
soft response.

\subsection{Representative channels}
\label{sec:representative-classification}

Table~\ref{tab:channel-classification} applies the four-axis language to a
representative set of channels.  The dynamical column specifies the operator calculation that determines the corresponding spectrum.

\begin{table}[htbp]
 \centering
 \caption{Representative renormalised channels.  Under the hypotheses of
 Theorem~\ref{thm:fixed-transfer} fixed-transfer decay holds away from the zero set of the routed Gaussian quadratic form. The soft and degenerate regions have the scaling of Eq.~\eqref{eq:soft-scaling}.}
 \label{tab:channel-classification}
 \footnotesize
 \setlength{\tabcolsep}{3pt}
 \renewcommand{\arraystretch}{1.08}
 \begin{tabular}{@{}>{\raggedright\arraybackslash}p{0.155\textwidth}>{\raggedright\arraybackslash}p{0.195\textwidth}>{\raggedright\arraybackslash}p{0.14\textwidth}>{\raggedright\arraybackslash}p{0.255\textwidth}>{\raggedright\arraybackslash}p{0.17\textwidth}@{}}
  \toprule
  Channel & $R_{\rm sym}$ & $R_{\rm kin}$ & $R_{\rm dyn}$ & $R_{\rm scope}$ \\
  \midrule
  $T_{\mu\nu}^{\TT}$
   & conserved. Full tensor covariance and positivity
   & fixed-transfer decay
   & $s>0$ spin-two kernel. $s=0$ TT atom excluded
   & QFT response kernel. $d=4$ vacuum \\
  \addlinespace[2pt]
  conserved $J_\mu$
   & charge normalisation conditional. Transverse weight unfixed
   & fixed-transfer decay
   & vector kernel calculation required
   & response kernel \\
  \addlinespace[2pt]
  $[\bar\psi\psi]_R$ block
   & no current normalisation
   & fixed-transfer decay
   & scalar bound-state kernel calculation required
   & response kernel \\
  \addlinespace[2pt]
  longitudinal nonsinglet $J^5_\mu$ block
   & exact anomaly-free chiral identity. State dependent
   & fixed-transfer decay
   & Goldstone test with nonzero symmetry-breaking expectation. Kernel otherwise
   & Ward identity / response \\
  \addlinespace[2pt]
  longitudinal $J^5_\mu$ block with anomaly or explicit breaking
   & anomaly or explicit breaking
   & fixed-transfer decay
   & pseudoscalar kernel calculation required
   & response kernel \\
  \addlinespace[2pt]
  $[\operatorname{tr}F^2]_R$ block
   & physical BRST cohomology class
   & fixed-transfer decay
   & scalar glueball kernel calculation required
   & response kernel \\
  \addlinespace[2pt]
  $[H^\dagger H]_R$ $[(H^\dagger H)^2]_R$
   & gauge invariant. No charge normalisation
   & fixed-transfer decay
   & mixed scalar kernel calculation required
   & response kernel \\
  \addlinespace[2pt]
  internal gauge zero modes
   & topology and background dependent
   & not applicable
   & index or harmonic analysis required
   & internal topology \\
  \bottomrule
 \end{tabular}
 \par\smallskip
 \begin{minipage}{\textwidth}
 \footnotesize
  Charge normalisation assumes stable external states and LSZ reduction with a
  continuous diagonal zero-transfer limit.
  The stress-tensor row assumes positivity and exact covariance of the
  complete tensor before TT compression. Proposition~\ref{prop:stress-null-shell} classifies its actual zero-mass shell.
 The axial rows refer to longitudinal current--scalar or pseudoscalar mixing blocks.  Identification of a physical zero-mass atom uses the centred connected vacuum correlator and the matching and positivity and threshold hypotheses of Theorem~\ref{thm:fredholm-kl-bridge}.
 \end{minipage}
\end{table}

Several distinctions in the table are essential.
The translation Ward identity alone does not fix the transverse-traceless
stress-tensor form factor.  With complete-tensor positivity and exact Lorentz covariance together with conservation and the invariant-vacuum hypotheses of
Proposition~\ref{prop:stress-null-shell} however the zero-mass TT atom is
excluded.  The spin-two spectrum at $s>0$ still requires a dynamical
calculation.  This distinction concerns the QFT representation. It does not
exclude a native tensor response or a discrete transfer residue that has
not been identified with that complete Wightman tensor measure.
A conserved vector current illustrates the same distinction. An exactly normalised charge can coexist with a gapped current--current spectral function.

The scalar bilinear is not assigned a zero source.
Depending on the theory and state its kernel can contain a bound-state pole.
The axial rows concern longitudinal mixing blocks. Transverse axial response
requires a separate projection and kernel calculation.
An anomaly-free nonsinglet current in a chirally symmetric Lagrangian obeys a different identity from a singlet anomalous current or a current with explicit mass breaking.
In the exact case a Goldstone pole additionally requires a nonzero
symmetry-breaking vacuum expectation in the current Ward identity. The operator name alone does not decide the result.

The notation $\operatorname{tr}F^2$ denotes the scalar $0^{++}$
gauge-invariant composite channel. Its response measures colour-singlet
intensity fluctuations. Gauge-field zero modes on a compact internal
space instead belong to the kernel of a background differential operator.
The compact geometry and gauge and matter representations of
ref.~\cite{Guo2026CompactSpectrumI} give a concrete realisation of these
different spectral questions. A scalar bound-state pole is determined by
its source response. Harmonic analysis determines the constant kernel of
the scalar Laplacian.

\paragraph{A generated colour-singlet response.}
The Yang--Mills construction of ref.~\cite{Guo2026CompactSpectrumI}
Section~8 fixes a common source law with stationary mean intensity $q_L$ on the condensed branch $L>L_c$ and response operator
\[
 \mathcal M_{\rm YM}^2
 =m_{\rm YM,rad}^2\oplus(\Delta_1^\perp+\delta_{q_L}I),\qquad
 \Delta_{\rm YM}^2
 =\min\{m_{\rm YM,rad}^2,4/L^2+\delta_{q_L}\}>0.
\]
The coexact one-form Laplacian $\Delta_1^\perp$ has bottom $4/L^2$
and $\delta_{q_L}>0$ is the source-induced restoring shift.
Its equilibrium response completion has the second-quantised generator
$K_{\rm YM}=d\Gamma(\mathcal M_{\rm YM}^2)$. The vacuum is the only
zero mode and the centred colour-invariant cyclic space has spectrum
bounded below by $\Delta_{\rm YM}^2$. For a centred source vector $v$
the response $\langle v,(K_{\rm YM}+z)^{-1}v\rangle$ therefore has a
positive spectral measure supported away from zero.
Write $\mathcal Q=\|F_B\|^2$ for the curvature intensity of the generated
connection sample $B$ and put
$\sigma_L^2=\operatorname{Var}_{q_L}\|B\|^2$. The source-mean radial
coefficient is the normalised derivative
$b_{\mathcal Q}=\sigma_L\partial_q\mathbb E_q\mathcal Q|_{q=q_L}>0$.
It gives
\[
 C_{\rm rad}(z)=\frac{b_{\mathcal Q}^{\,2}}{m_{\rm YM,rad}^2+z},
 \qquad
 d\rho_{\rm rad}(\mu)=b_{\mathcal Q}^{\,2}
                       \delta_{m_{\rm YM,rad}^2}(d\mu).
\]
Here $z$ is the Euclidean resolvent variable. This readout detects the
spectral bottom when the radial block is the lower block. The fixed-source
curvature polynomial retains its original products and higher Hermite
components in the same completion. The displayed response measures the
variation of its mean. In the condensed quadratic field realisation of
that reference Section~10 and Appendix~C a Wick intensity
also acquires a linear radial term. Its two-point function then has an
isolated positive-mass atom and a continuum beginning at the two-particle
threshold. This model shows how a generated spectrum and specified source together with a nonzero readout together determine an observable pole. The
mass-squared generator and the external invariant-mass measure are
identified by the particular field realisation.

Finally $\|D H\|^2$ and four-fermion monomials are local composite operators.
Whether they are independent representatives depends on integration by parts and field redefinitions together with EOM relations in the chosen operator basis.
They cannot be excluded merely because one is a kinetic insertion and the other is a four-point interaction.

\subsection{Symmetry permission and dynamical existence}
\label{sec:permission-existence}

For a symmetry-allowed block the principal output is a well-posed dynamical
question.
If a symmetry forbids a projected block no kernel calculation is needed for that block.
If the block is allowed the classification identifies the response operator and quantum-number sector together with the relevant soft kinematics in which a pole could occur.
The existence question is thereby reduced to a specified spectral problem.

The distinction can be phrased directly in spectral language.
For a Hermitian operator block symmetry controls linear relations among the matrix elements of the positive spectral measure $d\rho_{ij}(s)$.
It does not ordinarily fix whether $\rho_{ij}$ has an atom at $s=0$.
Such an atom can be forced when additional state information closes a Ward
identity as in the Goldstone theorem.  A critical response kernel yields
an observable pole under the hypotheses and source--readout visibility
conditions of Theorem~\ref{thm:kernel-pole}. Under the additional matching
hypotheses of Theorem~\ref{thm:fredholm-kl-bridge} its residue is the
physical zero-mass atom matrix.
Otherwise its weight is a dynamical observable.

This observation eliminates a false dichotomy between ``protected'' and ``unprotected'' operators.
Protection can refer to current conservation or gauge-parameter independence or a nonrenormalisation statement or membership in physical BRST cohomology.
These notions are related but not equivalent.
The classification therefore records the actual identity and its consequence instead of assigning a universal binary attribute.

\subsection{External basis covariance, response coordinates, and schemes}
\label{sec:basis-covariance}

Two different transformations must be kept separate.  First let
$O=(O_1,\ldots,O_m)^T$ be the external physical operator basis and let
\begin{equation}
 O'=M O,
 \qquad M\in GL(m,\mathbb C),
 \label{eq:basis-change}
\end{equation}
be a finite change of renormalised basis.
The Wightman spectral matrix and Euclidean correlator together with any isolated residue matrix transform by congruence
\begin{equation}
 d\rho'=M\,d\rho\,M^\dagger,
 \qquad
 C'=MCM^\dagger,
 \qquad
 Z'=MZM^\dagger.
 \label{eq:congruence-transform}
\end{equation}
This is an external change of source and readout coordinates.
The same physical source has coordinates $j'=(M^\dagger)^{-1}j$ hence
$j'^\dagger Z'j'=j^\dagger Zj$.  The pole strength measured by that
source is unchanged even when its individual matrix components change.

Second let $T:\mathcal X_{\cC}\to\mathcal X'_{\cC}$ be a bounded
isomorphism that changes only the internal coordinates of the response
state.  Then
\begin{equation}
 \begin{aligned}
  \cB_T&=T\cB T^{-1},
  &P_{X,T}&=TP_XT^{-1},\\
  \mathsf S_{X,T}&=T\mathsf S_X,
  &\mathsf R_{X,T}&=\mathsf R_XT^{-1}.
 \end{aligned}
 \label{eq:response-coordinate-similarity}
\end{equation}
It leaves the output matrix unchanged.  If an external basis change $M$ and
an internal coordinate change $T$ are performed together the complete
transformation law is
\begin{equation}
 \begin{aligned}
  \cB'&=T\cB T^{-1},\\
  P'_X&=TP_XT^{-1},\\
  \mathsf S'_X&=T\mathsf S_XM^\dagger,\\
  \mathsf R'_X&=M\mathsf R_XT^{-1},\\
  C'_{\rm reg}&=MC_{\rm reg}M^\dagger.
 \end{aligned}
 \label{eq:combined-basis-response-transform}
\end{equation}
Therefore
\begin{equation}
 C'=C'_{\rm reg}+\mathsf R'_X(I-\cB')^{-1}\mathsf S'_X
 =MCM^\dagger,
 \label{eq:combined-response-congruence}
\end{equation}
as required by the Wightman construction.
The adjoint on the source leg follows the convention
$C_{ij}=\langle O_iO_j^\dagger\rangle$. In a real bilinear convention it is
replaced consistently by the corresponding transpose.

\begin{proposition}[External and internal covariance]
\label{prop:basis-covariance}
A finite nonsingular external transformation $M$ and bounded internal isomorphism $T$ preserve the location and order of full-matrix poles. They preserve the rank and positive-semidefinite character of each isolated residue. The symmetry subspaces transform by the isomorphism induced by $M$ and retain their dimensions. The internal response spectrum is unchanged. At an isolated critical eigenvalue vanishing and nonvanishing of both source and observable overlaps are invariant. Individual residue components transform with the chosen basis.
\end{proposition}

The proof is given in Appendix~\ref{app:operator-mixing}.
For a scale-dependent finite transformation $M(\mu)$ the anomalous-dimension matrix obeys
\begin{equation}
 \gamma'=M\gamma M^{-1}
 -\left(\mu\frac{dM}{d\mu}\right)M^{-1}.
 \label{eq:gamma-transform}
\end{equation}
The extra term is why a running operator basis cannot be treated as a constant similarity transformation.

Exact physical pole positions are independent of a consistent renormalisation scheme.
The transformation changes intermediate Wilson coefficients and contact terms together with operator normalisations and mixing-matrix entries.
At finite truncation even the inferred kernel eigenvalues and crossing scales generally depend on the regulator and basis size together with the projection prescription.
Claims of scheme independence must therefore be demonstrated by the exact transformation law or replaced by a convergence and regulator-variation study.

For Higgs-sector applications the construction identifies the scalar mixing
block and the kernel governing its correlations.  Ultraviolet boundary
conditions for the mass and quartic coupling remain independent
renormalisation or phenomenological inputs.

%% file: sections-en/sec04_irreversibility.tex
\section{Coarse-response errors, residual states, and evolution}
\label{sec:irreversibility}

Exact observation-sufficient descent preserves the full response.  When a
chosen coarse description is only approximate the first object to compare
is that response even if the state dimensions differ.  The comparison
below needs no inverse coarse map.  Its four defects distinguish unresolved-state effects from branch displacement and source-coupling errors.

\subsection{A response error identity without an inverse coarse map}
\label{sec:noninvertible-error}

The two state spaces may have different dimensions so their resolvents
cannot in general be subtracted directly.  The map $C$ places the update
defect in a common space while the common source and output spaces allow
the observable responses to be compared.  The identity below propagates
each of the four defects to that output.

\begin{theorem}[Irreversible coarse-response error]
\label{thm:irreversible-error}
Let $B,S,R,D$ and $\bar B,\bar S,\bar R,\bar D$ be bounded response
data on Banach state spaces $\mathcal X,\mathcal Y$ with the same source
and output spaces.  Let $C:\mathcal X\to\mathcal Y$ be any bounded
linear comparison map.  Set $G=(I-B)^{-1}$ and $\bar G=(I-\bar B)^{-1}$
where both inverses exist and define
\begin{equation}
 \begin{aligned}
 E_B&=\bar BC-CB,&E_S&=\bar S-CS,\\
 E_R&=\bar RC-R,&E_D&=\bar D-D.
 \end{aligned}
 \label{eq:irreversible-defects}
\end{equation}
Then the complete responses obey the exact identity
\begin{equation}
 \bar{\mathscr C}-\mathscr C
 =E_D+\bar R\bar GE_S+\bar R\bar GE_BGS+E_RGS.
 \label{eq:irreversible-response-error}
\end{equation}
In particular writing $e_i=\|E_i\|$ in the specified operator norms
\begin{equation}
 \|\bar{\mathscr C}-\mathscr C\|
 \leq e_D+\|\bar R\|\|\bar G\|e_S
 +\|\bar R\|\|\bar G\|e_B\|G\|\|S\|
 +e_R\|G\|\|S\|.
 \label{eq:irreversible-response-bound}
\end{equation}
Neither injectivity nor surjectivity of $C$ is required.
\end{theorem}

\begin{proof}
From $(I-\bar B)C=C(I-B)-E_B$ obtain
$\bar GC-CG=\bar GE_BG$.  Inserting $\bar S=CS+E_S$ into
$\bar D+\bar R\bar G\bar S-(D+RGS)$ and using
$\bar RC=R+E_R$ gives Eq.~\eqref{eq:irreversible-response-error}.
Submultiplicativity gives the norm bound.
\end{proof}

The update defect is multiplied by both resolvents whereas the source
and readout defects each carry one.  Thus near criticality the response
error depends on where a defect enters as well as on its norm.
For example in a scalar reciprocal response with $S=R=1$ and $D=0$ and $B=1-\epsilon$ for $0<\epsilon<1$ the susceptibility is $\epsilon^{-1}$.  Changing $B$ to
$B+\delta b$ changes it by
\[
 \frac1{\epsilon-\delta b}-\frac1\epsilon
 =\frac{\delta b}{\epsilon(\epsilon-\delta b)}.
\]
For $0<|\delta b|\ll\epsilon$ the relative error is approximately
$|\delta b|/\epsilon$.  Thus the distance from criticality sets the
required accuracy of the reduced feedback.  Comparing the displaced
roots and their residues describes this sensitivity even when no uniform
pointwise comparison across the poles is possible.
Directional amplification also matters for nonnormal updates.  With the
Euclidean norm on $\mathbb R^2$ and $g\in\mathbb R$ take
\[
 B=\begin{pmatrix}0&g\\0&0\end{pmatrix},\qquad
 S=e_2,\quad R=e_1^{\mathsf T},\quad D=0.
\]
Here $B^2=0$ and both eigenvalues vanish but
$R(I-B)^{-1}S=R(I+B)S=g$.  The source excites the second coordinate
which the update transfers into the observed first coordinate with gain
$g$.  This finite response can be large at a fixed distance from the
kernel spectrum. The resolvent norms in the error estimate account for
such amplification as well as for proximity to a pole.  This role of
the resolvent is central to the pseudospectral analysis of nonnormal
operators~\cite{TrefethenEmbree:2005}.
For the source-dependent coarse map in
Eq.~\eqref{eq:moving-coarse-derivatives} apply this theorem with $C=T$
and the corrected data
\begin{equation}
 \bar S_{\rm corr}=\bar S-(I-\bar B)C_h,
 \qquad \bar D_{\rm corr}=\bar D+\bar RC_h.
 \label{eq:corrected-coarse-data}
\end{equation}
They give exactly the same coarse response since the two corrections
cancel in $\bar D_{\rm corr}+\bar R\bar G\bar S_{\rm corr}$.

All comparisons are at the actual reference states.  For a fixed linear
$C$ let $\delta=\|\bar X_*-CX_*\|$ be independently bounded and let
$e_i^{(0)}$ compare the data at $X_*$ and the transported point $CX_*$.
If the coarse data have state-Lipschitz constants $L_i$ on the intervening
neighbourhood the resulting bounds are
\begin{equation}
 \begin{aligned}
 e_B&\leq e_B^{(0)}+L_B\|C\|\delta,&
 e_S&\leq e_S^{(0)}+L_S\delta,\\
 e_R&\leq e_R^{(0)}+L_R\|C\|\delta,&
 e_D&\leq e_D^{(0)}+L_D\delta.
 \end{aligned}
 \label{eq:four-map-displacement}
\end{equation}
For a moving map its $T,C_h$ errors must also enter the corrected defects.
Near criticality neither a small stationary residual nor the norm of $C$
alone bounds $\delta$.  Resolvent growth in
Eq.~\eqref{eq:irreversible-response-bound} records the associated
sensitivity rather than removing it.

\paragraph{Contour control of the total output residue.}
Suppose both responses extend meromorphically to a neighbourhood of a
simple closed positively oriented rectifiable contour $\Gamma$ and its
interior with finitely many poles inside and none on $\Gamma$.  If
Eq.~\eqref{eq:irreversible-response-bound} is bounded by $b_\Gamma$
on $\Gamma$ then
\begin{equation}
 \left\|\sum_{\Gamma}\operatorname{Res}\bar{\mathscr C}
       -\sum_{\Gamma}\operatorname{Res}\mathscr C\right\|
 \leq\frac{\operatorname{length}(\Gamma)}{2\pi}b_\Gamma.
 \label{eq:contour-output-residue-error}
\end{equation}
Indeed integration of Eq.~\eqref{eq:irreversible-response-error} followed
by the matrix-valued residue theorem and the contour norm bound gives
Eq.~\eqref{eq:contour-output-residue-error}.

This controls the sum of output residues inside the contour.  A branch cut inside the contour would
contribute its jump integral and invalidate a pole-only interpretation.
Comparison of individual roots uses an isolated spectral branch with a simple zero and slope bounded away from zero together with control of its Riesz projection.  These conditions distinguish a root comparison from a bound
on the sum of residues inside a contour.

\paragraph{Quadratic observations.}
Many physical readouts are intensities or energies. On a common Hilbert
space let $\mathsf Q\geq0$ be bounded and self-adjoint and let
$\widehat{\mathsf Q}$ approximate the observation operator. Adding and
subtracting $\langle y,\mathsf Qy\rangle$ gives
\[
 \big|\langle x,\mathsf Qx\rangle
       -\langle y,\widehat{\mathsf Q}y\rangle\big|
 \leq\|\mathsf Q\|(\|x\|+\|y\|)\|x-y\|
       +\|\mathsf Q-\widehat{\mathsf Q}\|\,\|y\|^2.
\]
Thus a state-error bound and an observation-operator bound determine the
accuracy of the quadratic output. For an unbounded closed positive form
the corresponding comparison uses convergence in its form norm. A
bounded linear step $T$ that admits a source $j$ also changes the
quadratic output through the mixed term
\[
 \langle Tx+j,\mathsf Q(Tx+j)\rangle
 -\langle Tx,\mathsf QTx\rangle
 =2\operatorname{Re}\langle Tx,\mathsf Qj\rangle
   +\langle j,\mathsf Qj\rangle.
\]
The mixed term records the relative direction of the incoming source
and the state already present. The finite quadratic identities in
ref.~\cite{Guo2026GeneratedDynamics} Section~5 retain these terms
through successive operations. Combined with a residual bound for the governing evolution these identities relate the disposition of a measured quantity among retained and unresolved channels to the accuracy of the computed evolution.

\subsection{Residual states and exact memory}
\label{sec:residual-memory}

When fibre consistency fails the next coarse state depends on the
unresolved coordinate.  Its elimination produces memory as in the
projection-operator formulation of Mori~\cite{Mori:1965}.
For the bounded discrete linear system considered here the terms follow
directly from the block equations
\begin{equation}
 \begin{aligned}
 x_{n+1}&=A_{11}x_n+A_{12}u_n+S_1h_n,\\
 u_{n+1}&=A_{21}x_n+A_{22}u_n+S_2h_n.
 \end{aligned}
 \label{eq:mother-memory-blocks}
\end{equation}
Induction gives
\begin{equation}
 u_n=A_{22}^nu_0+
 \sum_{j=0}^{n-1}A_{22}^{n-1-j}(A_{21}x_j+S_2h_j).
 \label{eq:exact-hidden-history}
\end{equation}
Substitution into the first equation produces the initial-residual term
$A_{12}A_{22}^nu_0$ and two convolution memories acting on the retained state and source respectively.  An observation $R_1x_n+R_2u_n+D h_n$ retains the
same hidden history.  Static Schur elimination represents these contributions in the stationary or transform domain.

Projection-operator descriptions of transport exhibit the same mechanism. Eliminating microscopic coordinates leaves a generalized force
with memory and a local constitutive coefficient requires a further
Markovian or scale-separation approximation~\cite{Mori:1965}.  Here the
static Schur identity closes the stationary response while the actual
update retains the delayed influence of the eliminated coordinates.

Each factor $A_{12}A_{22}^{n-1-j}A_{21}$ describes a disturbance entering the hidden sector and evolving there before returning to the retained one.
The analogous factor with $S_2$ records a source applied directly to
that sector.  Memory therefore preserves the delayed physical influence
of eliminated variables while the initial term preserves the influence
of their preparation.  A closed description can retain these effects
through additional state coordinates or through the history terms.

If $\|A_{22}^k\|\leq M\rho^k$ with $0<\rho<1$ discarding lags
$k\geq m$ incurs on a history with bounds $X_H,H_H$ at most
\begin{equation}
 \frac{\|A_{12}\|M\rho^m}{1-\rho}
 (\|A_{21}\|X_H+\|S_2\|H_H)
 \label{eq:general-memory-tail}
\end{equation}
in one retained update. The initial term remains separately bounded by
$\|A_{12}\|M\rho^n\|u_0\|$.  Replace $A_{12}$ by $R_2$ for the
readout tail.  Summing a geometric series proves these bounds.
Propagation of these local update and readout errors uses a stability estimate for the full trajectory.  If an independently verified augmented evolution
bound gives $e_{n+1}\leq q e_n+\epsilon_n$ then induction gives
$e_n\leq q^ne_0+\sum_{j<n}q^{n-1-j}\epsilon_j$.  In particular
$0\leq q<1$ and bounded forcing errors give a uniform trajectory bound.
For nonlinear dynamics the corresponding propagation bound must be
established from the nonlinear evolution.

\paragraph{Resolution of a generated gauge--matter response.}
The common connection of ref.~\cite{Guo2026CompactSpectrumI} supplies
a nonlinear example in which the unresolved response is determined
explicitly. In its fixed coefficient realization put $U=(A,\psi)$ and
$\mathcal O_{\rm gm}(U)=(F_A,\nabla_A^\rho\psi)$. The same coefficient
heat map $G_t$ acts on the connection and matter fields. Their products
produce the difference
\[
 C_t(U)=G_t\mathcal O_{\rm gm}(U)-\mathcal O_{\rm gm}(G_tU),\qquad
 C_{s+t}(U)=G_sC_t(U)+C_s(G_tU).
\]
The second identity follows by inserting and subtracting
$G_s\mathcal O_{\rm gm}(G_tU)$. With unresolved response $\mathcal R$ define
\[
 \mathcal T_t(U,\mathcal R)=(G_tU,G_t\mathcal R+C_t(U)).
\]
The composition identity gives $\mathcal T_s\mathcal T_t=\mathcal T_{s+t}$
and
\[
 \mathcal O_{\rm gm}(U_t)+\mathcal R_t
   =G_t\bigl(\mathcal O_{\rm gm}(U)+\mathcal R\bigr).
\]
For $A,\psi\in H^\eta$ and $\mathcal R\in H^{\eta-1}$ with $\eta>5/2$ the
modal product estimates in that reference Appendix~A make this a
strongly continuous locally Lipschitz semigroup. Starting with $\mathcal R=0$
the residual at every later resolution is fixed by the initial coupled
fields. Differentiation carries source variations through the same
identity. The augmented state therefore preserves both the resolved
gauge--matter fields and the response carried by their unresolved mode
products. The parameter $t$ orders coefficient resolution on the fixed
internal carrier.

For semilinear evolution of generated spectra common stronger
bounds and a weaker difference norm give estimates on the interval
where the Duhamel map is contractive
\cite{Guo2026GeneratedDynamics}. The residual includes the discarded
nonlinear source as well as the initial discrepancy. Finite response operators add their own approximation and readout errors to this comparison. The state and source-variation estimates propagate local response and memory errors over the interval on which the stability bounds hold.

The arithmetic and two-mode examples below distinguish a faithful
environment retaining all full-system states from an environment sufficient for
one observation protocol.  Minimal dimension is proved there within a
specified finite-dimensional linear class.  The set quotient fixes
observational equivalence. The distinct exponential rates and the Hankel
rank determine the minimal linear realisation in these examples.

\subsection{Actual updates and stationary solvers}
\label{sec:preconditioning}

Internal coordinate similarity and external source-basis congruence re-express a response. Equation preconditioning changes its solver representation. Actual coarse compression removes state directions.  
For a residual $\mathcal E(X,h)=0$ let $M$ be a differentiable bounded
isomorphism from the state space to the residual space with bounded inverse
and introduce the solver $F_{\rm solve}=X-M^{-1}\mathcal E$.
At a solution derivatives of $M$ multiply $\mathcal E=0$ so
\begin{equation}
 I-B_{\rm solve}=M^{-1}L,\qquad
 S_{\rm solve}=-M^{-1}\mathcal E_h,\qquad
 (I-B_{\rm solve})^{-1}S_{\rm solve}=-L^{-1}\mathcal E_h.
 \label{eq:preconditioned-response}
\end{equation}
This proves invariance of the complete static response after the source
is transformed consistently.  It does not preserve the solver spectrum or
its convergence rate.  For $\mathcal E(X,h)=X-h$ the choice $M=I$ gives derivative
zero while $M=I/3$ gives derivative $-2$ and a divergent iteration. Both stationary responses are one.  Numerical convergence therefore
establishes no physical stability unless the iteration has separately been
identified with the actual update.  The update index and Gaussian scale together with the RG coordinate and physical time have their respective definitions.

\subsection{Invertible sector comparisons and displaced branches}
When the comparison map has a bounded inverse on a specified subspace
one can also transport the response operator and its spectral data
between the two sectors.  This additional hypothesis is useful for
coordinate changes and convergent approximations.  General irreversible
maps remain covered by Theorem~\ref{thm:irreversible-error}.  Let
$\mathcal C_{\tau_q}:\mathcal X_{\cC}\to
\mathcal X_{\cC,\tau_q}$ be the stated coarse-graining map and define
\begin{equation}
 E_{\tau_q}(X;z)
 :=\mathcal C_{\tau_q}\mathfrak F(X;z)
  -\mathfrak F_{\tau_q}(\mathcal C_{\tau_q}X;z).
 \label{eq:coarse-graining-closure-defect}
\end{equation}
If $X_*(z)=\mathfrak F(X_*(z);z)$ and
$\bar X_*(z):=\mathcal C_{\tau_q}X_*(z)$ then
\begin{equation}
 E_{\tau_q}(X_*(z);z)
 =\bar X_*(z)-\mathfrak F_{\tau_q}(\bar X_*(z);z).
 \label{eq:transported-fixed-point-residual}
\end{equation}
Thus $E_{\tau_q}(X_*;z)$ is the residual of the transported point in the
coarse fixed-point equation.  Its derivative compares Jacobians at that same
transported point. It does not by itself compare with a Jacobian evaluated
at a different coarse solution.  If $E_{\tau_q}=0$ on a neighbourhood of the
reference solution differentiation gives the exact intertwining relation
\begin{equation}
 \mathcal C_{\tau_q}D_X\mathfrak F(X_*;z)
 =D_X\mathfrak F_{\tau_q}(\bar X_*;z)
  \mathcal C_{\tau_q}.
 \label{eq:response-intertwining}
\end{equation}
For an approximate closure both the base-point residual and its derivative
must therefore be accounted for.

The following stability result uses exact matching of the reference branches.
With this matching the derivative defect compares the two response
Jacobians at actual solutions without a separate state-displacement term.

The estimates proceed from the derivative defect to the critical spectral
subspace then to the location of a simple root and the residue at that
root.  Source and readout errors enter when passing from the internal mode
to the measured response.  Resolvent factors record how these errors can
be amplified near criticality.

\begin{proposition}[Stability under a closure defect]
\label{prop:closure-stability}
Let $X_*(z)=\mathfrak F(X_*(z);z)$ be a reference branch on $U$ and set
$\bar X_*(z):=\mathcal C_{\tau_q}X_*(z)$ and assume the pointwise
reference-branch matching condition
\begin{equation}
 E_{\tau_q}(X_*(z);z)=0,\qquad z\in U.
 \label{eq:reference-branch-matching}
\end{equation}
Then $\bar X_*(z)$ is an actual coarse fixed point.  Define the response
Jacobians on the two matched branches by
\begin{equation}
 \cB(z):=D_X\mathfrak F(X_*(z);z),
 \qquad
 \cB_{\tau_q}(z):=
 D_X\mathfrak F_{\tau_q}(\bar X_*(z);z).
 \label{eq:matched-response-jacobians}
\end{equation}
Let $Y\subset\mathcal X_{\cC}$ and $\widetilde Y\subset
\mathcal X_{\cC,\tau_q}$ be comparison spaces invariant under $\cB(z)$ and
$\cB_{\tau_q}(z)$ respectively and suppose
$\mathcal C_{\tau_q}:Y\to\widetilde Y$ is a bounded isomorphism.  Assume
$\Ran\mathsf S_X(z)\subset Y$ and
$\mathsf R_X(z)|_Y:Y\to\mathbb C^m$ is bounded.  Suppose the response family and both source and readout families are analytic in $z$.  In the remainder of
the proposition $\cB(z)$ and $\cB_{\tau_q}(z)$ denote their restrictions to
$Y$ and $\widetilde Y$ respectively.  Define
\begin{equation}
 \begin{aligned}
 \cB_{\rm ind}(z)&:=\mathcal C_{\tau_q}\cB(z)
                    \mathcal C_{\tau_q}^{-1},\\
 \Delta\cB(z)&:=\cB_{\tau_q}(z)-\cB_{\rm ind}(z)\\
 &=-\bigl[D_XE_{\tau_q}(X_*(z);z)|_Y\bigr]
       \mathcal C_{\tau_q}^{-1}.
 \end{aligned}
 \label{eq:induced-kernel-defect}
\end{equation}
The consistently transported source and readout are
\begin{equation}
 \mathsf S_{\rm ind}(z):=\mathcal C_{\tau_q}\mathsf S_X(z):
       \mathbb C^m\longrightarrow\widetilde Y,
 \qquad
 \mathsf R_{\rm ind}(z):=
       \bigl(\mathsf R_X(z)|_Y\bigr)\mathcal C_{\tau_q}^{-1}:
       \widetilde Y\longrightarrow\mathbb C^m.
 \label{eq:induced-source-readout}
\end{equation}
They preserve the unfiltered response represented on the comparison sector
\begin{equation}
 \mathsf R_{\rm ind}(I_{\widetilde Y}-\cB_{\rm ind})^{-1}
 \mathsf S_{\rm ind}
 =\bigl(\mathsf R_X|_Y\bigr)
  [I_Y-(\cB|_Y)]^{-1}\mathsf S_X
 \label{eq:transported-response-invariance}
\end{equation}
wherever the two inverses exist.  Let
\begin{equation}
 \widetilde{\mathsf S}_X(z):\mathbb C^m\to\widetilde Y,
 \qquad
 \widetilde{\mathsf R}_X(z):\widetilde Y\to\mathbb C^m
 \label{eq:actual-coarse-source-readout-types}
\end{equation}
be the source and readout of the coarse response evaluated on $\bar X_*(z)$.

Let $K\Subset U$ and let a positively oriented contour $\gamma$ in the
spectral $\zeta$ plane enclose exactly one algebraically simple eigenvalue of
$\cB_{\rm ind}(z)$ and no other spectrum for every $z\in K$.  Put
\begin{equation}
 \begin{aligned}
 \epsilon&:=\sup_{z\in K}\|\Delta\cB(z)\|
 \leq\sup_{z\in K}\|D_XE_{\tau_q}(X_*(z);z)|_Y\|
       \|\mathcal C_{\tau_q}^{-1}\|,\\
 K_\gamma&:=\sup_{z\in K,\,\zeta\in\gamma}
 \|(\zeta-\cB_{\rm ind}(z))^{-1}\|,\\
 L_\gamma&:=\operatorname{length}(\gamma).
 \end{aligned}
 \label{eq:closure-stability-constants}
\end{equation}
If $K_\gamma\epsilon<1$ the contour remains in the resolvent set of the
coarse family and its Riesz projection obeys
\begin{equation}
 \|\Pi_{\tau_q}(z)-\Pi_{\rm ind}(z)\|
 \leq d_\Pi:=
 \frac{L_\gamma}{2\pi}
 \frac{K_\gamma^2\epsilon}{1-K_\gamma\epsilon}.
 \label{eq:riesz-defect-bound}
\end{equation}
Writing
$B_*:=\sup_{z\in K}\|\cB_{\rm ind}(z)\|$ and
$\Pi_*:=\sup_{z\in K}\|\Pi_{\rm ind}(z)\|$ the enclosed eigenvalues satisfy
\begin{equation}
 |\lambda_{\tau_q}(z)-\lambda_{\rm ind}(z)|
 \leq d_\lambda:=
 \epsilon(\Pi_*+d_\Pi)+2B_*d_\Pi.
 \label{eq:eigenvalue-defect-bound}
\end{equation}

Suppose $\lambda_{\rm ind}(z_*)=1$ and
$\alpha:=|\partial_z\lambda_{\rm ind}(z_*)|>0$.  For sufficiently small
$d_\lambda>0$ assume the disk
$D_*:=\{z:|z-z_*|\leq4d_\lambda/\alpha\}$ lies in $K$ that $z_*$ is the
only zero of $\lambda_{\rm ind}-1$ in $D_*$ and that
\begin{equation}
 |\lambda_{\rm ind}(z)-1|
 \geq\frac{\alpha}{2}|z-z_*|,
 \qquad z\in\partial D_*.
 \label{eq:unperturbed-root-lower-bound}
\end{equation}
Then $\lambda_{\tau_q}-1$ has exactly one zero $\widetilde z_*$ in $D_*$
and
\begin{equation}
 \delta_z:=|\widetilde z_*-z_*|
 \leq\frac{4d_\lambda}{\alpha}.
 \label{eq:critical-root-displacement}
\end{equation}
When $d_\lambda=0$ take $\widetilde z_*=z_*$.

For $z\in D_*$ set
\begin{equation}
 \begin{aligned}
 \mathcal N_{\rm ind}(z)
  &:=\mathsf R_{\rm ind}(z)\Pi_{\rm ind}(z)\mathsf S_{\rm ind}(z),\\
 \widetilde{\mathcal N}(z)
  &:=\widetilde{\mathsf R}_X(z)\Pi_{\tau_q}(z)
      \widetilde{\mathsf S}_X(z).
 \end{aligned}
 \label{eq:closure-residue-numerators}
\end{equation}
Introduce the same-point map discrepancies and uniform norms
\begin{equation}
 \begin{gathered}
 d_R:=\sup_{D_*}\|\widetilde{\mathsf R}_X-\mathsf R_{\rm ind}\|,
 \qquad
 d_S:=\sup_{D_*}\|\widetilde{\mathsf S}_X-\mathsf S_{\rm ind}\|,\\
 R_*:=\sup_{D_*}\|\mathsf R_{\rm ind}\|,
 \quad \Pi_*^{D}:=\sup_{D_*}\|\Pi_{\rm ind}\|,
 \quad
 \widetilde\Pi_*:=\sup_{D_*}\|\Pi_{\tau_q}\|,
 \quad
 \widetilde S_*:=\sup_{D_*}\|\widetilde{\mathsf S}_X\|,
 \end{gathered}
 \label{eq:closure-map-error-constants}
\end{equation}
and define
\begin{equation}
 d_{\rm num}^{(0)}
 :=d_R\widetilde\Pi_*\widetilde S_*
   +R_*d_\Pi\widetilde S_*
   +R_*\Pi_*^{D}d_S.
 \label{eq:same-point-numerator-defect}
\end{equation}
Analyticity on $D_*$ gives finite constants
\begin{equation}
 L_{\rm num}:=\sup_{z\in D_*}
 \|\partial_z\widetilde{\mathcal N}(z)\|,
 \quad
 d_{\lambda'}^{(0)}:=\sup_{z\in D_*}
 |\partial_z\lambda_{\tau_q}(z)
  -\partial_z\lambda_{\rm ind}(z)|,
 \quad
 L_{\lambda'}:=\sup_{z\in D_*}
 |\partial_z^2\lambda_{\tau_q}(z)|.
 \label{eq:cross-root-variation-constants}
\end{equation}
In particular $L_{\rm num}$ may be bounded by the product rule from the
separate $z$-derivative bounds for $\widetilde{\mathsf R}_X$ and $\Pi_{\tau_q}$ together with $\widetilde{\mathsf S}_X$.  Put
\begin{equation}
 D_{\rm num}:=d_{\rm num}^{(0)}+L_{\rm num}\delta_z,
 \qquad
 D_{\lambda'}:=d_{\lambda'}^{(0)}+L_{\lambda'}\delta_z.
 \label{eq:cross-root-total-errors}
\end{equation}
If $D_{\lambda'}<\alpha/2$ the residues at their respective critical roots
\begin{equation}
 Z_{{\rm alg},{\rm ind}}
 :=-\frac{\mathcal N_{\rm ind}(z_*)}
          {\partial_z\lambda_{\rm ind}(z_*)},
 \qquad
 \widetilde Z_{\rm alg}
 :=-\frac{\widetilde{\mathcal N}(\widetilde z_*)}
          {\partial_z\lambda_{\tau_q}(\widetilde z_*)},
 \label{eq:cross-root-residue-definitions}
\end{equation}
obey
\begin{equation}
 \|\widetilde Z_{\rm alg}-Z_{{\rm alg},{\rm ind}}\|
 \leq \frac{2D_{\rm num}}{\alpha}
       +\frac{2\|\mathcal N_{\rm ind}(z_*)\|}{\alpha^2}
        D_{\lambda'}.
 \label{eq:residue-defect-bound}
\end{equation}
Thus the response estimate uses reference-branch matching and a bounded comparison inverse together with nonnormal resolvent control and transversality.
The residue estimate additionally transports the numerator and eigenvalue
slope between the two actual critical roots.
\end{proposition}

The proof is given in Appendix~\ref{app:kernel-proof}.  Condition
\eqref{eq:reference-branch-matching} is substantive.  If it fails an actual
coarse branch $\widetilde X_*(z)$ and an independent bound
\begin{equation}
 d_X:=\sup_{z\in K}
 \|\widetilde X_*(z)-\mathcal C_{\tau_q}X_*(z)\|
 \label{eq:coarse-fixed-point-displacement}
\end{equation}
must first be established.  Provided the reevaluated families preserve the
same comparison spaces re-evaluate $\cB_{\tau_q}$ and every coarse source or
readout map at $\widetilde X_*$.  If
$D_X\mathfrak F_{\tau_q}$ is $L_B$-Lipschitz in $X$ the kernel bound then
obeys
\begin{equation}
 \epsilon\leq
 \sup_{z\in K}\|D_XE_{\tau_q}(X_*(z);z)|_Y\|
 \|\mathcal C_{\tau_q}^{-1}\|+L_Bd_X.
 \label{eq:displaced-base-point-kernel-budget}
\end{equation}
Base-point-dependent source and readout maps acquire their corresponding
Lipschitz displacement terms.
At a critical response the residual
$\|E_{\tau_q}(X_*;z)\|$ alone does not supply such a displacement bound
because $I-\cB_{\tau_q}$ need not be invertible.

The word ``flow'' is used for several inequivalent operations in field theory.
The Gaussian maps $\cH_\tau$ form a heat semigroup ordered by resolution.
A Wilsonian or functional RG trajectory orders effective descriptions by a cutoff scale $k$.
Physical time parametrises actual evolution through distinguishable
calibrated clock readings.  Hamiltonian evolution and open-system
dynamical maps provide field-theoretic realisations of that evolution.
In the generated event carrier of
ref.~\cite{Guo2026CompactSpectrumI} Section~10 a timelike translation
changes a matter phase at a rate fixed by its energy. Reading the two
phase quadratures calibrates elapsed time. Gaussian variance orders
resolution while this calibrated phase orders the physical evolution.

\subsection{The Gaussian scale semigroup}
\label{sec:scale-semigroup}

Equation~\eqref{eq:heat-multiplier} gives
\begin{equation}
 \cH_0=I,
 \qquad
 \cH_{\tau+s}=\cH_\tau\cH_s,
 \qquad
 \partial_\tau\cH_\tau=\Delta\cH_\tau.
 \label{eq:heat-semigroup-generator}
\end{equation}
The parameter $\tau$ has dimensions of length squared and labels the amount of smoothing.
Increasing $\tau$ orders resolution by stronger attenuation at each fixed nonzero wave number.

For every finite $\tau$ the Fourier multiplier $e^{-\tau|q|^2}$ is pointwise nonzero.
Consequently $\cH_\tau$ is injective on $L^2(\mathbb R^d)$.
Its inverse on the range is unstable for $\tau>0$.

\paragraph{Operational irreversibility.}
For every $\tau>0$ the inverse of $\cH_\tau$ on
$\Ran\cH_\tau\subset L^2(\mathbb R^d)$ is unbounded.
Hence arbitrarily small errors in coarse-grained high-frequency data can correspond to order-one errors before coarse-graining.
Choose $f_n\in L^2$ with $\|f_n\|_2=1$ and Fourier support in the annulus
$n\leq|q|\leq n+1$.
Plancherel's theorem gives
\begin{equation}
 \|\cH_\tau f_n\|_2
 \leq e^{-\tau n^2}\|f_n\|_2
 \xrightarrow[n\to\infty]{}0.
 \label{eq:unbounded-inverse}
\end{equation}
If the inverse were bounded then
the inequality
$1=\|f_n\|_2\leq\|\cH_\tau^{-1}\|\,\|\cH_\tau f_n\|_2$
would contradict Eq.~\eqref{eq:unbounded-inverse}.

This supplies the precise sense in which Gaussian smoothing is
operationally irreversible at finite precision.  Algebraically the map
remains injective. The loss of recoverability is the instability of its
inverse.  A partial trace is a separate state-level operation.

This is the familiar obstruction in deblurring.  Every finite Fourier mode survives while sufficiently fine structure falls below any
fixed observational tolerance and inversion amplifies that uncertainty
without bound.

\subsection{Wilsonian integration and information loss}
\label{sec:wilsonian-information}

Wilsonian integration produces an effective functional for retained
variables by integrating over the eliminated momentum modes
\cite{Wilson:1974rg}.
At the level of quantum states discarding a factor of a specified Hilbert
space acts by partial trace. Restricting observations can instead be
described by a conditional expectation onto an observable algebra.
For density matrices a completely positive trace-preserving reduction
contracts their quantum relative entropy~\cite{Lindblad:1975}.
In an interacting field theory identifying such factors or algebras is
part of the physical realisation of the reduction.

For a finite carrier with $A=A^*$ the random-phase channel
\[
\mathcal E_v(\rho)=\mathbb E_{\theta\sim N(0,v)}
[e^{-\mathrm i\theta A}\rho e^{\mathrm i\theta A}]
\]
is completely
positive and trace preserving. In an eigenbasis its multiplier is
$e^{-v(a_i-a_j)^2/2}$.  Random phases leave the populations unchanged and
attenuate coherence between different eigenvalues. The squared eigenvalue
difference determines the decay.  Field-amplitude smoothing instead acts
on each mode with $e^{-\tau a_i}$. Composition adds the variances.

For a Euclidean field representation reflection positivity tests whether
positive-time insertions define nonnegative state norms after reflection
\cite{OsterwalderSchrader:1973,OsterwalderSchrader:1975}.
This question involves the covariance across time slices. Let $A\geq0$ be self-adjoint and let $m>0$. Put $H=(A+m^2)^{1/2}$ and let $B$ be a
bounded time-independent operator on each time slice. For
$C_B(t)=B e^{-|t|H}(2H)^{-1}B^*$ and $t_i\geq0$
\[
\sum_{i,j}\langle f_i,C_B(t_i+t_j)f_j\rangle
=\left\|\sum_i e^{-t_iH}(2H)^{-1/2}B^*f_i\right\|^2\geq0.
\]
Each positive-time insertion thus contributes a vector to a common
Hilbert space and the reflected quadratic form is their Gram norm.
The mass gap bounds the inverse square root. $B$ need not commute with
$H$. This proves reflection positivity of the declared free covariance.
The time Gaussian $C(t)=e^{-t^2}$ fails the test at $t_1=1/4,t_2=3/4$ since its reflected matrix has determinant
$e^{-5/2}-e^{-2}<0$. Finite-dimensional evaluations of these covariances
and of the random-phase channel use the software of
Ref.~\cite{Guo2026CoarseGrainingFramework}.

The Gaussian averaging semigroup and an FRG regulator serve different roles.
The former is the canonical averaging prescription within the composition
class of Proposition~\ref{prop:gaussian-unique} and is the one used in
Theorem~\ref{thm:fixed-transfer}. The latter is a computational term chosen
to interpolate between effective actions.  Different admissible regulators
can be used to estimate a channel but regulator variation measures the uncertainty of a finite approximation.

\subsection{Composite-source functional flow}
\label{sec:composite-frg}

The dynamical kernel in Eq.~\eqref{eq:inhomogeneous-bse} can be computed with several nonperturbative methods.
For FRG calculations the operator block must be coupled to sources before the flow is differentiated.
Let $h_i$ couple to the renormalised composite insertions and introduce
\begin{equation}
 Z_k[J,h]=\int\!\mathcal D\phi\,
 \exp\left[-S[\phi]-\Delta S_k[\phi]
 +J\!\cdot\!\phi+\sum_i h_i\!\cdot\![\cO_i]_R\right],
 \label{eq:composite-generating-functional}
\end{equation}
where
$\Delta S_k=\tfrac12\phi\!\cdot\!R_k\!\cdot\!\phi$
is written schematically with the superfield generalisation understood.
The modified Legendre transform is taken in $J$ while retaining $h$ as
an external source~\cite{Wetterich:1993,FloerchingerWetterich:2009}.
For bosonic variables write $W_k=\log Z_k$ and $\varphi=\delta_JW_k$ and take $R_k$ independent
of $h$.  On a twice differentiable local Legendre branch
\[
 \begin{aligned}
 \Gamma_k[\varphi,h]&=J\cdot\varphi-W_k[J,h]
                     -\tfrac12\varphi\cdot R_k\varphi,\\
 \Gamma_{k,\varphi}+R_k\varphi&=J,
 \qquad W_{k,h}=-\Gamma_{k,h}.
 \end{aligned}
\]
At fixed $J$ where $L_k:=\Gamma_{k,\varphi\varphi}+R_k$ has a
bounded inverse on the chosen response spaces differentiation gives
\begin{equation}
 \begin{aligned}
 \varphi_h&=-L_k^{-1}\Gamma_{k,\varphi h},\\
 W_{k,hh}&=-\Gamma_{k,hh}
       +\Gamma_{k,h\varphi}L_k^{-1}\Gamma_{k,\varphi h}.
 \end{aligned}
 \label{eq:effective-action-complete-response}
\end{equation}
The first term is the change of the composite expectation at fixed mean
field.  The second is the response transmitted by adjustment of that
field.  In the notation $D+RL^{-1}S$ the same action gives $D=-\Gamma_{k,hh}$ and $S=-\Gamma_{k,\varphi h}$ with $R=-\Gamma_{k,h\varphi}$. Thus the source and readout maps together with the direct response are
fixed together by the generating functional.

The choice of state variables determines where a collective singularity
appears in this identity.  For an even composite in a symmetric phase
for example the mixed derivative with an odd fundamental mean field
vanishes at the symmetric state and $-\Gamma_{k,hh}$ carries the composite
response.  Including composite expectation values or propagators among
the response variables exposes their adjustment in the inverse operator
as in composite-field and 2PI formulations
\cite{FloerchingerWetterich:2009,CornwallJackiwTomboulis:1974}.
Application of the isolated-pole theorem then uses the analyticity of
the direct and complementary terms in that particular representation.

At fixed composite sources the effective action obeys the Wetterich equation
\cite{Wetterich:1993}
\begin{equation}
 \partial_t\Gamma_k[\varphi,h]
 =\frac12\operatorname{STr}
 \left[(\Gamma_k^{(2)}[\varphi,h]+R_k)^{-1}
 \partial_tR_k\right],
 \qquad t=\ln(k/k_0).
 \label{eq:wetterich-composite}
\end{equation}
Functional derivatives with respect to $h_i$ generate flows for composite
insertions including their mixing and contact terms.  A formulation with
explicit composite fields and scale-dependent bosonisation is given in
Ref.~\cite{FloerchingerWetterich:2009}.
Derivatives with respect to both $h_i$ and the fundamental fields generate
the vertices entering a projected response equation.  After its
self-consistent variables $X$ have been specified that equation defines a
map $\mathfrak F_k(X;z)$ and the exact response operator is its Jacobian
$D_X\mathfrak F_k$ as in Eq.~\eqref{eq:response-jacobian}.  In the 2PI
formalism the self-energy and the Bethe--Salpeter kernel are obtained from
successive functional derivatives of the same effective action
\cite{CornwallJackiwTomboulis:1974}. The renormalisation of the resulting
scalar two- and higher-point functions is treated in
Ref.~\cite{Berges:2005}.  This gives an equivalent
response representation only when the source pairing and propagators agree together with the renormalisation conditions and continuation convention with those
used to define $\mathfrak F_k$.  A selected list of diagrams does not by
itself establish that equivalence.

Equation~\eqref{eq:wetterich-composite} is exact but not closed on a finite collection of vertices.
A practical computation specifies an operator basis and field expansion together with momentum resolution and a regulator. Standard truncation schemes and
their applications are reviewed in Ref.~\cite{Dupuis:2021frg}.
The resulting $\cB_{k,N,\mathcal R}$ is a truncated response operator.  Its
omitted-vertex and channel-closure errors contribute to the derivative of
the closure defect in Eq.~\eqref{eq:coarse-graining-closure-defect}.
An eigenvalue reaching one in that operator is evidence within the stated
implementation. Stability under increasing $N$ and improving momentum resolution together with variation of $R_k$ is required before extrapolating to the exact
channel.
The retained vertices and their closure specify which contributions enter
each finite approximation.

The parameter $t=\ln(k/k_0)$ in Eq.~\eqref{eq:wetterich-composite} is an RG scale coordinate and orders features along a family of effective descriptions.
Lorentzian formation is described through a real-time evolution law with specified initial data.

%% file: sections-en/sec05_synergy.tex
\section{Generated examples and independent response tests}
\label{sec:synergy}

The examples follow a disturbance from its source through the unresolved
variables to the observed response.  An arithmetic merge first shows how
those variables change a static susceptibility and retain a relaxation
history.  A reflecting-heat update then makes the collective feedback and critical mode together with the memory error explicit.  In the field-theory examples
the mass spectrum determines how the response extends through space
while symmetry determines which source couplings are allowed.

\subsection{An arithmetic fibre: merge, environment, and static response}
\label{sec:arithmetic-example}

This example separates preservation of a static source response from
closure of the underlying update.  The eliminated variables can be
accounted for in the static response even though they must be retained
or represented by memory for evolution.

The arithmetic merge construction of Ref.~\cite{Guo:2026arithmetic}
on multiplicative and additive closure motivates the ordered fixed-sum
fibre used here.  Relative to the finite-fibre expression resolution in
that reference the present fibre is its two-input additive value-level compression. It retains the ordered evaluated inputs while discarding their
internal expression histories.  The finite data follow directly from the ordered input basis.  For
$k\geq2$ take the orthonormal basis $|m,k-m\rangle$ for $1\leq m\leq k-1$
and define $N_1,N_2$ by their displayed degrees.  The total degree is
$N_+=kI$ on this fibre. The cross-term of its quadratic cost is
\begin{equation}
 K_+=N_+^2-N_1^2-N_2^2=2N_1N_2
   =\diag\{2m(k-m):1\leq m\leq k-1\}.
 \label{eq:arithmetic-cost}
\end{equation}
All these equalities follow on the basis.  They specify the cost of the arithmetic operation.

Write $C_k|m,k-m\rangle=c_m$.  Invariance under every permutation of
this basis forces equal coefficients and $C_kC_k^*=1$ forces
$|c_m|=(k-1)^{-1/2}$.  Requiring positive coefficients fixes the target
phase.  This proves locally the naturality of the normalised merge
\begin{equation}
 C_k=\frac{(1,\ldots,1)}{\sqrt{k-1}}.
 \label{eq:arithmetic-merge}
\end{equation}
A faithful orthogonal dilation $|m,k-m\rangle\mapsto|k\rangle\otimes f_m$
requires $k-1$ orthonormal environment vectors.  That lower bound is attained
by retaining the ordered input label and all such minimal orthogonal
dilations are unitarily equivalent. This dimension counts the records
required to distinguish all ordered inputs. A future observation
protocol specifies which of those distinctions its outputs retain.

Since $N_+^2=k^2I$ the heat family obeys the exact descent relation $C_ke^{-tN_+^2}=e^{-tk^2}C_k$.  The connected cost need not descend through
the same merge.  For $k=4$ it gives
\begin{equation}
 K_+=\diag(6,8,6),\qquad C=\frac{(1,1,1)}{\sqrt3},\qquad
 Ce^{-tK_+}(1,-1,0)^{\mathsf T}
 =\frac{e^{-6t}-e^{-8t}}{\sqrt3}\ne0\quad(t>0).
 \label{eq:arithmetic-descent-failure}
\end{equation}
The input vector lies in $\Ker C$ so
Eq.~\eqref{eq:fibre-consistency} fails.  The compressed family $f(t)=Ce^{-tK_+}C^*=(2e^{-6t}+e^{-8t})/3$ satisfies
$f''(0)=136/3\ne400/9=f'(0)^2$ whereas a differentiable scalar semigroup
with $f(0)=1$ has $f''(0)=f'(0)^2$.

To define an actual static source response use real amplitudes and the
declared cost and source pairing
\begin{equation}
 \Phi(X,J)=\tfrac12 X^{\mathsf T}K_+X-JCX,
 \qquad K_+X=C^*J,\qquad O(X)=CX.
 \label{eq:arithmetic-sourced-cost}
\end{equation}
Positive definiteness gives a unique minimiser.  Its derivative with
respect to the source gives
\begin{equation}
 \mathscr C(4)=CK_+^{-1}C^*
 =\tfrac13(\tfrac16+\tfrac18+\tfrac16)=\frac{11}{72},
 \qquad \mathcal S_+(4)=\mathscr C(4)^{-1}=\frac{72}{11}.
 \label{eq:arithmetic-static-response}
\end{equation}
The Schur stiffness $\mathcal S_+(4)$ differs from
$CK_+C^*=20/3$.  More generally
\[
 C_kK_+^{-1}C_k^*
 =\frac1{k-1}\sum_{m=1}^{k-1}\frac1{2m(k-m)}
 =\frac{H_{k-1}}{k(k-1)},\qquad H_j=\sum_{m=1}^j\frac1m.
\]
The identity $1/[m(k-m)]=(1/m+1/(k-m))/k$ proves the sum.
Minimising the quadratic cost at fixed $C_kX$ gives the reciprocal as
the Schur stiffness proving the required finite version of the exact
arithmetic Schur theorem.  Static elimination and the failure of
memoryless heat descent are therefore compatible results.
The effective stiffness measures the response after unresolved amplitudes
have adjusted to the applied source.  Their distinct relaxation rates
still affect transient observations so this static coefficient alone
does not determine the coarse evolution.

For a concrete environment calculation at $k=4$ put
\[
 v=(1,1,1)/\sqrt3,\qquad b=(1,-2,1)/\sqrt6,\qquad
 a=(1,0,-1)/\sqrt2.
\]
In this basis
\begin{equation}
 K_+=\begin{pmatrix}20/3&-2\sqrt2/3&0\\
                    -2\sqrt2/3&22/3&0\\0&0&6\end{pmatrix}.
 \label{eq:arithmetic-environment-block}
\end{equation}
The $a$ coordinate is invisible to $C=v^*$ under the fixed heat family
and affine source injections along $v$.  The observational equivalence is equality of $(X_1+X_3,X_2)$. The two distinct exponentials
$e^{-6t},e^{-8t}$ separate these two coordinates.  They give a two-state
realisation.  A one-state linear realisation cannot produce both rates. Equivalently its Hankel determinant would vanish whereas
$f(0)f''(0)-f'(0)^2=8/9>0$.  Thus two states or one retained coordinate
and one residual are minimal in the linear class for this restricted
protocol.  Allowing all fibre permutations changes the quotient. Permuting $a$ before heat action can make it visible and the observation
rows $Ce^{-tK_+}P$ span all three coordinates for fixed $t>0$.  The full
three-state information must then be retained.  The minimal observation
space thus depends on the allowed protocol whereas the faithful
orthogonal environment preserves all input labels.

In the two visible coordinates $(x,u)$ the heat equation has
$\dot x=-20x/3+(2\sqrt2/3)u$ and
$\dot u=(2\sqrt2/3)x-22u/3$.  Hence
\[
 u(t)=e^{-22t/3}u(0)+\frac{2\sqrt2}{3}
       \int_0^t e^{-22(t-s)/3}x(s)\,ds.
\]
The retained equation contains the initial residual and the memory kernel
$(8/9)e^{-22t/3}$.  The parameter $t$ orders the arithmetic heat operation. Equation~\eqref{eq:arithmetic-sourced-cost} specifies its cost and source pairing.

\subsection{A reflecting-window protocol with a visible critical mode}
\label{sec:neumann-example}

The mechanism combines spatial smoothing with an inhomogeneous modulation.
Smoothing attenuates a nonconstant mode. Modulation couples it to the
constant mode allowing a disturbance to return through the unresolved
coordinate.  We construct this update on a reflecting window and follow
the resulting feedback through its static response and relaxation together with the induced memory.

The interval can be connected to the same Gaussian operation used in
Section~\ref{sec:criteria}.  The finite-window construction in Section~3 of
Ref.~\cite{Guo:2026arithmetic} uses a calibrated local information level
and a global peak-equivalent capacity.  The endpoint calculation fixes the window on which the boundary closure is imposed.  The entropy calculation uses the
continuous-density form of R\'enyi's entropy family
(Ref.~\cite{Renyi:1961} Eq.~(1.21)) with natural logarithms and a fixed
Lebesgue reference measure. The endpoint rules are specified separately.

\paragraph{Calibration of the Gaussian correlation window.}
Fix a displacement unit and Lebesgue measure in that unit.  For $b>0$
the one-dimensional heat kernel of Eq.~\eqref{eq:heat-kernel} gives the
unit-peak profile
\begin{equation}
 K_b(d)=\frac{g_{b^2/2}(d)}{g_{b^2/2}(0)}
       =e^{-d^2/(2b^2)}.
 \label{eq:window-gaussian-source}
\end{equation}
Let $\mathcal I:(0,1]\to[0,\infty)$ be continuous and additive under products and calibrated by $\mathcal I(e^{-1})=1$.  Define the local
endpoint by unit information loss and the global endpoint by complete
peak-equivalent capacity
\begin{equation}
 \ell_{\rm UV}:=\inf\{d>0:\mathcal I(K_b(d))\geq1\},\qquad
 L_{\rm IR}:=\frac{\int_{\mathbb R}K_b(d)\,dd}{K_b(0)}.
 \label{eq:window-endpoint-rules}
\end{equation}
Then these endpoints are unique and
\begin{equation}
 \ell_{\rm UV}=\sqrt2\,b,\quad L_{\rm IR}=\sqrt{2\pi}\,b,\quad
 \Delta_b=L_{\rm IR}-\ell_{\rm UV}
          =b(\sqrt{2\pi}-\sqrt2)>0.
 \label{eq:window-endpoint-values}
\end{equation}
For $p_b=K_b/\int_{\mathbb R}K_b$ the same normalization gives
$L_{\rm IR}=\exp H_\infty^{\rm R}(p_b)$ where
$H_\infty^{\rm R}(p_b)=-\log\|p_b\|_\infty$ is the infinite-order
R\'enyi entropy.
The function $u\mapsto\mathcal I(e^{-u})$ is continuous and additive
on $[0,\infty)$ hence linear. The unit calibration makes it $u$.
Thus $\mathcal I(K_b(d))=d^2/(2b^2)$ strictly increasing for $d>0$
which gives the unique local endpoint.  The Gaussian integral gives
$\int_{\mathbb R}K_b=\sqrt{2\pi}\,b$ proving the global endpoint
and window length.  More explicitly for $\alpha>1$
\[
 \int_{\mathbb R}p_b^\alpha\,dd
 =\frac{(\sqrt{2\pi}\,b)^{1-\alpha}}{\sqrt\alpha},\qquad
 H_\alpha^{\rm R}(p_b)
 =\log(\sqrt{2\pi}\,b)+\frac{\log\alpha}{2(\alpha-1)}.
\]
Taking $\alpha\to\infty$ proves the asserted entropy characterization.

These endpoint rules use the full signed displacement line for capacity. There is no extra factor of one half.  The resulting global carrier
$[0,L_{\rm IR}]$ represents the peak-equivalent mass at unit height.
Resolving its part beyond unit local information gives the active window
$[\ell_{\rm UV},L_{\rm IR}]$.  In the information coordinate its
endpoints are $1$ and $\pi$.  The declared calibration and capacity rule fix this length uniquely. A different
information unit or endpoint rule would be different data.  In particular
Shannon entropy would give the different Gaussian width
$\sqrt{2\pi e}\,b$.  Fixing the reference length unit above is essential
when using a differential entropy as a width.

\paragraph{Reflecting closure and the operational scale.}
Translate the active window by $x=d-\ell_{\rm UV}$ to
$[0,\Delta_b]$.  Among separated homogeneous self-adjoint boundary
conditions for the local Laplacian preserving constants selects Neumann closure. An endpoint condition $a f+c f'=0$ with $(a,c)\ne(0,0)$ must
have $a=0$ because $f=1$ is in the domain.  Consequently $f'=0$ at each
endpoint giving zero diffusive flux there.  Within this boundary class
preservation of a uniform profile therefore fixes the closure.  The
Neumann heat operator fixes the constant function and conserves mass
by self-adjointness.

For the common-source realisation set $\Delta=\Delta_b$.  The calculations
below also hold for an independently specified $\Delta>0$ with the same
reflecting closure.  Our convention is $H_\tau=e^{-\tau A}$. The
normalisation $e^{-tA/2}$ in Ref.~\cite{Guo:2026arithmetic} has $\tau=t/2$.
Its endpoint scale convention
$t(d)=d^2/2$ therefore gives $\tau_{\rm UV}=b^2/2$ and
$\tau_{\rm IR}=\pi b^2/2$.  The protocol may choose either or any declared
positive $\tau$. These parameters specify the averaging scale.

Let $\mathcal H=L^2(0,\Delta)$.  The nonnegative Neumann operator is
$A=-\partial_x^2$ with domain
$\{f\in H^2(0,\Delta):f'(0)=f'(\Delta)=0\}$.  Its eigenfunctions and
eigenvalues are
\begin{equation}
 e_0=\Delta^{-1/2},\quad
 e_n=\sqrt{2/\Delta}\cos(n\pi x/\Delta),\quad
 E_n=(n\pi/\Delta)^2\quad(n\geq1),\qquad E_0=0.
 \label{eq:neumann-generated-spectrum}
\end{equation}
Solving $-e''=Ee$ with the two boundary conditions gives these modes. Even extension to a $2\Delta$-periodic function and Fourier completeness
give a complete orthonormal basis.  Integration by parts gives
$\langle f,Af\rangle=\|f'\|_2^2$.  Equivalently the cosine diagonal
realisation is self-adjoint with domain $\sum E_n^2|f_n|^2<\infty$.
Thus $H_\tau=e^{-\tau A}$ is a bounded positive self-adjoint contraction
with $H_\tau e_n=e^{-\tau E_n}e_n$.  Its kernel is
\[
 h_\tau(x,y)=\sum_{j\in\mathbb Z}
 [g_\tau(x-y+2j\Delta)+g_\tau(x+y+2j\Delta)],
\]
the even periodisation of the line Gaussian.  Fourier coefficients give
the displayed Neumann eigenvalues and multiplier.  This specifies the
reflection operation. Neither a finite interval nor reflection truncates
its infinite spectrum.  Line increments define the heat law. Reflection imposes the boundary closure and a finite spectral projection selects retained modes.

Fix $\tau>0$ and $\varepsilon\in\mathbb R$ with $0<|\varepsilon|<1$. Let $q(x)=\cos(\pi x/\Delta)$ and let $P_2$ project onto
$\operatorname{span}\{e_0,e_1\}$.  Compose the declared operations
\begin{equation}
 T=P_2H_\tau M_{1+\varepsilon q}H_\tau P_2
   =\begin{pmatrix}1&\gamma\\\gamma&r^2\end{pmatrix},
 \qquad r=e^{-\tau E_1},\qquad \gamma=\frac{\varepsilon r}{\sqrt2}.
 \label{eq:two-mode-generated-update}
\end{equation}
The multiplication identities
$\langle e_0,qe_0\rangle=\langle e_1,qe_1\rangle=0$ and
$\langle e_0,qe_1\rangle=1/\sqrt2$ follow by elementary cosine
integration.  Heat action contributes the two diagonal factors $(1,r)$.
The multiplier is bounded between $1-|\varepsilon|$ and
$1+|\varepsilon|$ so
$\langle v,Tv\rangle\geq(1-|\varepsilon|)\|H_\tau v\|^2>0$ on
this two-mode space.  The finite projection completes the specified
two-mode operation.  Its amplitudes are Hilbert vectors with positivity
understood in the quadratic-form sense just proved.  The factor $r$
measures attenuation of the nonconstant profile in each heat step while
$\gamma$ measures its coupling to the uniform profile through the
spatial modulation.

Define the actual discrete update and observation by
\begin{equation}
 X_{n+1}=\eta TX_n+e_0h_n,\qquad O_n=e_0^*X_n,
 \qquad \eta\geq0.
 \label{eq:two-mode-sourced-dynamics}
\end{equation}
In the common-source specialization $\Delta$ is fixed by $b$ through
Eq.~\eqref{eq:window-endpoint-values}. Otherwise it is a
declared window parameter.  The remaining $\tau,\varepsilon,\eta$ and
the source pairing specify the protocol.  The window length fixes $E_1$ and the heat scale fixes $r$. The modulation fixes $\gamma$ and $\eta$ sets
the overall feedback strength per update.  These relations connect the
Gaussian calibration and boundary closure to the observed response.

\paragraph{Failure of one-state descent and minimal residual.}
Write $X_n=(x_n,u_n)$ and set $a=\eta$ and $\beta=\eta\gamma$ with $d=\eta r^2$.
For $\beta\ne0$ the projection onto $x$ fails fibre consistency because
$x_{n+1}=a x_n+\beta u_n+h_n$.  Keeping $u$ closes the actual system and
\begin{equation}
 \begin{aligned}
 u_n&=d^nu_0+\beta\sum_{j=0}^{n-1}d^{n-1-j}x_j,\\
 x_{n+1}&=a x_n+h_n+\beta d^nu_0
          +\beta^2\sum_{j=0}^{n-1}d^{n-1-j}x_j.
 \end{aligned}
 \label{eq:two-mode-memory}
\end{equation}
These identities follow by induction including the empty sum at $n=0$.
If $0\leq d<1$ dropping lags $k\geq m$ gives the effective tail bound
\begin{equation}
 |\text{discarded memory}|\leq
 \frac{\eta^2\gamma^2d^m}{1-d}\sup_{j<n}|x_j|,
 \qquad |\text{initial residual}|=|\beta|d^n|u_0|.
 \label{eq:two-mode-memory-tail}
\end{equation}
The history bound is essential. Critical forcing does not have a uniformly
bounded history merely because $d<1$.

For zero initial state put $m_j=e_0^*(\eta T)^je_0$.  Direct multiplication
gives $m_0=1$ and $m_1=\eta$ with $m_2=\eta^2(1+\gamma^2)$ and
$m_0m_2-m_1^2=\eta^2\gamma^2>0$ when $\eta>0$.  In an $n$-state
linear time-invariant realisation a Hankel matrix factors as an
observability matrix times a reachability matrix so has rank at most $n$. This is the rank bound underlying linear minimal-realisation methods
\cite{HoKalman:1966}.
Two states are therefore necessary and sufficient for this transfer
function.  The one residual coordinate is minimal in that class.  At $\eta=0$ the update vanishes and the Hankel determinant is zero.  Independently observing $x$ and its next unforced update
recovers $u=(x_{n+1}-a x_n)/\beta$ so the observation-sufficient quotient
for $\beta\ne0$ separates both amplitude coordinates.

\paragraph{Actual stability, crossing, and response.}
The two eigenvalues obtained from the characteristic polynomial are
\begin{equation}
 \lambda_\pm=\frac{1+r^2\pm\sqrt{(1-r^2)^2+4\gamma^2}}2,
 \qquad \eta_*=\lambda_+^{-1}.
 \label{eq:two-mode-critical-parameter}
\end{equation}
Since $r^2-\gamma^2=r^2(1-\varepsilon^2/2)>0$ and the characteristic
polynomial is negative at $r^2$ and at $1$ one has
$0<\lambda_-<r^2<1<\lambda_+$.  Therefore
$\|\eta T\|=\eta\lambda_+<1$ for $0\leq\eta<\eta_*$.
In the unmodulated limit $\varepsilon=0$ the uniform mode has
$\lambda_+=1$ and the other mode decays by $r^2$ per application of $T$.
For nonzero modulation the two modes mix raising $\lambda_+$ above one
and lowering the critical gain $\eta_*$.  The critical direction is
therefore a collective profile of the full update.  Reflection supplies
the constant mode. Modulation and feedback determine when its coupled
response becomes critical.
The actual update is a strict contraction there. Constant forcing has a
unique attracting fixed point.  At $\eta_*$ one eigenvalue reaches one
with control derivative $\partial_\eta(\eta\lambda_+)=\lambda_+>0$. The other is strictly below one and $d_*:=\eta_*r^2<1$.  Above
$\eta_*$ the largest mode grows.  These stability properties belong to the specified update.

For exponential input $h_n=w^nh$ define the particular amplitude by
$X_n=w^nX$.  The transfer equation is
$(wI-\eta T)X=e_0h$.  Put $w=1+z$. This discrete transfer variable measures the frequency of the input sequence.  Differentiating the sourced update gives
\begin{equation}
 \mathscr C_\eta(z)=e_0^*[(1+z)I-\eta T]^{-1}e_0
 =\frac{1+z-\eta r^2}
 {(1+z-\eta)(1+z-\eta r^2)-\eta^2\gamma^2}.
 \label{eq:two-mode-transfer}
\end{equation}
Schur elimination gives the inverse scalar denominator
$1+z-\eta-\beta^2/(1+z-d)$ retaining the transformed hidden-state
propagator.  At $\eta_*$
\begin{equation}
 \mathscr C_{\eta_*}(z)=\frac{Z}{z}
 +\frac{1-Z}{z+1-\eta_*\lambda_-},\qquad
 Z=\frac{1-\eta_*r^2}{1-\eta_*\lambda_-}\in(0,1).
 \label{eq:two-mode-positive-residue}
\end{equation}
The inequalities for $\lambda_-$ prove the stated range of $Z$.
Here $L(z)=(1+z)I-\eta_*T$. Its complementary inverse is bounded by
$(1-\eta_*\lambda_-)^{-1}$ at zero and $L'(0)=I$.  Hence $G_0=I$ on
the one-dimensional critical space.  The $e_0$ source sees that space
since an eigenvector with zero first component would require $\gamma=0$.
This verifies every hypothesis of Theorem~\ref{thm:native-positive-residue}
from the actual operations.  The variable $z$ is the transfer coordinate of
this discrete update.  In the QFT realisation below invariant mass enters
through the corresponding spectral representation.
Here $e_0$ is the coarse uniform level and $e_1$ is the first spatial
contrast.  Heat attenuates the contrast. The inhomogeneous multiplier converts level into contrast and back while the gain $\eta$ closes the feedback loop.  The critical threshold therefore belongs to a collective
profile of the full update.
At criticality the unforced fixed points form the critical line. Constant
nonzero forcing accumulates along its visible direction so there is no
unique forced stationary state at that endpoint.

The same decomposition displays critical slowing in the actual discrete
update.  Let $P_+=(T-\lambda_-I)/(\lambda_+-\lambda_-)$ be the orthogonal
projector at $\lambda_+$ and $\mu_\pm=\eta\lambda_\pm$.
Then $e_0^*P_+e_0=Z$ so zero initial
state and constant input $h\ne0$ give for $0<\eta<\eta_*$
\begin{equation}
 \frac{x_n}{h}
 =Z\frac{1-\mu_+^n}{1-\mu_+}
  +(1-Z)\frac{1-\mu_-^n}{1-\mu_-}.
 \label{eq:two-mode-step-response}
\end{equation}
This follows from $X_n=\sum_{j=0}^{n-1}(\eta T)^je_0h$ and the two
orthogonal spectral projectors.  The slow contribution has relaxation
length $(-\log\mu_+)^{-1}\sim(1-\mu_+)^{-1}$ as $\eta\uparrow\eta_*$
in update steps while its stationary susceptibility is $Z/(1-\mu_+)$.
At criticality the first term becomes $Zn$ as the source accumulates in
the visible neutral direction.  Thus the transfer-pole residue is also
the weight of the slow response.  A physical relaxation time additionally
requires a calibration of these update steps.

The error from dropping the residual altogether can also be computed.
With $C=(1,0)$ and $B=\eta T-zI$ together with $\bar B=a-z$ and the inherited
scalar source and readout the four defects are
$E_B=(0,-\beta)$ and $E_S=E_R=E_D=0$.  On the common inverse domain
Theorem~\ref{thm:irreversible-error} gives
\begin{equation}
 \bar{\mathscr C}-\mathscr C
 =-\frac{\beta^2}{(w-a)[(w-a)(w-d)-\beta^2]},\qquad w=1+z.
 \label{eq:two-mode-projection-error}
\end{equation}
Direct subtraction of $1/(w-a)$ and
Eq.~\eqref{eq:two-mode-transfer} verifies the same identity.
Thus update closure failure alone produces a nonzero output error even with compatible source and readout maps and direct terms.

\paragraph{Schur reduction and finite-memory error.}
To expose terms hidden by a one-leg probe allow two source columns $S=I_2$
and two readouts $R=I_2$ with zero full-system direct term.  Let $w=1+z$ and $H=(w-d)^{-1}$ with $L_{\rm eff}=w-a-\beta^2H$.  Eliminating the second
state gives
\begin{equation}
 S_{\rm eff}=(1,\beta H),\quad R_{\rm eff}=(1,\beta H)^{\mathsf T},\quad
 D_{\rm eff}=\begin{pmatrix}0&0\\0&H\end{pmatrix},\quad
 \mathscr C=D_{\rm eff}+R_{\rm eff}L_{\rm eff}^{-1}S_{\rm eff}.
 \label{eq:two-mode-four-schur-terms}
\end{equation}
Solving the two linear equations proves the formula including the
off-diagonal observations and the nonzero direct term.  In the time
domain the second source also enters the hidden memory as in
Eq.~\eqref{eq:exact-hidden-history}.

For a finite-memory approximation use
$H_m(w)=\sum_{k=0}^{m-1}d^kw^{-k-1}$ in the complete source--response tuple.  On $|w|\geq w_0>d$
\begin{equation}
 |H-H_m|\leq\frac{d^m}{w_0^m(w_0-d)}=:\delta_m.
 \label{eq:two-mode-transform-tail}
\end{equation}
Thus the effective-kernel error is bounded by $\beta^2\delta_m$. Both source and readout errors are bounded by $|\beta|\delta_m$ and the direct-term error by $\delta_m$.  Use $B_{\rm eff}=I-L_{\rm eff}$ to
apply Theorem~\ref{thm:irreversible-error}.  On a compact invertible
domain with $\sup|L_{\rm eff}^{-1}|\leq M$ the condition
$M\beta^2\delta_m<1$ gives
$\sup|L_{{\rm eff},m}^{-1}|\leq M/(1-M\beta^2\delta_m)$ by the Neumann
series.  Substitution gives an explicit convergent complete-response
error bound.  A small contour around the critical root excluding the
second pole and $w=d$ then gives the total-residue bound of
Eq.~\eqref{eq:contour-output-residue-error}.  All tail bounds have supplied
geometric majorants. 

The pole displacement and residue error can be resolved more sharply at
the critical parameter.  The following consequence also distinguishes
approximation of the memory kernel from approximation of its derivative.

\begin{corollary}[Pole and residue under finite-memory truncation]
\label{cor:finite-memory-pole-rates}
Fix the two-mode model at $\eta=\eta_*$ and fix
$d=\eta_*r^2\in(0,1)$ as $m\to\infty$.  All $O(\cdot)$ constants below
may depend on this fixed model including $d$.  Let $m\geq1$ be an integer
and write $a=\eta_*$ and $\beta=\eta_*\gamma\ne0$.
In the transfer variable $w$ set
\begin{equation}
 L(w)=w-a-\frac{\beta^2}{w-d},\qquad
 L_m(w)=w-a-\beta^2H_m(w),\qquad
 Z=\frac1{L'(1)}.
 \label{eq:finite-memory-critical-denominators}
\end{equation}
Here $L(1)=0$ and $Z$ is the scalar residue in
Eq.~\eqref{eq:two-mode-positive-residue}.  For sufficiently large $m$ the function $L_m$ has exactly one zero $w_m$ in a fixed sufficiently small
neighbourhood of $1$.  It is real and simple.  The residue
$Z_m=\operatorname*{Res}_{w=w_m}L_m(w)^{-1}$ satisfies
\begin{align}
 w_m-1&=-\frac{Z\beta^2}{1-d}\,d^m+O(md^{2m}),
 \label{eq:finite-memory-root-rate}\\
 Z_m-Z&=\frac{Z^2\beta^2}{1-d}\,m d^m+O(d^m).
 \label{eq:finite-memory-residue-rate}
\end{align}
For the two-source response in Eq.~\eqref{eq:two-mode-four-schur-terms}
replace $H$ by $H_m$ in all four terms and denote its residue at $w_m$
by $\mathsf Z_m$.  With
$v=(1,\beta/(1-d))^{\mathsf T}$ and $\mathsf Z=Zvv^{\mathsf T}$
\begin{equation}
 \|\mathsf Z_m-\mathsf Z\|=O(md^m).
 \label{eq:finite-memory-matrix-residue-rate}
\end{equation}
The constants depend on the fixed value of $d$.  The control parameter remains $\eta_*$
and the residues are compared at their respective transfer poles.
\end{corollary}

\begin{proof}
The geometric sum gives the exact identity
\begin{equation}
 L_m(w)-L(w)=\frac{\beta^2d^m w^{-m}}{w-d},\qquad
 L_m(1)=\frac{\beta^2d^m}{1-d}.
 \label{eq:finite-memory-exact-defect}
\end{equation}
Choose a closed disk about $1$ that excludes $d$ and zero together with the second zero of $L$ and on which $|w|>d$.  The defect converges uniformly
there.  Rouch\'e's theorem gives exactly one zero $w_m$ counted with
multiplicity for large $m$. Reality of the coefficients makes it real.
To obtain the sharper rate put $t=d^m$.  On $|w-1|=Kt$ with any fixed
$K>Z\beta^2/(1-d)$ one has $w^{-m}=1+O(mt)$
\[
 L(w)=Z^{-1}(w-1)+O(t^2),\qquad
 L_m(w)-L(w)=\frac{\beta^2}{1-d}t+O(mt^2).
\]
Rouch\'e's theorem on this smaller circle yields $w_m-1=O(t)$.
Substitution in $L_m(w_m)=0$ then proves
Eq.~\eqref{eq:finite-memory-root-rate}.

Differentiating the defect at $1$ gives
\begin{equation}
 L_m'(1)-L'(1)
 =-\beta^2d^m\left(\frac{m}{1-d}+\frac1{(1-d)^2}\right).
 \label{eq:finite-memory-derivative-defect}
\end{equation}
On $|w-1|\leq Kt$ the second derivative of the defect is $O(m^2t)$
whereas $L''$ is bounded.  Thus transport from $1$ to $w_m$ gives
\[
 L_m'(w_m)=Z^{-1}-\frac{\beta^2}{1-d}\,mt
             +O(t+m^2t^2).
\]
Since $m^2t\to0$ taking the reciprocal proves
Eq.~\eqref{eq:finite-memory-residue-rate}.  For the full matrix response
$H_m(w_m)=H(1)+O(t)$ by the geometric remainder and the root estimate.
The direct term is analytic near $w_m$ and the residue equals
$Z_m v_m v_m^{\mathsf T}$ with
$v_m=(1,\beta H_m(w_m))^{\mathsf T}=v+O(t)$.
Equation~\eqref{eq:finite-memory-matrix-residue-rate} follows.
\end{proof}

The nonzero leading coefficient in
Eq.~\eqref{eq:finite-memory-residue-rate} shows that the factor $m$ is present in the asymptotic scalar residue error.
This is an explicit instance of the derivative and displaced-root terms
in Proposition~\ref{prop:closure-stability}. A memory defect of order
$d^m$ moves the pole by that order while its derivative changes the
residue by order $md^m$.
Locating the critical response and predicting its strength therefore
require different accuracies in the same finite-memory model.  The
extra derivative weights long lags more strongly which accounts for
the factor $m$ from differentiating the tail.

Ref.~\cite{Guo2026CoarseGrainingFramework} supplies supplementary finite-matrix
implementations of this reduction.  The response formula retains all four terms
with the row functional continued analytically in $z$.
An elementary matter example is the sixteen-dimensional operator
$D=\left(\begin{smallmatrix}0&rI_8\\rI_8&0\end{smallmatrix}\right)$ with $r>0$ and one right coordinate designated neutral. Its charged
compression has dimension fifteen. The operator sends the neutral
coordinate into its left partner so the neutral line is not invariant.
The charged matrix is therefore a compression. On domains where the
complementary block of the sourced inverse is invertible
Proposition~\ref{prop:response-closed-channel} makes direct inversion and
complete elimination of that coordinate identical. Excursions from the
retained left partner into the neutral coordinate and back enter the
effective kernel. Sources and observations
with neutral components also contribute to the effective source and readout maps together with the direct term. The actual full and charged carriers in
ref.~\cite{Guo2026CompactSpectrumI} obey the same operator-descent
condition described in Appendix~\ref{app:operator-mixing}.

\paragraph{Replicated matter and its quadratic source.}
For the companion compact-matter construction consider a field-theoretic
realisation with fixed family multiplicity.  Its one-family carrier
$\mathcal E_1$ is generated by the internal representation and covariant
transport of Ref.~\cite{Guo2026CompactSpectrumI}.  The common
family module in the companion application
Ref.~\cite{Guo2026CompactSpectrumII} is taken to be unitarily identified with
$\mathcal E_g=\mathbb C^g\otimes\mathcal E_1$.
At fixed common bosonic background with vanishing Higgs field the
regulated fermion integrals of the families are independent.  With the same source variation of the domain and measure together with the spectral projectors and regulator in every copy
\begin{equation}
 K_g[h]=I_g\otimes K_1[h],\qquad
 W_g[h]=gW_1[h],\qquad
 \partial_a\partial_bW_g
 =g\bigl(\langle S_{1,a}S_{1,b}\rangle_c
                -\langle S_{1,ab}\rangle\bigr).
 \label{eq:replicated-matter-source-response}
\end{equation}
The generating functionals are evaluated at zero fermion sources. Source-independent constants are immaterial. Their product makes the
logarithms add giving the second identity for the fermion contribution
at that background. The third
follows by twice differentiating the same Euclidean integral including
the measure identification and moving-domain contributions.
The common physical TT compression therefore has the same factor $g$.
This is the quadratic response used in the companion compact-endpoint
calculation of ref.~\cite{Guo2026CompactSpectrumII}. Within its
replication family equality of the matter and colour normalisations of
the same screening coefficient selects $g=3$. The one-generation charge
moments fix the torsion readout and the shell and three-position
conditions select $R=5$ in the declared integer response-level family.
This bandwidth specifies the source domain. Numerical mode enumeration
and quadrature evaluate the response on that domain. Higgs vertices and
a displaced matter background are incorporated by
differentiating their full action. This example fixes the direct term as
well as the propagated source terms in the classification.

\subsection{Scalar feedback and the fold value 4/27}
\label{sec:feedback-example}

For comparison with the transfer pole consider the explicitly specified
scalar feedback
\begin{equation}
 \Pi_{\rm d}=\frac{\Pi_0}{(1-V\Pi_{\rm d})^2},\qquad
 F_x(y)=\frac{x}{(1-y)^2},\quad y=V\Pi_{\rm d},\quad x=V\Pi_0,
 \label{eq:toy-feedback}
\end{equation}
on $0\leq y<1$ with $x\geq0$.  Its fixed points satisfy $x=y(1-y)^2$.
The derivative of this cubic is $(1-y)(1-3y)$ so it increases up to
$y=1/3$ and then decreases to zero.  Thus
\begin{equation}
 x_c=\frac4{27},\quad y_c=\frac13,
 \label{eq:toy-fourtwentyseven}
\end{equation}
and $0<x<x_c$ has exactly two admissible branches
$y_-\in(0,1/3)$ and $y_+\in(1/3,1)$.  At $x=0$ the only admissible
fixed point is zero. At $x_c$ the branches meet and above it there is no
fixed point in this domain.  At a fixed point
$F_x'(y)=2y/(1-y)$ giving local attraction of $y_-$ and repulsion of
$y_+$.  If $a<1/3$ and $0\leq x\leq a(1-a)^2$ then
$F_x([0,a])\subset[0,a]$ and
$\sup_{[0,a]}F_x'\leq2a/(1-a)<1$.  The contraction theorem gives a unique fixed point in this invariant interval and convergence to it.  In particular the declared initialization $y_0=0$
selects $y_-$ for $0<x<x_c$.

Implicit differentiation of the same rule gives
\begin{equation}
 \frac{dy}{dx}=\frac1{(1-y)(1-3y)}.
 \label{eq:feedback-susceptibility}
\end{equation}
Writing $y=1/3+u$ gives $x=4/27-u^2+u^3$.  The susceptibility therefore
diverges with inverse square-root magnitude as $x\uparrow x_c$ along
either branch with opposite signs.  At the endpoint $F_x'=1$ so the
strict contraction proof no longer applies.  This control-parameter fold has an inverse-square-root susceptibility.
Replacing the iteration by a preconditioned solver leaves the implicit
derivative unchanged when sources are transformed as in
Eq.~\eqref{eq:preconditioned-response} while solver stability can change.
Even where iteration is the declared update here a physical time or a
physical branch-selection law requires its own model identification.
Applying this feedback to an arithmetic or TT construction uses a matching of the state and source normalisation together with the parameter maps.  These maps
determine whether the same fold mechanism is present.

\subsection*{Field-theory cross-checks}
The following previously established channels test the additional QFT
representation and individual parts of the general response theory.
They examine spectral weight and symmetry constraints in physical
representations complementing the operation-based examples above.

\subsection{A free field: preservation of an existing pole}
\label{sec:free-pole}

For a canonically normalised free scalar field of mass $m$
\begin{equation}
 C_\phi(Q^2)=\frac{1}{Q^2+m^2}.
 \label{eq:free-scalar-propagator}
\end{equation}
Applying the external representation
Eq.~\eqref{eq:external-gaussian-representation} gives
\begin{equation}
 C_{\phi,\tau}(Q^2)
 =\frac{e^{-\tau Q^2}}{Q^2+m^2}.
 \label{eq:smeared-free-propagator}
\end{equation}
The multiplier is entire and nonzero at every finite $Q^2$.
It therefore leaves the pole location and order unchanged.
For $m=0$ it also leaves the residue unchanged
\begin{equation}
 \lim_{Q^2\to0}Q^2C_{\phi,\tau}(Q^2)=1.
 \label{eq:free-massless-residue}
\end{equation}
This verifies the pole-preservation statement surrounding
Eq.~\eqref{eq:pole-residue}. Pole generation belongs to the
dynamical step.

\subsection{A free composite field: a bubble continuum}
\label{sec:free-composite}

Let $\phi$ be a free real scalar of mass $m>0$ and choose the normal-ordered composite
\begin{equation}
 \cO(x)=\frac{1}{\sqrt2}:\!\phi^2(x)\!:.
 \label{eq:phi2-operator}
\end{equation}
Wick's contraction rule~\cite{Wick:1950} gives the connected Euclidean two-point function
\begin{equation}
 C_{\phi^2}(Q^2)
 =\int\frac{d^4p}{(2\pi)^4}
 \frac{1}{(p^2+m^2)((p+Q)^2+m^2)}.
 \label{eq:phi2-bubble}
\end{equation}
After a momentum-independent subtraction dimensional regularisation gives
\begin{equation}
 C_{\phi^2,R}(Q^2)-C_{\phi^2,R}(0)
 =-\frac{1}{16\pi^2}\int_0^1dx\,
 \ln\left[1+\frac{x(1-x)Q^2}{m^2}\right].
 \label{eq:phi2-subtracted}
\end{equation}
Analytic continuation has a branch point at $s=4m^2$.
With the spectral convention of Appendix~\ref{app:kl} the nonlocal part has
\begin{equation}
 \rho_{\phi^2}(s)
 =\frac{1}{16\pi^2}
 \sqrt{1-\frac{4m^2}{s}}\,
 \theta(s-4m^2),
 \label{eq:phi2-spectral-density}
\end{equation}
with the normalisation in Eq.~\eqref{eq:phi2-operator}.
There is no isolated one-particle atom.
This example separates a positive threshold density from an isolated pole.
In the massless limit the continuum begins at $s=0$.  Its spatial response
can be evaluated at separated points where the subtraction at zero
momentum is unnecessary.  Write $r=|\mathbf x|>0$.  The four-dimensional
massless propagator is $G(t,\mathbf x)=[4\pi^2(t^2+r^2)]^{-1}$ and
Wick contraction with the normalisation in Eq.~\eqref{eq:phi2-operator}
gives $\langle\cO(t,\mathbf x)\cO(0)\rangle_c=G(t,\mathbf x)^2$.
The response to a time-independent source is the zero-frequency kernel
obtained by integration over Euclidean time
\begin{equation}
 \begin{aligned}
 K_{\rm cont}(r)
 &=\frac1{16\pi^4}\int_{-\infty}^{\infty}
                  \frac{dt}{(t^2+r^2)^2}
   =\frac1{32\pi^3r^3},\qquad r>0.
 \end{aligned}
 \label{eq:composite-static-continuum}
\end{equation}
The substitution $t=r\tan\theta$ reduces the integral to
$r^{-3}\int_{-\pi/2}^{\pi/2}\cos^2\theta\,d\theta$.
Equivalently the constant massless density in
Eq.~\eqref{eq:phi2-spectral-density} gives the same result by integrating
the static Yukawa kernels
$\int_0^\infty ds\,e^{-\sqrt{s}r}/(64\pi^3r)$.
Local subtraction terms are supported at $r=0$ and leave this expression
unchanged.  The threshold continuum therefore gives a long-distance algebraic response.
A continuum reaching zero mass superposes all small masses and produces
the algebraic $r^{-3}$ tail in this example.  The zero-mass atom of
Eq.~\eqref{eq:static-gaussian-pole-kernel} selects an isolated propagation
channel and supplies the Coulombic $1/r$ term.  The threshold tail is
determined by the full spectral integral whereas the isolated collective
pole admits the local Fredholm decomposition.

\subsection{An unbroken conserved current without a massless pole}
\label{sec:conserved-current-example}

Consider a free massive Dirac field and its exactly conserved global current
$J_\mu=\bar\psi\gamma_\mu\psi$.
The Ward identity makes the separated-point two-point function transverse
\begin{equation}
 \langle J_\mu(Q)J_\nu(-Q)\rangle
 =(Q_\mu Q_\nu-Q^2\delta_{\mu\nu})\Pi_V(Q^2)
 +\text{contact terms}.
 \label{eq:vector-transverse}
\end{equation}
It also fixes the charge form factor of a one-particle fermion at zero transfer.
Nevertheless the vacuum spectral density of $\Pi_V$ starts at the fermion--antifermion threshold
\begin{equation}
 \rho_V(s)\ \propto\
 \left(1+\frac{2m^2}{s}\right)
 \sqrt{1-\frac{4m^2}{s}}\,
 \theta(s-4m^2),
 \label{eq:vector-threshold}
\end{equation}
and contains no massless atom.
The example directly separates charge normalisation from the existence of a long-range vacuum channel.

\subsection{A Goldstone channel: symmetry plus state information}
\label{sec:goldstone-example}

Let an exact continuous global symmetry have conserved current $J_A^\mu$ and
let a local Lorentz-scalar field $\Phi$ satisfy
\begin{equation}
 v_A\equiv\langle\Omega|\delta_A\Phi(0)|\Omega\rangle\neq0.
 \label{eq:goldstone-order-parameter}
\end{equation}
The nonzero expectation value specifies a vacuum that breaks the symmetry of the action.
Apply Eq.~\eqref{eq:ward-contact} to the mixed correlator.  After subtracting
polynomial contact contributions from that correlator retain the
order-parameter contact term on the right-hand side of the Ward identity.
With the Fourier convention of Section~\ref{sec:criteria} this gives
\begin{equation}
 p_\mu G_A^\mu(p)=v_A,
 \qquad
 G_A^\mu(p)=\int d^dx\,e^{\iu p x}
 \langle\Omega|T J_A^\mu(x)\Phi(0)|\Omega\rangle.
 \label{eq:goldstone-ward}
\end{equation}
Indeed the Fourier transform of the divergence is $-\iu p_\mu G_A^\mu$
and that of its Ward contact term is $-\iu v_A$.
For a Lorentz-invariant vacuum there is no independent vector so Lorentz covariance gives $G_A^\mu(p)=p^\mu F_A(p^2)$.
Consequently modulo the contact terms already removed the longitudinal form factor has the singular part
\begin{equation}
 F_A(p^2)=\frac{v_A}{p^2+\iu0},
 \label{eq:goldstone-pole}
\end{equation}
which is the massless Goldstone pole~\cite{Goldstone:1962}.
The conclusion required both the exact identity and Eq.~\eqref{eq:goldstone-order-parameter}.
With $v_A=0$ the same Ward identity does not force the pole.
For an operator carrying Lorentz indices additional covariants may occur and
the one-form-factor argument must be replaced by the corresponding tensor
decomposition.

\subsection{A separable response kernel}
\label{sec:separable-kernel}

Let the response operator on $\mathcal X_{\cC}$ be rank one
\begin{equation}
 \cB(z)=\lambda(z)|r\rangle\langle\ell|,
 \qquad \langle\ell|r\rangle=1,
 \label{eq:rank-one-kernel}
\end{equation}
with $\lambda$ analytic.
The resolvent is obtained exactly
\begin{equation}
 [I-\cB(z)]^{-1}
 =I+\frac{\lambda(z)}{1-\lambda(z)}
 |r\rangle\langle\ell|.
 \label{eq:rank-one-resolvent}
\end{equation}
Suppose $\lambda(0)=1$ and $\partial_z\lambda(0)\neq0$.
For the response-space maps
$\mathsf S_X:\mathbb C^m\to\mathcal X_{\cC}$ and
$\mathsf R_X:\mathcal X_{\cC}\to\mathbb C^m$ the singular part is
\begin{equation}
 C_{\rm sing}(z)
 =-\frac{\mathsf R_X|r\rangle
          \langle\ell|\mathsf S_X}
         {\partial_z\lambda(0)z}
 +O(1).
 \label{eq:rank-one-residue}
\end{equation}
The pole disappears from the measured channel if either matrix factor
vanishes even though the kernel still has eigenvalue one.  This solvable
example contains every algebraic condition in
Theorem~\ref{thm:kernel-pole}.  Only when the response is the unfiltered
physical correlator and the hypotheses of
Theorem~\ref{thm:fredholm-kl-bridge} hold does the displayed residue equal the
Hermitian positive-semidefinite atom $Z_0$.  The two source factors determine the weight carried by the zero of the denominator.

\subsection{Spectral extraction and approximation control}
\label{sec:verification}

In lattice and FRG calculations together with Dyson--Schwinger and Bethe--Salpeter approaches the
extraction of a spectral pole involves both the definition of the
renormalised channel and control of its finite approximation. 

\paragraph{Renormalised channel.}
The action and vacuum or ensemble specify the system with a chosen spacetime signature and dimension.  The operator basis and mixing matrix together with the EOM/BRST quotient and spin/internal projector specify the physical channel. The operator source--readout maps and response space together with its projector and realisation maps specify $\mathsf S_X$ and $\mathsf R_X$.  In a gauge theory the
Slavnov--Taylor residuals and the treatment of ghost and gauge-fixing
sectors are part of these data.  Closure is tested by
Eq.~\eqref{eq:response-channel-invariance}.  If it fails the calculation
requires an enlarged block or Eq.~\eqref{eq:feshbach-response}.

\paragraph{Spectral observable.}
The spectral observable is the projected connected Euclidean correlation
matrix of centred insertions.  Removed one-point products and subtracted
contact polynomials have distinct origins and are recorded separately.
For the Poincar\'e-invariant vacuum spectral branch of
Appendix~\ref{app:kl} a candidate massless atom is tested through
\begin{equation}
 Z_0=\lim_{Q^2\to0}Q^2
 \bigl[C(Q^2)-P(Q^2)\bigr]
 \label{eq:verification-residue}
\end{equation}
or the corresponding positive residue matrix.
The exchange of this limit with the spectral integral uses the continuum
integrability condition of Proposition~\ref{prop:extract-atom}.
The spectral bridge requires this unfiltered projected matrix to be the same
response used in the Fredholm factorisation with common vacuum centring and source pairing together with the same continuation and contact subtraction.
An ensemble calculation must instead supply a spectral
representation and diagnostic appropriate to that state.
For a spectral-FRG calculation the spectral representation and its
continuation must be specified in the truncation.  For example
Ref.~\cite{Horak:2024spectral} computes two- and four-point spectral
functions in scalar $\phi^4$ theory in three spacetime dimensions directly
in a Lorentzian formulation.

\paragraph{Finite-volume scaling.}
In a massive theory the large-volume behaviour of an isolated stable
particle is treated in Ref.~\cite{Luscher:1986stable} and that of
two-particle scattering levels in Ref.~\cite{Luscher:1986scattering}.
Applying a level-scaling test requires the finite-volume assumptions and
spectral separation appropriate to the state being extrapolated.  The
massive-theory estimates do not supply a zero-mass extrapolation theorem.
A pole claim requires a
sequence of volumes and resolutions showing the scaling of the candidate energy and matrix element together with their separation from multiparticle levels.
When the infinite-volume continuum also reaches $s=0$ the cumulative weight near zero distinguishes atomic concentration from continuous threshold weight.
In the isolated-state setting the result must survive
the continuum and infinite-volume extrapolations.

\paragraph{Kernel test.}
In the two-particle realisation the dimensionless operator
$\cB=\cK_{\rm 2PI}G_0$ is constructed in the same projected basis.
Following its relevant eigenvalue branches in the complex external
invariant $z$ gives $\lambda(0)$ and $\partial_z\lambda(0)$.  The associated
left and right eigenvectors together with the source and readout give
both matrix factors in Eq.~\eqref{eq:rank-one-residue}.  Onset along a control-parameter
family also calls for a separate evaluation of
$\partial_\eta\lambda(0;\eta_c)$ with higher-order analysis if it
vanishes.  At intermediate FRG scale this remains a flow diagnostic. A physical pole is a property of the endpoint response.

For a nonnormal discretised operator sensitivity is assessed through
the right and left eigenvector residuals and the condition number
\begin{equation}
 \kappa_N
 =\frac{\|r_N\|\,\|\ell_N\|}
        {|\langle\ell_N,r_N\rangle|}.
 \label{eq:eigenvalue-condition-number}
\end{equation}
The smallest singular value of $I-\cB_N$ or an equivalent resolvent bound
should accompany the eigenvalue.  A small value of
$|1-\lambda_N(0)|$ without overlap and conditioning and convergence information is not a pole determination.

\paragraph{Approximation analysis.}
Approximation errors are tested by controlled enlargement of the operator and momentum basis together with regulator variation and quadrature refinement and changes of finite-volume parameters.
Stability under these changes tests sensitivity to the approximation.
The observed spread is not a certified truncation bound. Certification
requires a remainder estimate relative to the limiting operator.
Once the assumptions and normalisation of
Proposition~\ref{prop:gaussian-unique} are fixed the Gaussian averaging
prescription is held fixed during this analysis while FRG regulator
variation probes the truncation uncertainty.

For a continuum Fredholm claim all approximants are realised on a common
response space.  On a compact set $K$ in the $z$ plane a sufficient
operator-norm convergence condition of the form
\begin{equation}
 \sup_{z\in K}
 \|\widetilde\cB_N(z)-\cB(z)\|\longrightarrow0,
 \label{eq:uniform-kernel-convergence}
\end{equation}
together with convergence of $\partial_z\cB_N$ and $\mathsf S_{X,N}$ and $\mathsf R_{X,N}$ supplies the required approximation control.
This condition is directly checkable when a derivation supplies an explicit
remainder bound.  Alternatively one may bound the nested increments by a
summable sequence and supply the computable tail estimate
$\sum_{n\geq N}\varepsilon_n(K)\leq r_N(K)\to0$ from
Eqs.~\eqref{eq:summable-increment-bound}--\eqref{eq:summable-increment-tail}.
This effective tail controls the continuum error once a consistency argument identifies
the limit with the exact response Jacobian.

For a coarse-grained closure the reference-state comparison starts with
$E_{\tau_q}(X_*;z)$.  If it is nonzero an independent fixed-point
displacement estimate is needed.  Lipschitz constants for
$D_X\mathfrak F_{\tau_q}$ and every base-point-dependent source or readout
then convert that displacement into kernel and map-error terms.  In the
invertible-sector comparison the remaining data include $\|D_XE_{\tau_q}\|$ and $\|\mathcal C_{\tau_q}^{-1}\|$ together with the resolvent constants
in Eqs.~\eqref{eq:closure-stability-constants}--\eqref{eq:residue-defect-bound}
and the transported-map discrepancies
$\|\widetilde{\mathsf S}_X-\mathsf S_{\rm ind}\|$ and
$\|\widetilde{\mathsf R}_X-\mathsf R_{\rm ind}\|$.
A residue comparison concerns the reference root $z_*$ and the perturbed
root $\widetilde z_*$.  Equation~\eqref{eq:critical-root-displacement}
and the analytic-variation bounds transport the Riesz numerator and
$\partial_z\lambda$ between these points.  This step is what permits the
comparison of residues at their actual generally different roots.
If $\gamma$ is a contour in the spectral $\zeta$ plane that
encloses the eigenvalue one and remains in the resolvent set the relevant
Riesz projections are
\begin{equation}
 \Pi_{\gamma,N}(z)
 =\frac{1}{2\pi\iu}\oint_\gamma
  (\zeta-\widetilde\cB_N(z))^{-1}\,d\zeta.
 \label{eq:riesz-truncation-test}
\end{equation}
Their convergence is required alongside that of the eigenvalues.  Isolated-eigenvalue
approximation and its error bounds are treated in Chapter~5 of
Ref.~\cite{Chatelin:1983}. Nonnormal resolvent sensitivity is the subject
of Ref.~\cite{TrefethenEmbree:2005}.  When
Eq.~\eqref{eq:uniform-kernel-convergence} is unavailable checks for spectral
pollution and pseudospectral sensitivity remain necessary diagnostics
but do not by themselves certify convergence to the exact operator.

\paragraph{Numerical results and reproducibility.}
The action and normalisation conventions together with the regulator and basis determine
the finite calculation. Numerical tolerances and eigenvector residuals
specify the accuracy of its evaluation.  The software in
Ref.~\cite{Guo2026CoarseGrainingFramework} provides supplementary finite-matrix response
calculations.  The companion compact-spectrum application
Ref.~\cite{Guo2026CompactSpectrumII} uses common spectral data to evaluate
effective couplings and mass scales. Its spectral-tail bounds and endpoint sensitivities distinguish numerical
approximation errors from changes in the scale prescription.
For a scalar observable realised by a Wick polynomial the one-particle
and multiparticle contributions depend on its centred expansion and source
pairing.  Its use in Appendix~\ref{app:kl} requires the positive vacuum
representation and subtraction convention of that same observable.
For a channel classification the result comprises the four quantities
in Eq.~\eqref{eq:classification-output} subject to the compatibility
assumptions in Section~\ref{sec:execution}. A finite-truncation estimate
must be distinguished from a limiting result supported by the convergence
bounds above.

The analytic examples distinguish pole preservation from continuous spectral weight and charge normalisation. They also exhibit the Goldstone mechanism and source visibility at a kernel zero.  Their closed-form responses allow individual
parts of a numerical calculation to be compared with exact results.
The closed-form identities supply comparison values for the dynamical and convergence analysis of interacting channels.

%% file: sections-en/sec06_summary.tex
\section{Conclusions}
\label{sec:summary}

A reduced description preserves collective response when it retains the
effects of eliminated variables on feedback and source coupling
together with their contribution to the observation.
Hidden dynamical variables also leave preparation and memory terms in
the reduced evolution.
Joint differentiable descent preserves the direct and propagated response.
Departures from descent obey the four-defect identity of
Theorem~\ref{thm:irreversible-error}.
Its comparison map may remove state coordinates.
The two resolvents in the update-defect term account for the enhanced
sensitivity near criticality.
The displaced-root estimates then carry branch and map errors to a
comparison of observed residues.

The arithmetic and reflecting-heat examples separate stationary
susceptibility from dynamical closure.
A hidden state can be eliminated from the former and remain necessary
for future observations.
Retaining that state gives a local update.
Eliminating it gives an exact memory equation.
A finite-memory approximation introduces an error that can be estimated
explicitly in the two-mode system.
The pole displacement decays geometrically with the retained memory
length. Differentiating the memory remainder adds a factor of that
length to the residue error.
The generated gauge--matter example gives a nonlinear counterpart.
Its residual state carries the mode products discarded by coefficient
smoothing. The compensation identity preserves their response through
successive resolutions.

The transfer pole of the two-mode system measures the slow-mode weight
in the response to constant forcing.
Its relaxation length diverges at the critical update.
In three spatial dimensions a constant contracted zero-mass residue
instead gives a Gaussian-smoothed Coulomb kernel.
Continuous spectral weight in the massless free composite gives a
faster algebraic tail.
The decay of the response thus records the source coupling to the
low-mass spectrum.

An isolated critical residue depends on the crossing slope and both
source overlaps together with the complementary resolvent.
For the reciprocal Hilbert-space response of
Theorem~\ref{thm:native-positive-residue} a positive crossing matrix gives
a positive-semidefinite finite-rank residue under the stated endpoint
and isolation hypotheses.
Theorem~\ref{thm:fredholm-kl-bridge} identifies this residue with a
zero-mass atom when the response matches the physical vacuum covariance.
The free two-helicity realisation has positive null-shell weight.
The complete stress tensor obeys the exclusion in
Proposition~\ref{prop:stress-null-shell} for its positive covariant
conserved measure.

The generated colour-singlet example has a positive spectral bottom
and a visible massive radial mode.
Its source-mean derivative fixes the weight.
Its field realisation supplies the mass interpretation.
For finite evolutions the measured response depends jointly on the
state residual and the operator and readout errors.
Quadratic observations retain the cross terms between an admitted
source and the state already present.
These examples specify the source and approximation data entering the
general comparison.

The estimates apply locally to specified reference branches with
isolated critical roots and controlled resolvents.
Within this domain the same four response maps connect generated
operations to finite-memory models and quantum-field-theory channels.
The quotient determines which state distinctions can be removed.
The response derivative follows the source through the retained state.
Schur elimination includes the feedback of omitted variables.
The spectral analysis identifies the collective modes visible in the
resulting observation.

%% file: sections-en/app_kl_representation.tex
\section{Response parameters and physical spectral measures}
\label{app:kl}

The spectral identification in Section~\ref{sec:classification} matches the analytic response to the unfiltered connected vacuum covariance in a positive physical representation. The projection and source pairing are fixed together with the continuation and subtraction prescription and the endpoint conventions of Theorem~\ref{thm:fredholm-kl-bridge}.  We treat
local physical insertions and covariances constructed directly on the
physical helicity space.  For each realisation we specify the positive measure and its continuation together with the zero-mass atom.  Under those exact matching
hypotheses these data identify a Laurent residue of the response with a
vacuum spectral weight.

\subsection{Response parameters and their representations}

The internal Neumann eigenvalue $E_n$ determines heat attenuation in
Section~\ref{sec:neumann-example}. The transfer parameter there is
$z=w-1$ with $w$ the discrete input frequency variable.  The control
$\eta$ changes the actual update and is not $z$.  The arithmetic heat parameter orders a cost semigroup. A physical clock requires a calibrated readout of the evolution.
Theorem~\ref{thm:native-positive-residue} certifies a Laurent coefficient
in the analytic parameter actually specified.  For a finite self-adjoint
$A$ with orthonormal eigenvectors $v_j$ direct diagonalisation gives
\[
 S^*(wI-A)^{-1}S=\sum_j\frac{S^*v_jv_j^*S}{w-a_j}.
\]
Each residue is positive. This measure describes the spectrum of the transfer operator. A Lorentzian realisation instead uses the invariant mass $s=p^2$ and its Euclidean continuation $Q^2$. A spatial Fourier representation gives the spatial Green function. These identifications depend on the representation of the response parameter.

The positive measure is obtained from connected Wightman two-point functions of centred physical operators.
Time-ordered correlators and commutators follow from the Wightman measure through their respective boundary prescriptions. Positivity belongs to the Wightman measure and to its physical channel compressions.
The invariant-mass representation has its origins in
Refs.~\cite{Kallen:1952,Lehmann:1954}. The vacuum framework used below is
that of Ref.~\cite{Streater:2000pct}.  For tensor channels we state the
physical-polarisation compression separately and check its positivity.

\subsection{Physical realisations and their positive measures}

\begin{definition}[Physical spectral realisation]
\label{def:physical-spectral-realisation}
A physical spectral realisation consists of a positive-energy unitary representation of translations with an invariant vacuum. Centred source insertions define the connected covariance through one of the following constructions.
\begin{enumerate}[label=(\roman*)]
 \item Renormalised local physical operators act on the positive-metric
 vacuum Hilbert space. Their tensor Wightman measure is compressed onto
 the chosen physical polarisations.
 \item The covariance is constructed on the reduced physical helicity
 space. The states carry the unitary Poincar\'e representation of those
 helicities and sources are test sections of the corresponding
 polarisation bundle. Positivity and the spectrum condition hold for this reduced covariance and its continuation is defined on the same physical space.
\end{enumerate}
In both cases the realisation specifies the invariant-mass disintegration and physical source pairing together with the Euclidean continuation and local subtraction terms. The vacuum projection and integrability assumptions stated
below apply to that same covariance.
\end{definition}

A change of helicity frame is unitary on the source fibres. It preserves
both the positive measure and the source-contracted covariance. For a
directly constructed helicity covariance the null-shell compression in
Assumption~\ref{ass:physical-polarisation-compression} is this fibre
pairing. The following local-operator construction supplies the first
realisation. Eq.~\eqref{eq:realized-helicity-bridge} supplies the second.

For the invariant-mass spectral representation work in a positive-metric
physical representation with a Poincar\'e-invariant vacuum $|\Omega\rangle$
and the positive-energy spectrum condition.  We assume that the joint
spectral projection of the momentum operators at zero is
$E_P(\{0\})=|\Omega\rangle\langle\Omega|$. The vacuum is thus the unique zero-momentum state within the chosen representation.
For the local realisation let the operators $O_i$ with $i=1,\ldots,n$ form a finite block of renormalised local physical operators with $O_i^\dagger$ included when necessary and set
\begin{equation}
 v_i=\langle\Omega|O_i|\Omega\rangle,
 \qquad \Theta_i=O_i-v_i I.
 \label{eq:centred-physical-insertions}
\end{equation}
The two-point matrix used throughout this appendix is connected
\begin{equation}
 \begin{aligned}
 W_{ij}(x)
 &=\langle\Omega|\Theta_i(x)\Theta_j^\dagger(0)|\Omega\rangle\\
 &=\langle\Omega|O_i(x)O_j^\dagger(0)|\Omega\rangle
   -v_i\overline v_j.
 \end{aligned}
 \label{eq:wightman-matrix}
\end{equation}
The omitted constant has Fourier transform
$(2\pi)^d v_i\overline v_j\delta^{(d)}(p)$.  This vacuum contribution is supported at zero four-momentum. A massless particle shell has nonzero null momentum whereas a contact term is polynomial in momentum.  Since $E_P(\{0\})\Theta_j^\dagger|\Omega\rangle=0$
the connected measure has no atom at $p=0$.  Its possible $\delta_0(ds)$
term instead arises from disintegration on the nonzero future null shell.

Thermal states and other ensembles generally have spectral representations that depend separately on frequency and spatial momentum. Their response analysis uses these variables in place of the vacuum invariant-mass representation.

Inserting the spectral resolution of the four-momentum operator yields a
matrix-valued measure supported in the closed forward cone.  For tensor
operators write this Fourier-space Wightman measure as
$d\Sigma_{ij,AB}(p)$ where $A,B$ collect covariant tensor indices.
Positivity is first a statement on the positive-metric physical Hilbert
space for every smooth compactly supported physical test section $f_i^A(p)$
\begin{equation}
 \sum_{i,j}\int_{\overline V_+}
 \overline{f_i^A(p)}\,d\Sigma_{ij,AB}(p)\,f_j^B(p)\geq0.
 \label{eq:full-wightman-positivity}
\end{equation}

\begin{assumption}[Positive physical-polarisation compression]
\label{ass:physical-polarisation-compression}
For $s>0$ a tensor or spin channel $a$ is represented by a measurable
orthonormal family
$\{\varepsilon_{a,r}(p)\}_{r=1}^{N_a}$ in the positive-norm physical
polarisation space.  Its channel measure is obtained from the compression
\begin{equation}
 d\Sigma^{(a)}_{ij;rs}(p)
 :=\overline{\varepsilon_{a,r}^{A}(p)}\,
   d\Sigma_{ij,AB}(p)\,\varepsilon_{a,s}^{B}(p),
 \label{eq:physical-polarisation-compression}
\end{equation}
followed by Lorentz-covariant disintegration with respect to $s=p^2$ and a
positive contraction in the polarisation indices.  Here a positive
contraction means
$X\mapsto\operatorname{tr}_{\rm pol}[H_a(p)X]$ for a specified measurable
$0\preceq H_a(p)\preceq I$ that transforms covariantly under a change of
polarisation frame. $H_a=I$ gives the full irreducible block.  If $s=0$ is
included the endpoint compression is defined separately on the null shell by an
orthonormal family spanning the specified physical helicities.  The
Euclidean tensor projector used in the channel is required after analytic
continuation to induce this same helicity compression.  The zero-mass limit is defined for the compressed forms. The null-shell helicity compression determines this endpoint independently of the full massive-spin projector.  For a scalar channel the compression is
the identity.  A measurable change of physical polarisation frame is unitary
on each polarisation fibre and does not affect positivity.
\end{assumption}

Under Assumption~\ref{ass:physical-polarisation-compression} denote the
invariant-mass part of the compressed spectral measure by
$d\rho^{(a)}_{ij}(s)$ on $s\geq0$.  A Hilbert-space
compression and a positive contraction preserve
Eq.~\eqref{eq:full-wightman-positivity}.  Hence for every
$c\in\mathbb C^n$ and every nonnegative Borel test function $f$
\begin{equation}
 \sum_{i,j}\bar c_i
 \int_0^\infty f(s)\,d\rho^{(a)}_{ij}(s)c_j
 \geq0.
 \label{eq:matrix-measure-positive}
\end{equation}
Thus $d\rho^{(a)}$ is positive semidefinite as a matrix-valued measure.
The physical channel compression is essential to this positivity statement.
The Fourier transform of a commutator contains the difference of positive- and negative-energy boundary values and is not itself a nonnegative scalar measure on all four-momenta.

\subsection{The zero-mass stress-tensor channel}
\label{sec:stress-null-shell}

For a complete covariant stress tensor positivity before any channel
compression imposes an additional restriction at the null shell.  This is
stronger input than positivity of a projected response alone.

\begin{proposition}[Null-shell TT exclusion for a complete positive covariant stress tensor]
\label{prop:stress-null-shell}
In four-dimensional Minkowski spacetime work in a positive-metric representation with positive energy with a Poincar\'e-invariant vacuum
$|\Omega\rangle$ satisfying
$E_P(\{0\})=|\Omega\rangle\langle\Omega|$ so that the vacuum is the unique
zero-momentum state in the chosen representation.  Let $T_{\mu\nu}$ be a Hermitian symmetric conserved tensor that transforms exactly as a Lorentz tensor on the physical Hilbert space. Assume that its centred connected Wightman measure
$d\Sigma_{\mu\nu,\rho\sigma}$ is locally finite and positive on all
symmetric tensor test functions before TT compression.  On the nonzero
future null shell $H_0^+=\{p:p^0>0,\ p^2=0\}$ it then has the form
\begin{equation}
 \left.d\Sigma_{\mu\nu,\rho\sigma}(p)\right|_{H_0^+}
 =c\,p_\mu p_\nu p_\rho p_\sigma\,d\Omega_0(p),
 \qquad c\geq0,
 \label{eq:stress-null-shell-measure}
\end{equation}
where $d\Omega_0$ is a fixed invariant measure on that shell.  Its contraction
with every conserved source is zero.  Since the connected measure has no
atom at $p=0$ every conserved TT source matrix has zero invariant-mass atom
at $s=0$.
\end{proposition}

\begin{proof}
The nonzero null shell is a single Lorentz orbit.  Its restricted measure
has a covariant matrix density $H(p)$ relative to $d\Omega_0$. To obtain this density average compensated Lorentz translates against a smooth compactly supported
Haar weight of integral one.  Covariance leaves the measure unchanged.
Since the orbit map is a submersion this convolution smooths a
distributional section along the orbit and hence gives the asserted
density.  Positivity of the measure gives $H(p)\succeq0$.  This argument
is confined to one orbit and does not require a density in the mass
variable $s$.

At a reference null momentum $\ell$ factor $H(\ell)=VV^\dagger$ with
$V$ of full column rank. The case $H=0$ is immediate.  If $D(g)$ is the
symmetric-tensor representation of an element of the stabiliser of $\ell$
covariance gives $D(g)H(\ell)D(g)^\dagger=H(\ell)$.  Thus
$D(g)V=VU(g)$ and a left inverse of $V$ shows $U(g)U(g)^\dagger=I$.
Null rotations are unipotent in $D$ so their restrictions $U$ are both
unitary and unipotent and hence equal to the identity.  Every column of $V$ is
therefore fixed by all null rotations.

For the component calculation raise the tensor indices and choose a null
tetrad $(\ell,n,e_1,e_2)$ with $\ell\cdot n=1$ and
$e_i\cdot e_j=-\delta_{ij}$.  The null-rotation generators satisfy
$N_k\ell=0$ together with $N_kn=e_k$ and $N_ke_i=\delta_{ki}\ell$.  Write a symmetric
tensor with juxtaposition denoting tensor product as
\begin{align*}
 u={}&a\ell\ell+bnn+c_1(\ell n+n\ell)
       +\sum_i d_i(\ell e_i+e_i\ell)\\
    &+\sum_i f_i(ne_i+e_i n)+\sum_{i,j}g_{ij}e_i e_j,
       \qquad g_{ij}=g_{ji}.
\end{align*}
Applying $N_k$ to both indices and setting every coefficient to zero gives
$b=f_i=d_i=0$ and $g_{ij}=-c_1\delta_{ij}$.  Consequently
$u=a\ell\ell+c_1\eta$ where
$\eta=\ell n+n\ell-\sum_i e_i e_i$.
Conservation of the measure places every column of $V$ in the transverse
subspace. $\ell_\mu u^{\mu\nu}=c_1\ell^\nu=0$ then forces $c_1=0$.
The range of $H(\ell)$ is therefore spanned by $\ell\ell$.
Positivity and Lorentz transport give
Eq.~\eqref{eq:stress-null-shell-measure} with $c\geq0$.
The conservation identity $p_\mu J^{\mu\nu}=0$ makes the source contraction vanish.
The absence of a connected atom at $p=0$ proved above completes the
invariant-mass endpoint statement.
\end{proof}

The exclusion classifies the actual null-shell measure of the complete positive exactly covariant tensor. Its positive-mass spectrum remains a dynamical question. A generated tensor response or a discrete transfer residue falls within this classification when it admits the stated Wightman representation and satisfies the proposition's hypotheses. The general response theorems also apply to other representations.

\subsection{Time-ordered and Euclidean correlators}

Choose a spacelike subtraction point $Q_*^2>0$ and an integer
$n_{\rm sub}\geq0$ large enough for ultraviolet convergence.
The projected Euclidean form factors then have the matrix-valued subtracted dispersion representation
\begin{equation}
 C^{(a)}_{ij}(Q^2)
 =P^{(a)}_{ij}(Q^2)
 +(Q_*^2-Q^2)^{n_{\rm sub}}
  \int_0^\infty
  \frac{d\rho^{(a)}_{ij}(s)}
  {(s+Q^2)(s+Q_*^2)^{n_{\rm sub}}}.
 \label{eq:matrix-kl-euclidean}
\end{equation}
Here $P^{(a)}_{ij}$ is the finite subtraction polynomial including the allowed local contact terms.
For $n_{\rm sub}=0$ the dispersive term in Eq.~\eqref{eq:matrix-kl-euclidean} reduces to the ordinary Stieltjes transform. A scheme-dependent local polynomial may still be present.
Changing $n_{\rm sub}$ or $Q_*^2$ redistributes analytic terms between the integral and the polynomial but does not alter the spectral measure or an isolated pole residue.
The Euclidean axioms of Ref.~\cite{OsterwalderSchrader:1973} with the
reconstruction result corrected and extended in
Ref.~\cite{OsterwalderSchrader:1975} provide a Hilbert-space interpretation
when the full reconstruction hypotheses hold.  Reflection positivity is
one of these hypotheses rather than a sufficient condition by itself.

Equation~\eqref{eq:matrix-kl-euclidean} keeps distinct the external Euclidean variable $Q^2$ and the Lorentzian spectral variable $s$.
A Gaussian factor in external momentum multiplies the Euclidean correlator. It is not by that operation alone a redefinition of $d\rho(s)$ or a new
Wightman theory.

\subsection{Zero-mass atoms and the Euclidean diagnostic}

Decompose the spectral measure in a channel as
\begin{equation}
 d\rho(s)=Z_0\,\delta_0(ds)+d\rho_{\rm c}(s),
 \qquad \rho_{\rm c}(\{0\})=0,
 \label{eq:atom-continuum-decomposition}
\end{equation}
where $Z_0$ is a positive-semidefinite matrix.
The atom contributes $Z_0/Q^2$ to the singular part of Eq.~\eqref{eq:matrix-kl-euclidean}. The remaining terms from its subtraction factor are analytic at the origin.

\begin{proposition}[Euclidean extraction of the atom]
\label{prop:extract-atom}
For every vector $c$ let
$d\nu_c(s)=c^\dagger d\rho_{\rm c}(s)c/(s+Q_*^2)^{n_{\rm sub}}$.
Suppose $\nu_c$ is locally finite at $s=0$ with no atom there and satisfies
$\int_1^\infty s^{-1}d\nu_c(s)<\infty$.
Then entrywise and as a quadratic form
\begin{equation}
 Z_0=\lim_{Q^2\downarrow0}Q^2
 \bigl[C(Q^2)-P(Q^2)\bigr].
 \label{eq:atom-limit}
\end{equation}
\end{proposition}

\begin{proof}
The atom term in Eq.~\eqref{eq:matrix-kl-euclidean} multiplied by $Q^2$ tends to $Z_0$.
For a quadratic form $c$ the remaining term is
\begin{equation}
 Q^2(Q_*^2-Q^2)^{n_{\rm sub}}
 \int_{(0,\infty)}\frac{d\nu_c(s)}{Q^2+s}.
 \label{eq:continuum-limit-integral}
\end{equation}
Apart from the bounded prefactor $(Q_*^2-Q^2)^{n_{\rm sub}}$ the integrand is $Q^2/(Q^2+s)$.
On $(0,1]$ it is bounded by one and converges pointwise to zero. Local finiteness and the absence of an atom at zero permit dominated convergence.
On $[1,\infty)$ it is bounded by $Q^2/s$ and the assumed weighted ultraviolet bound makes the tail vanish.
Polarisation of the quadratic-form result gives the entrywise limit.
The polynomial term is removed before the limit and the result follows.
\end{proof}

Proposition~\ref{prop:extract-atom} applies to the unfiltered connected
vacuum correlator of the centred insertions.  Theorem~\ref{thm:fredholm-kl-bridge}
connects it to a response residue when the vacuum centring and operator projection agree together with the source pairing and the continuation and subtraction prescriptions for the endpoint correlator.

If the continuum falls outside the hypotheses of
Proposition~\ref{prop:extract-atom} the endpoint requires a separate spectral
analysis.
Even when the limit exists a branch cut can begin at $s=0$ and coexist with an atom.
The decomposition~\eqref{eq:atom-continuum-decomposition} distinguishes the atomic weight from the threshold continuum through their spectral measures.

The connection with long-range response is direct.
For $d>2$ away from coincident points
\begin{equation}
 \int\frac{d^dQ}{(2\pi)^d}\,
 \frac{e^{\iu Q\cdot x}}{Q^2}
 =\frac{\Gamma(d/2-1)}{4\pi^{d/2}}
  \frac{1}{|x|^{d-2}}.
 \label{eq:massless-long-range}
\end{equation}
Thus a nonzero $Z_0$ produces a power-law tail whereas the subtraction polynomial contributes only distributions supported at coincident points.
In $d=2$ the corresponding Green function is logarithmic.
A continuum beginning at $s=0$ can also generate long-distance power laws but it is not an isolated atom and must be analysed with its threshold density.

\subsection{Action of the Gaussian multiplier}

Smearing both insertions with $\cH_{\tau_Q/2}$ gives the external diagnostic
map
\begin{equation}
 C_{\tau_Q}(Q^2)
 =\cG^{\rm ext}_{\tau_Q}C(Q^2)
 =e^{-\tau_Q Q^2}C(Q^2).
 \label{eq:gaussian-euclidean-correlator}
\end{equation}
If $C(Q^2)=Z_0/Q^2+C_{\rm rem}(Q^2)$ near the origin with
$Q^2C_{\rm rem}(Q^2)\to0$ then
\begin{equation}
 C_{\tau_Q}(Q^2)=\frac{Z_0}{Q^2}
 +C_{\rm rem}(Q^2)-\tau_Q Z_0+o(1).
 \label{eq:gaussian-pole-expansion}
\end{equation}
The atom residue is unchanged.
Conversely multiplication of a regular form factor by an entire function cannot generate a pole.
This proves the pole-preservation statement in Section~\ref{sec:criteria}.

The multiplier acts on a fixed correlator. Its positive KL measure is defined by the original Wightman covariance. Gaussian filtering generally changes the Stieltjes form. In particular
$e^{-\tau_QQ^2}P(Q^2)$ is no longer a contact polynomial. In position space it is a short-range smearing of a contact distribution.
For this reason the physical contact subtraction and spectral measure are
defined from the unfiltered correlator before
Eq.~\eqref{eq:gaussian-euclidean-correlator} is applied.  Pole preservation
may be checked on the filtered object whereas physical positivity must be
checked on the original Wightman matrix.

The scale dependence of a full interacting correlator also enters through
the effective vertices and response operator generated by integrating
fluctuations.

\subsection{Non-Hermitian and gauge-theory operators}

For a non-Hermitian operator $O$ positivity is formulated with its centred
insertion $\Theta=O-\langle O\rangle I$ and
$\langle\Theta(x)\Theta^\dagger(0)\rangle$ or for a mixing block with the
enlarged connected matrix containing $\Theta$ and $\Theta^\dagger$.

In a covariantly gauge-fixed formulation the state space before the physical quotient has indefinite metric and elementary gauge-potential or ghost correlators need not possess positive spectral measures.
The positivity used in Eq.~\eqref{eq:matrix-measure-positive} applies to
physical BRST cohomology classes and to the physical-polarisation compression
of Assumption~\ref{ass:physical-polarisation-compression}.
Slavnov--Taylor identities control the gauge dependence and mixing of their
representatives.  Gauge invariance and the weight of a spectral atom are then
assessed in the physical cohomology class and polarisation channel.

%% file: sections-en/app_operator_hierarchy.tex
\section{Compatible operations, observations, and operator realisations}
\label{app:operator-mixing}

For the channel reduction in Section~\ref{sec:classification} invariance
of a retained subspace and descent through a projection impose different
block identities.  These identities determine the renormalised-operator realisation and basis covariance of Proposition~\ref{prop:basis-covariance}.

\subsection{Projection, restriction, and quotient compatibility}

For a linear update $B$ and projector $P$ put $Q=I-P$.  Invariance of the
retained subspace requires $QBP=0$ whereas descent through the coarse map
$X\mapsto PX$ requires $PBQ=0$.  The first prevents escape from an
initially retained state. The second prevents hidden states from changing
the next observation.  Indeed $PX=PY$ means $X-Y\in\Ran Q$ so
$PBX=PBY$ for all such pairs exactly when $PBQ=0$.  A readout descends
through $P$ exactly when $RQ=0$.  Neither a dimension cutoff nor smoothing
proves either identity.

In the finite Neumann protocol the spectrum and bounded heat action are
proved directly in Section~\ref{sec:neumann-example}. Multiplication and
projection are explicit bounded operations.  Their finite composition
defines the system studied rather than asserting nonlinear closure for
arbitrary omitted modes.  For a gauge or other equivalence relation one
must likewise show that the update preserves equivalent states and that
observations agree on them before passing to the quotient.  The physical
operator construction below implements this statement in its BRST and
separated-point setting.

The charged matter quotient of ref.~\cite{Guo2026CompactSpectrumI}
Sections~4 and~5 gives an explicit finite example. If $N$ is its neutral
line an operator $A$ on the sixteen-dimensional full module descends to
the fifteen-dimensional charged quotient precisely when $AN\subset N$.
For self-adjoint $A$ this is equivalent to $[A,P_N]=0$. The reduced
gauge action satisfies this condition. The full joint algebra contains
matrix units that connect the neutral and charged directions so its
response is computed on the full module or by complete Schur elimination.
The compressed matrix $P_{N^\perp}AP_{N^\perp}$ is the retained block in
that calculation. Source and readout maps together with the direct term retain the effect of
the complementary direction as in
Eqs.~\eqref{eq:feshbach-source}--\eqref{eq:feshbach-direct-response}.

Internal coordinates and external source coordinates have different
covariance laws even without QFT.  For a bounded state isomorphism $T$
and a nonsingular source--output basis matrix $M$ set
\[
 B'=TBT^{-1},\quad S'=TSM^*,\quad R'=MRT^{-1},\quad D'=MDM^*.
\]
Substitution in $D+R(I-B)^{-1}S$ proves
$\mathscr C'=M\mathscr C M^*$.  Full visible pole locations and orders
and the rank of a residue are unchanged. Positive residues remain positive.
Pure internal similarity leaves the output unchanged.  This argument
supplies general response covariance. The proof additionally
identifies Wightman and renormalised-operator data in a QFT realisation.

\subsection{Mixing and composite sources}

For operators with common exact quantum numbers renormalisation is a matrix problem
\begin{equation}
 O_R=ZO_B,
 \qquad
 \mu\frac{dO_R}{d\mu}=-\gamma O_R,
 \qquad
 \gamma=-\left(\mu\frac{dZ}{d\mu}\right)Z^{-1},
 \label{eq:mixing-rge}
\end{equation}
with the sign convention used in Eq.~\eqref{eq:gamma-transform}.
If the source term is $h_R^TO_R=h_B^TO_B$ then
$h_B=Z^Th_R$.
Two or more composite insertions generally require additional local counterterms polynomial in the sources.
Those counterterms are the source-space origin of the contact polynomial $P_{ij}(Q^2)$ in Eq.~\eqref{eq:matrix-kl-euclidean}.

Lorentz representation and internal charges together with ghost number and canonical dimension can be used to construct finite computational spaces
$\mathcal V_N$.  Membership in such a space does not imply invariance under
mixing.  In an effective field theory increasing the dimension bound enlarges
the space. Conclusions at finite $N$ are truncation dependent unless a
power-counting or operator-norm estimate controls the omitted sector.

\subsection{BRST, EOM, and descendant sectors}

In an anomaly-free gauge theory first form the on-shell space
$\overline{\mathcal V}=\mathcal V/\mathcal I_{\rm EOM}$.
The BRST differential $\mathsf s_{\rm B}$ preserves the EOM ideal and
therefore induces a nilpotent differential on $\overline{\mathcal V}$ under
the standard regularity assumptions.
A convenient off-shell basis used to determine counterterms contains three types of representatives
\begin{equation}
 O_G\in\Ker\mathsf s_{\rm B},
 \qquad
 O_A=\mathsf s_{\rm B}\Psi,
 \qquad
 O_E=\mathcal E_F^R.
 \label{eq:gauge-exact-eom}
\end{equation}
Here $\mathcal E_F^R$ is the renormalised insertion associated with a
local change of integration variables. For a regulated bosonic integral
with flat measure its defining identity is
\begin{equation}
 \mathcal E_F=F^a\frac{\delta S}{\delta\phi^a}
                 -\frac{\delta F^a}{\delta\phi^a},
 \qquad
 \langle\mathcal E_F\,\mathcal O\rangle
     =\left\langle F^a\frac{\delta\mathcal O}{\delta\phi^a}\right\rangle.
 \label{eq:quantum-eom-insertion}
\end{equation}
Repeated field labels include spacetime integration. The identity follows
by integrating the variation of $F^a\mathcal O e^{-S}$. For a general
measure the second term is its field-space divergence with the graded
signs for fermionic variables. Renormalisation includes the Jacobian and
operator-mixing counterterms so that this Schwinger--Dyson identity holds
for renormalised insertions. With local variation support its right-hand
side has contact support in the other insertion coordinates. These are
the EOM representatives used in $\mathcal I_{\rm EOM}$. Derivatives of
source-dependent insertions and the associated local source counterterms
are retained in the direct response.
The stress-tensor application uses the same renormalised representative
in Ref.~\cite{Guo2026EmergentGravity}.
Under the standard assumptions for local composite-operator
renormalisation the renormalisation matrix can be chosen triangular between physical cohomology classes and BRST-exact and equation-of-motion sectors~\cite{JoglekarLee:1976}.  Section~8.6 of
Ref.~\cite{Barnich:2000} gives the local BRST cohomological formulation.
BRST-exact insertions vanish between physical cohomology classes while EOM insertions reduce to contact terms in separated-point correlators.
After the mixing problem is closed these statements justify taking the ghost-number-zero cohomology in Eq.~\eqref{eq:physical-quotient}. The corresponding counterterms remain in the determination of $Z$.

This cohomological quotient is an algebraic construction.  Functional-analytic
claims made with bounded channel projectors or source--readout maps use the
explicit Banach completion and norm included in the channel data of
Definition~\ref{def:physical-channel}.  The required maps must extend
continuously to that completion. This specified completion fixes the operator norm.

Infrared singularities can obstruct decoupling at exceptional kinematics.
The on-shell gluonic matrix elements at zero momentum transfer in
Ref.~\cite{CollinsScalise:1994} provide a concrete example.
The classification therefore forms the quotient at nonexceptional kinematics and retains all contact terms before taking a zero-momentum limit only after the projected renormalised correlator has been constructed.

Total derivatives occupy a different position.
When the boundary term vanishes integration by parts gives
\begin{equation}
 \int d^dx\,h\,\partial_\mu V^\mu
 =-\int d^dx\,(\partial_\mu h)V^\mu.
 \label{eq:descendant-source-gradient}
\end{equation}
Thus a total derivative is redundant for a constant-source integrated
insertion or at zero insertion momentum when that limit is well defined.
A nonconstant source generally leaves the source-gradient term on the
right-hand side and does not make the insertion redundant.
For a local insertion carrying momentum $q$ however
\begin{equation}
 [\partial_\mu V^\mu](q)=-\iu q_\mu V^\mu(q),
 \label{eq:descendant-momentum}
\end{equation}
and the descendant can contribute away from $q=0$ or mix with other local operators.
It is retained whenever the external kinematics resolve it.

\subsection{Symmetry blocks and anomalies}

Let $R(g)$ be the action of an exact global or spacetime symmetry on the renormalised operator space.
A projector $\mathcal P$ defines a closed channel only if its range is stable under the mixing matrix and the unbroken symmetry action.
When the regulator temporarily violates a symmetry modified Ward or Slavnov--Taylor identities must be solved so that the physical $k\to0$ block obeys the unmodified identity.

An anomalous current is not placed in the same block as an exactly conserved current.
Its divergence contains a renormalised local insertion and that insertion participates in the mixing problem.
Similarly explicit mass breaking modifies the operator identity. Spontaneous breaking by the state changes its matrix elements while preserving the identity.

For the stress tensor improvement transformations of the form
\begin{equation}
 T_{\mu\nu}\longmapsto T_{\mu\nu}
 +(\partial_\mu\partial_\nu-\eta_{\mu\nu}\Box)X
 \label{eq:stress-improvement}
\end{equation}
with $X$ a renormalised Lorentz scalar preserve conservation.  A physical TT
polarisation satisfies $p_\mu\epsilon^{\mu\nu}=0$ and
$\eta_{\mu\nu}\epsilon^{\mu\nu}=0$ and hence
\[
 \epsilon^{\mu\nu}(-p_\mu p_\nu+\eta_{\mu\nu}p^2)=0.
\]
The improvement therefore vanishes under TT compression including at the
massless endpoint when the physical helicity compression of
Appendix~\ref{app:kl} exists.  Contact terms are treated using the same
subtraction convention.  Spin-zero contributions and the overlaps of an
individual representative with physical states need not be invariant.

Let $D_X=(\partial_\mu\partial_\nu-\eta_{\mu\nu}\Box)X$ be retained among
the representatives of a closed block.  When
$(T,D_X,\ldots)\mapsto(T+D_X,D_X,\ldots)$ is a finite nonsingular constant
basis change on that block the covariance statement of
Proposition~\ref{prop:basis-covariance} preserves the pole locations and orders
of the full correlation matrix and the ranks of
its isolated residues.  No such invariance is asserted for a single entry
after the descendant sector has been discarded.

\subsection{Proof of external and internal covariance}

\begin{proof}[Proof of Proposition~\ref{prop:basis-covariance}]
Let $O'=MO$ with constant finite and nonsingular $M$.
In the same reference vacuum the centred insertions obey $\Theta'=M\Theta$.
The connected Wightman matrices therefore satisfy
\begin{equation}
 W'(x)=MW(x)M^\dagger.
 \label{eq:wightman-congruence}
\end{equation}
Uniqueness of the distributional spectral decomposition then gives
$d\rho'=M\,d\rho\,M^\dagger$ and analytic continuation gives
$C'=MCM^\dagger$ including the transformed contact polynomial.

Suppose near $z=z_0$ the correlator matrix has a Laurent expansion
\begin{equation}
 C(z)=\sum_{r=-p}^{\infty}C_r(z-z_0)^r,
 \qquad C_{-p}\neq0.
 \label{eq:matrix-laurent}
\end{equation}
Then
$C'(z)=\sum_r MC_rM^\dagger(z-z_0)^r$.
Because $M$ and $M^\dagger$ are invertible
$MC_{-p}M^\dagger\neq0$.
The pole location and order are unchanged.
For a simple spectral pole $Z'=MZM^\dagger$ has the same rank as $Z$ and
\begin{equation}
 c^\dagger Z'c=(M^\dagger c)^\dagger Z(M^\dagger c)\geq0
 \label{eq:residue-positivity-transform}
\end{equation}
whenever $Z$ is positive semidefinite.

A linear symmetry constraint $DC=0$ is carried to
$(DM^{-1})C'=0$ after multiplication on the right by $M^\dagger$.
Thus the dimensions of its forbidden and constrained subspaces together with its allowed subspace are unchanged.

Now let $T:\mathcal X_{\cC}\to\mathcal X'_{\cC}$ be a bounded response-state
isomorphism.  Direct substitution gives
\begin{equation}
 \mathsf R_{X,T}(I-\cB_T)^{-1}\mathsf S_{X,T}
 =\mathsf R_XT^{-1}T(I-\cB)^{-1}T^{-1}T\mathsf S_X
 =\mathsf R_X(I-\cB)^{-1}\mathsf S_X.
 \label{eq:internal-coordinate-invariance-proof}
\end{equation}
Thus $T$ changes internal coordinates by similarity but does not transform
the physical output matrix.  The response projector transforms as
$P_{X,T}=TP_XT^{-1}$ so
$Q_{X,T}\cB_TP_{X,T}=T(Q_X\cB P_X)T^{-1}$. Response closure is invariant.
Combining $T$ with the external basis change
$M$ instead gives
\begin{align}
 &\mathsf R'_X(I-\cB')^{-1}\mathsf S'_X \notag\\
 &\qquad
 =M\mathsf R_XT^{-1}T(I-\cB)^{-1}T^{-1}T\mathsf S_XM^\dagger
 =M\mathsf R_X(I-\cB)^{-1}\mathsf S_XM^\dagger,
 \label{eq:combined-covariance-proof}
\end{align}
and the same congruence holds for $C_{\rm reg}$.

Similarity preserves the response spectrum and algebraic multiplicities.
At a simple critical eigenvalue
$|r'\rangle=T|r\rangle$ and
$\langle\ell'|=\langle\ell|T^{-1}$ whence
\begin{equation}
 \mathsf R'_X|r'\rangle=M\mathsf R_X|r\rangle,
 \qquad
 \langle\ell'|\mathsf S'_X
 =\langle\ell|\mathsf S_XM^\dagger.
 \label{eq:combined-overlap-transform}
\end{equation}
Invertibility of $M$ preserves vanishing or nonvanishing of both overlaps
and Eq.~\eqref{eq:matrix-algebraic-residue} gives
$Z'_{\rm alg}=MZ_{\rm alg}M^\dagger$.
These observations prove Proposition~\ref{prop:basis-covariance}.
\end{proof}

\subsection{Scale-dependent bases and finite schemes}

For $O'=M(\mu)O$ differentiation of Eq.~\eqref{eq:basis-change} and use of Eq.~\eqref{eq:mixing-rge} give
\begin{align}
 \mu\frac{dO'}{d\mu}
 &=-\left[M\gamma M^{-1}
 -\left(\mu\frac{dM}{d\mu}\right)M^{-1}\right]O',
 \label{eq:gamma-transform-proof}
\end{align}
which proves Eq.~\eqref{eq:gamma-transform}.
Physical poles of the exact correlator remain unchanged by a finite nonsingular scheme transformation but the residue coordinates and Wilson coefficients change together with the anomalous-dimension matrix.
At finite truncation two implementations need not be related by an exact
similarity transformation because the transformed operator basis can have
components outside $P_N\mathcal V_{\rm phys}$ and its realised response can
leave the retained range of a response projector $P_{X,N}$.  The off-block
response operators
\begin{equation}
 \Delta_N^{\rm out}=(I_X-P_{X,N})\cB P_{X,N},
 \qquad
 \Delta_N^{\rm in}=P_{X,N}\cB(I_X-P_{X,N})
 \label{eq:truncation-coupling-defects}
\end{equation}
measure the direct coupling omitted by the compression.  If they do not
vanish either the Schur complement
Eq.~\eqref{eq:feshbach-response} or a convergence argument such as
Proposition~\ref{prop:riesz-convergence} is required.  Regulator and basis variation measure the sensitivity of the finite approximation. Exact covariance follows from the full transformed response.

%% file: sections-en/app_kernel_proof.tex
\section{Descent, Gaussian, kernel, and error proofs}
\label{app:kernel-proof}

The proofs establish differentiability through a Banach quotient and Gaussian averaging estimates together with control of isolated response residues.  In the last step
the resolvent identity first controls the spectral projection. Root
displacement and variation of the source and readout then enter the residue
estimate of Proposition~\ref{prop:closure-stability}.

\subsection{Regular descent through a Banach quotient}

\begin{proof}[Proof of regular Banach descent]
Put $K=\Ker C$ and let $q:\mathcal X\to\mathcal X/K$ be the quotient
map with
\[
 \|q u\|_{\mathcal X/K}=\inf_{k\in K}\|u+k\|_{\mathcal X}.
\]
Since $K$ is closed $\mathcal X/K$ is Banach.  The map
\[
 \widehat C:\mathcal X/K\longrightarrow\mathcal Y,
 \qquad \widehat C(qx)=Cx,
\]
is a bounded bijection and hence a bounded isomorphism.  The quotient map $q$ is open since $q(B_{\mathcal X}(0,\rho))=B_{\mathcal X/K}(0,\rho)$ for every $\rho>0$.  Thus $q(U)$ is
open and so is $C(U)=\widehat C(q(U))$.  Fibre consistency defines a
unique map $\widetilde F:q(U)\times V\to\mathcal Y$ by
\[
 \widetilde F(qx,h)=CF(x,h),
 \qquad
 \bar F(y,h)=\widetilde F(\widehat C^{-1}y,h).
\]

Fix $(x,h)\in U\times V$ and choose $\rho>0$ with
$B_{\mathcal X}(x,\rho)\subset U$.  If $k\in K$ then
$CF(x+tk,h)=CF(x,h)$ for all sufficiently small $t$. Differentiation at
$t=0$ gives
\[
 CD_xF(x,h)k=0.
\]
It follows that
\[
 \widetilde A_{x,h}(qu):=CD_xF(x,h)u
\]
is a well-defined bounded operator on $\mathcal X/K$.  Indeed taking the
infimum over representatives gives
\[
 \|\widetilde A_{x,h}\|
 \leq \|C\|\,\|D_xF(x,h)\|.
\]

Let $\eta\in\mathcal X/K$ satisfy $\|\eta\|<\rho/2$.  By the definition
of the quotient norm there is a representative $u$ with $qu=\eta$ and
$\|u\|\leq2\|\eta\|<\rho$.  For $j$ small enough that $h+j\in V$ fibre
consistency and the joint Fr\'echet expansion at $(x,h)$ give
\begin{align}
 &\widetilde F(qx+\eta,h+j)-\widetilde F(qx,h)
   -\widetilde A_{x,h}\eta-CD_hF(x,h)j \nonumber\\
 &\qquad
 =C\bigl[F(x+u,h+j)-F(x,h)
   -D_xF(x,h)u-D_hF(x,h)j\bigr].
 \label{eq:quotient-frechet-remainder}
\end{align}
The right-hand side is
$o(\|u\|+\|j\|)=o(\|\eta\|+\|j\|)$.  This estimate is uniform over the
chosen near-minimal representatives because $\|u\|\leq2\|\eta\|$. Its
left-hand side depends only on $\eta$.  Hence $\widetilde F$ is
Fr\'echet differentiable at $(qx,h)$.

It remains to verify continuity of the derivative.  If
$(\zeta,h')\to(qx,h)$ in $q(U)\times V$ choose $u_\zeta$ with
$q(x+u_\zeta)=\zeta$ and
$\|u_\zeta\|\leq2\|\zeta-qx\|$.  For $\zeta$ sufficiently close to $qx$ the point $x+u_\zeta$ belongs to $U$.  Both
$CD_xF(x+u_\zeta,h')$ and $CD_xF(x,h)$ annihilate $K$ and therefore
their induced quotient operators satisfy
\[
 \|\widetilde A_{x+u_\zeta,h'}-\widetilde A_{x,h}\|
 \leq \|C\|\,
 \|D_xF(x+u_\zeta,h')-D_xF(x,h)\|\longrightarrow0.
\]
The same argument applied to $CD_hF$ proves continuity of the source
derivative.  Thus $\widetilde F$ and consequently $\bar F$ are $C^1$.
Composition with $\widehat C^{-1}$ gives precisely
Eq.~\eqref{eq:banach-source-descent}.  The induced derivatives do not
depend on a representative either by the construction above or by
uniqueness of the Fr\'echet derivative.  Replacing $CF$ by a fibre-constant
$C^1$ observation $O$ proves the final assertion in the same way.
\end{proof}

This proof also identifies the bounded isomorphism
$\widehat C:\mathcal X/\Ker C\to\mathcal Y$ supplied by the quotient
norm.  Its inverse acts on equivalence classes in the Banach quotient.  The error identity of
Theorem~\ref{thm:irreversible-error} applies more generally to any bounded
linear comparison map.

\subsection{Differentiating exact block elimination}

For fixed complementary projectors abbreviate the blocks by $B_{ij}$
and put $H=(I-B_{22})^{-1}$.  A dot denotes a total derivative along the
declared parameter and reference branch.  Then $\dot H=H\dot B_{22}H$
and the four Schur derivatives are
\begin{align}
 \dot B_{\rm eff}
 &=\dot B_{11}+\dot B_{12}HB_{21}
   +B_{12}\dot H B_{21}+B_{12}H\dot B_{21},\nonumber\\
 \dot S_{\rm eff}
 &=\dot S_1+\dot B_{12}HS_2+B_{12}\dot H S_2+B_{12}H\dot S_2,
 \nonumber\\
 \dot R_{\rm eff}
 &=\dot R_1+\dot R_2HB_{21}+R_2\dot H B_{21}+R_2H\dot B_{21},
 \nonumber\\
 \dot D_{\rm eff}
 &=\dot D+\dot R_2HS_2+R_2\dot H S_2+R_2H\dot S_2.
 \label{eq:four-schur-derivatives}
\end{align}
They follow by differentiating $(I-B_{22})H=I$ and the four product
formulae.  If the projectors depend on the parameter first differentiate
the projected blocks themselves including every projector derivative.
Together with $\dot B=\partial_aB+D_XB\,\dot X_*$ and the corresponding
identities for $S,R,D$ these formulae retain the actual reference-branch
dependence.  The product rule applies to all four maps including their
source and direct-response contributions.

\subsection{Uniqueness of the Gaussian semigroup}

\begin{proof}[Proof of Proposition~\ref{prop:gaussian-unique}]
Write
\begin{equation}
 \varphi_\tau(q)=\int_{\mathbb R^d}e^{\iu q\cdot x}\,d\mu_\tau(x).
 \label{eq:characteristic-semigroup}
\end{equation}
Weak continuity and convolution give
$\varphi_{\tau+s}(q)=\varphi_\tau(q)\varphi_s(q)$ and
$\varphi_\tau(q)\to1$ as $\tau\downarrow0$.  Orthogonal invariance includes
$x\mapsto-x$ so $\varphi_\tau$ is real.  Moreover
$\varphi_\tau(q)=\varphi_{\tau/(2n)}(q)^{2n}\geq0$.  It is in fact
strictly positive.  If $\varphi_\tau(q)=0$ for some $\tau>0$ then
\[
 0=\varphi_\tau(q)=\varphi_{\tau/n}(q)^n
\]
would imply $\varphi_{\tau/n}(q)=0$ for every positive integer $n$ in
contradiction with $\varphi_{\tau/n}(q)\to\varphi_0(q)=1$.  The continuous
real logarithm is therefore defined and the scalar semigroup has the form
\begin{equation}
 \varphi_\tau(q)=e^{-\tau\psi(q)}
 \label{eq:characteristic-exponent}
\end{equation}
with $\psi(q)\geq0$ continuous.

Square-root self-similarity also gives
$\varphi_\tau(q)=\varphi_1(\sqrt\tau q)$.  Comparison with
Eq.~\eqref{eq:characteristic-exponent} yields
 $\psi(rq)=r^2\psi(q)$ for $r>0$.  Orthogonal invariance then gives
$\psi(q)=D|q|^2$.  Finite nonzero covariance makes $D$ finite and strictly
positive.  Hence
\begin{equation}
 \widehat\mu_\tau(q)=e^{-D\tau|q|^2},
 \label{eq:gaussian-characteristic-proof}
\end{equation}
which proves Proposition~\ref{prop:gaussian-unique}.
\end{proof}

The exact square-root scaling fixes the homogeneity degree at two. Removing the second-moment assumption alone does not change that conclusion.
The stated covariance condition fixes a finite nonzero scale.  Without
orthogonal invariance the Gaussian covariance need not be isotropic.
Non-Gaussian stable semigroups require a different scaling exponent
$\tau^{1/\alpha}$ with $0<\alpha<2$ and relaxation of the
finite-second-moment assumption.  Dropping exact self-similarity permits
symmetric jump components including ones of finite variance. The retained
orthogonal invariance excludes a nonzero drift.

\subsection{Tensor lift and suppression away from zero transfer}

\begin{proof}[Proof of Theorems~\ref{thm:gaussian-tensor-lift}
and~\ref{thm:fixed-transfer}]
The Fourier transform converts convolution in the $i$th coordinate into
multiplication by $e^{-a_i\tau|p_i|^2}$.  These smooth multipliers and all
their derivatives are polynomially bounded so they act continuously on
tempered distributions.  Since the coordinate convolutions commute their
tensor product gives
\begin{equation}
 \prod_{i=1}^ne^{-a_i\tau|p_i|^2}
 =\exp\!\left[-\tau\sum_{i=1}^na_i|p_i|^2\right].
 \label{eq:tensor-product-proof}
\end{equation}
For the routed statement $\widehat F_{\rm red}$ is part of the data after the
overall delta distribution has been factored out.  Write
$\mathsf W=\diag(a_1I_d,\ldots,a_nI_d)$ and let
$T(P,\boldsymbol\xi)=AP+L\boldsymbol\xi$ be the adapted linear coordinate
isomorphism of Theorem~\ref{thm:gaussian-tensor-lift}.  With
$g_\tau(p)=e^{-\tau p^{\mathsf T}\mathsf Wp}$ the distributional identity
in Eq.~\eqref{eq:reduced-momentum-distribution} means that for every
$\Phi\in\mathcal S(\mathbb R^d\times\mathbb R^r)$
\begin{equation}
 \langle\widehat F_{\rm full},\Phi\rangle
 =(2\pi)^d\langle\widehat F_{\rm red},\Phi(0,\cdot)\rangle,
 \qquad
 \int_{\mathbb R^d}\delta^{(d)}(P)f(P)\,dP=f(0).
 \label{eq:reduced-distribution-pairing}
\end{equation}
Multiplication by the smooth Gaussian is therefore computed without
restricting an arbitrary distribution to $P=0$
\begin{align}
 \langle g_\tau(T)\widehat F_{\rm full},\Phi\rangle
 &=(2\pi)^d\left\langle\widehat F_{\rm red},
   g_\tau(T(0,\cdot))\Phi(0,\cdot)\right\rangle\nonumber\\
 &=(2\pi)^d\left\langle\widehat F_{\rm red},
   e^{-\tau\boldsymbol\xi^{\mathsf T}L^{\mathsf T}\mathsf WL
   \boldsymbol\xi}\Phi(0,\boldsymbol\xi)\right\rangle.
 \label{eq:routed-distribution-pairing}
\end{align}
Since $T(0,\boldsymbol\xi)=L\boldsymbol\xi$ parametrises all of
$\Ker\Pi$ this is exactly the reduced action in
Eq.~\eqref{eq:routed-gaussian-form} with
$M_L=L^{\mathsf T}\diag(a_1I_d,\ldots,a_nI_d)L$.  This matrix is positive
semidefinite.  On a subspace $V$ on which $M_L\succeq cI$ one has
$\boldsymbol\xi^{\mathsf T}M_L\boldsymbol\xi
\geq c\|\boldsymbol\xi\|^2$.  This proves
Theorem~\ref{thm:gaussian-tensor-lift} including its bound and the absence
of a conclusion on $\Ker M_L$.

For Theorem~\ref{thm:fixed-transfer}(i)
$q^{\mathsf T}Mq\geq\delta^2$ on $\Omega_\delta$ and hence
\begin{equation}
  \|e^{-\tau q^{\mathsf T}Mq}A\|_{L^\infty(\Omega_\delta)}
 \leq e^{-\tau\delta^2}
 \|A\|_{L^\infty(\Omega_\delta)}.
 \label{eq:uniform-proof}
\end{equation}
For part~(ii) put
$\chi_\tau(q)=e^{-\tau q^{\mathsf T}Mq}\chi(q)$.  Each derivative of
$\chi_\tau$ is a finite sum of a polynomial in $q$ and $\tau$ multiplied by
$e^{-\tau q^{\mathsf T}Mq}$ and a derivative of $\chi$.  The support is
fixed inside $\Omega_\delta$ so every Schwartz seminorm is bounded by a
polynomial in $\tau$ times $e^{-\tau\delta^2}$ and tends to zero.  Continuity
of the tempered distribution $A$ gives
\begin{equation}
 \langle A_\tau,\chi\rangle
 =\langle A,\chi_\tau\rangle\longrightarrow0.
 \label{eq:distribution-proof}
\end{equation}
Finally Eq.~\eqref{eq:amplitude-growth-condition} gives
$\|A_\tau^{\rm full}\|\leq Ce^{-\tau(\delta^2-a)}$ proving part~(iii).
\end{proof}

\subsection{Response factorisation and the closure defect}

\begin{proof}[Proof of Propositions~\ref{prop:response-factorisation}
and~\ref{prop:closure-stability}]
Response factorisation.
Apply the implicit-function theorem on Banach spaces to
$G(X,J)=X-\mathfrak F_{\tau_q}(X;z,\eta)
-\mathsf S_{X,\tau_Q}(z)J$.  The derivative $D_XG=I-\cB$ is invertible by
hypothesis so a unique differentiable solution branch exists locally.
Differentiating Eq.~\eqref{eq:sourced-self-consistency} at $J=0$ with
$\delta X=D_JX\,\delta J$ and
$\cB=D_X\mathfrak F_{\tau_q}(X_*;z,\eta)$ one obtains
\begin{equation}
 (I-\cB)\delta X=\mathsf S_{X,\tau_Q}\delta J.
 \label{eq:linearised-self-consistency}
\end{equation}
Where $I-\cB$ is invertible composing its solution with the readout and
differentiating the direct term in
Eq.~\eqref{eq:measured-response-output} gives
Eq.~\eqref{eq:coarse-grained-response-factorisation}.  Under
the hypotheses of Theorem~\ref{thm:kernel-pole} the same formula extends meromorphically.  This proves
Proposition~\ref{prop:response-factorisation}.

Transport of the closure defect.
Assume $\mathcal C_{\tau_q}$ is bounded and linear and
that the two maps in Eq.~\eqref{eq:coarse-graining-closure-defect} are
Fr\'echet differentiable.  The chain rule gives
\begin{equation}
 D_XE_{\tau_q}(X;z)
 =\mathcal C_{\tau_q}D_X\mathfrak F(X;z)
 -D_X\mathfrak F_{\tau_q}(\mathcal C_{\tau_q}X;z)
  \mathcal C_{\tau_q}.
 \label{eq:closure-defect-derivative}
\end{equation}
Thus an exact commuting coarse-graining has the response intertwining
Eq.~\eqref{eq:response-intertwining}. In an approximate closure
$D_XE_{\tau_q}$ measures its failure at the same base point.  Compactness and analyticity together with threshold isolation are additional properties of the realised response family.

For Proposition~\ref{prop:closure-stability}
Eq.~\eqref{eq:reference-branch-matching} and
Eq.~\eqref{eq:transported-fixed-point-residual} first give
$\mathfrak F_{\tau_q}(\bar X_*;z)=\bar X_*$.  Thus the derivative at
$\bar X_*$ is the response Jacobian at an actual coarse solution rather than
at a transported trial point.  On the comparison spaces multiply
Eq.~\eqref{eq:closure-defect-derivative} on the right by
$\mathcal C_{\tau_q}^{-1}$ to obtain
\begin{equation}
 \bigl[D_XE_{\tau_q}(X_*(z);z)|_Y\bigr]
 \mathcal C_{\tau_q}^{-1}
 =\cB_{\rm ind}(z)-\cB_{\tau_q}(z)=-\Delta\cB(z),
 \label{eq:closure-kernel-error-proof}
\end{equation}
This proves the first norm estimate.

Spectral projection and root displacement.
For $R(\zeta,z)=(\zeta-\cB_{\rm ind}(z))^{-1}$ the condition
$K_\gamma\epsilon<1$ permits the Neumann expansion of the perturbed
resolvent.  The resolvent identity gives uniformly on $K\times\gamma$
\begin{equation}
 \|R_{\tau_q}(\zeta,z)-R(\zeta,z)\|
 \leq\frac{K_\gamma^2\epsilon}{1-K_\gamma\epsilon}.
 \label{eq:resolvent-defect-bound-proof}
\end{equation}
Integrating
$(2\pi\iu)^{-1}[R_{\tau_q}(\zeta,z)-R(\zeta,z)]\,d\zeta$
around $\gamma$ proves Eq.~\eqref{eq:riesz-defect-bound}.  The same resolvent
estimate holds for $\cB_{\rm ind}+t\Delta\cB$ throughout $0\leq t\leq1$.
The associated Riesz projections form a norm-continuous path of finite-rank
idempotents and therefore have constant rank.  Since the rank at $t=0$ is
one $\Pi_{\tau_q}$ also has rank one.
The finite-rank trace identities
\begin{equation}
 \lambda_{\rm ind}=\operatorname{tr}(\cB_{\rm ind}\Pi_{\rm ind}),
 \qquad
 \lambda_{\tau_q}=\operatorname{tr}(\cB_{\tau_q}\Pi_{\tau_q})
 \label{eq:eigenvalue-riesz-trace}
\end{equation}
then imply Eq.~\eqref{eq:eigenvalue-defect-bound}. The term
$\Delta\cB\Pi_{\tau_q}$ has rank at most one while
$\cB_{\rm ind}(\Pi_{\tau_q}-\Pi_{\rm ind})$ has rank at most two.

On $\partial D_*$ Eq.~\eqref{eq:unperturbed-root-lower-bound} gives
$|\lambda_{\rm ind}(z)-1|\geq2d_\lambda$ whereas
Eq.~\eqref{eq:eigenvalue-defect-bound} bounds the perturbation by
$d_\lambda$.  Rouch\'e's theorem therefore gives exactly one zero
$\widetilde z_*$ of $\lambda_{\tau_q}-1$ in $D_*$ and proves
Eq.~\eqref{eq:critical-root-displacement}.

Residue at the displaced root.
At a common value of $z$ insert and subtract the two intermediate products
$\mathsf R_{\rm ind}\Pi_{\tau_q}\widetilde{\mathsf S}_X$ and
$\mathsf R_{\rm ind}\Pi_{\rm ind}\widetilde{\mathsf S}_X$.  The triangle inequality and Eq.~\eqref{eq:riesz-defect-bound} together with the uniform constants in
Eqs.~\eqref{eq:closure-map-error-constants}--%
\eqref{eq:same-point-numerator-defect} give
\begin{equation}
 \sup_{z\in D_*}
 \|\widetilde{\mathcal N}(z)-\mathcal N_{\rm ind}(z)\|
 \leq d_{\rm num}^{(0)}.
 \label{eq:residue-numerator-error-proof}
\end{equation}
The fundamental theorem of calculus along the line segment from $z_*$ to
$\widetilde z_*$ gives
\begin{equation}
 \|\widetilde{\mathcal N}(\widetilde z_*)
       -\widetilde{\mathcal N}(z_*)\|
 \leq L_{\rm num}\delta_z.
 \label{eq:numerator-root-transport-proof}
\end{equation}
Combining the last two estimates yields
\begin{equation}
 \|\widetilde{\mathcal N}(\widetilde z_*)
       -\mathcal N_{\rm ind}(z_*)\|
 \leq D_{\rm num}.
 \label{eq:cross-root-numerator-proof}
\end{equation}
The same argument applied to $\partial_z\lambda_{\tau_q}$ gives
\begin{equation}
 |\partial_z\lambda_{\tau_q}(\widetilde z_*)
   -\partial_z\lambda_{\rm ind}(z_*)|
 \leq d_{\lambda'}^{(0)}+L_{\lambda'}\delta_z
 =D_{\lambda'}.
 \label{eq:cross-root-slope-proof}
\end{equation}
The bound $D_{\lambda'}<\alpha/2$ gives $|\partial_z\lambda_{\tau_q}(\widetilde z_*)|\geq\alpha/2$.
The elementary difference-of-quotients estimate now gives
\begin{equation}
 \left\|-\frac{\widetilde{\mathcal N}(\widetilde z_*)}
                  {\partial_z\lambda_{\tau_q}(\widetilde z_*)}
          +\frac{\mathcal N_{\rm ind}(z_*)}
                  {\partial_z\lambda_{\rm ind}(z_*)}\right\|
 \leq\frac{2D_{\rm num}}{\alpha}
      +\frac{2\|\mathcal N_{\rm ind}(z_*)\|D_{\lambda'}}{\alpha^2},
 \label{eq:residue-error-proof}
\end{equation}
which is Eq.~\eqref{eq:residue-defect-bound}.
This proves Proposition~\ref{prop:closure-stability}.
\end{proof}

\subsection{Fredholm resolvent near a simple spectral zero}

\begin{proof}[Proof of Theorem~\ref{thm:kernel-pole}]
Under the hypotheses of Theorem~\ref{thm:kernel-pole} the analytic Fredholm theorem implies
that $[I-\cB(z)]^{-1}$ is a meromorphic operator-valued function on
$U$. See Chapter~VII of Ref.~\cite{Kato:1995}.  If
$1\notin\spec\cB(0)$ bounded invertibility is open
in operator norm.  The inverse is analytic near $z=0$ and insertion into
Eq.~\eqref{eq:coarse-grained-response-factorisation} proves
Theorem~\ref{thm:kernel-pole}(i).

Suppose now that $\lambda(0)=1$ is algebraically simple.  Analytic
perturbation theory supplies local analytic continuations $\lambda(z)$ and $|r(z)\rangle$ together with $\langle\ell(z)|$ after reducing $U$ if necessary satisfying
\begin{equation}
 \cB(z)|r(z)\rangle=\lambda(z)|r(z)\rangle,
 \quad
 \langle\ell(z)|\cB(z)=\lambda(z)\langle\ell(z)|,
 \quad
 \langle\ell(z)|r(z)\rangle=1.
 \label{eq:left-right-eigenvectors}
\end{equation}
The associated Riesz projection is
\begin{equation}
 \Pi_\lambda(z)=|r(z)\rangle\langle\ell(z)|.
 \label{eq:riesz-projector}
\end{equation}
On the complementary invariant subspace $I-\cB(z)$ remains invertible near
the origin.  Hence
\begin{equation}
 [I-\cB(z)]^{-1}
 =\frac{\Pi_\lambda(z)}{1-\lambda(z)}+\mathcal R_\lambda(z),
 \label{eq:fredholm-resolvent-decomposition}
\end{equation}
where $\mathcal R_\lambda$ is analytic.  The simple-zero condition gives
\begin{equation}
 1-\lambda(z)=-\partial_z\lambda(0)z+O(z^2).
 \label{eq:eigenvalue-crossing}
\end{equation}
Substitution in Eq.~\eqref{eq:coarse-grained-response-factorisation} yields
\begin{equation}
 C_{\rm sing}(z)
 =-\frac{\mathsf R_X(0)|r(0)\rangle
            \langle\ell(0)|\mathsf S_X(0)}
           {\partial_z\lambda(0)z}.
 \label{eq:fredholm-residue}
\end{equation}
The outer product is nonzero exactly when both displayed factors are
nonzero.  This proves part~(ii) and Eq.~\eqref{eq:matrix-algebraic-residue}.

Equation~\eqref{eq:fredholm-residue} is an algebraic matrix residue in the
chosen response convention and proves part~(iii).  Its identification with a
Wightman atom follows under the matching hypotheses of
Theorem~\ref{thm:fredholm-kl-bridge}.  For an approximate nonnormal kernel
those hypotheses must be tested independently of the critical eigenvalue.
\end{proof}

\subsection{Uniform truncation and Riesz projections}

Compactness of a finite matrix gives no convergence statement for the
operator it approximates.  The next proposition is a standard consequence
of norm-resolvent control and analytic perturbation
\cite{Kato:1995,Chatelin:1983}.  Its final assertion specifies the additional
simple-zero hypothesis needed to pass from spectral projections to response
residues at moving roots.

\begin{proposition}[Convergence of an isolated spectral subspace]
\label{prop:riesz-convergence}
Let $K\Subset U$ be compact and let $\widetilde\cB_N(z)$ and $\cB(z)$ act on
one Banach space and be operator-norm continuous on a neighbourhood of $K$.
Suppose
\begin{equation}
 \sup_{z\in K}\|\widetilde\cB_N(z)-\cB(z)\|\longrightarrow0.
 \label{eq:kernel-convergence-proof}
\end{equation}
Let $\gamma$ be a closed contour in the spectral $\zeta$ plane which lies in
the resolvent set of $\cB(z)$ for every $z\in K$ and is uniformly separated
from its spectrum.  Then for all sufficiently large $N$ the contour $\gamma$ also lies
in the resolvent set of $\widetilde\cB_N(z)$ and
\begin{equation}
 \sup_{z\in K}
 \|\Pi_{\gamma,N}(z)-\Pi_\gamma(z)\|\longrightarrow0.
 \label{eq:riesz-projection-convergence}
\end{equation}
For convergence of critical roots and residues assume additionally that
the kernel families are analytic on a common open neighbourhood of a closed
disk $D\subset\operatorname{int}K$.  Let $\gamma$ enclose exactly one
algebraically simple eigenvalue $\lambda(z)$ and no other spectrum for every $z\in D$ and
suppose $z_*\in\operatorname{int}D$ is the only zero of $\lambda-1$ in
$D$ with
\[
 \lambda(z_*)=1,\qquad \lambda'(z_*)\ne0.
\]
Assume the kernel derivatives and the analytic source and readout maps
converge uniformly on $D$ in their specified operator norms.  Then for
all sufficiently large $N$ the corresponding branch $\lambda_N-1$ has exactly one zero $z_{*,N}$ in $D$. This zero is simple and
$z_{*,N}\to z_*$.  The algebraic residues at their respective roots obey
\[
 -\frac{\mathsf R_{X,N}(z_{*,N})\Pi_{\gamma,N}(z_{*,N})
              \mathsf S_{X,N}(z_{*,N})}{\lambda_N'(z_{*,N})}
 \longrightarrow
 -\frac{\mathsf R_X(z_*)\Pi_\gamma(z_*)\mathsf S_X(z_*)}
              {\lambda'(z_*)}.
\]
This is a residue of the selected critical response. Analytic direct
terms do not contribute.  Algebraic simplicity in the spectral variable
alone does not give the required first-order zero in $z$.
\end{proposition}

\begin{proof}
Indeed the exact resolvents are uniformly bounded on the compact set
$K\times\gamma$.  Equation~\eqref{eq:kernel-convergence-proof} and a Neumann
series give existence and uniform convergence of the approximate resolvents.
Contour integration then proves
Eq.~\eqref{eq:riesz-projection-convergence}.  The explicit norm-continuity
hypothesis gives the asserted compact-set bound even for nonnormal kernels.
For the additional assertion nearby projections have the same rank one.
Analytic perturbation on these spectral subspaces gives uniform convergence
of $\lambda_N$ to $\lambda$ and using the kernel-derivative convergence
of $\lambda_N'$ to $\lambda'$~\cite{Kato:1995,Chatelin:1983}.
Since $\lambda-1$ is nonzero on $\partial D$ Rouch\'e's theorem gives
exactly one approximate zero in $D$ counted with multiplicity hence a
simple zero.  Applying the same argument to arbitrarily small disks about
$z_*$ gives $z_{*,N}\to z_*$.  Uniform derivative convergence and
continuity imply $\lambda_N'(z_{*,N})\to\lambda'(z_*)\ne0$. The
denominators are therefore bounded away from zero for large $N$.
Uniform convergence of the projections and source and readout maps gives
convergence of the numerators at these moving roots proving the residue
limit.  This is the qualitative version of the cross-root comparison in
Proposition~\ref{prop:closure-stability} specifically
Eqs.~\eqref{eq:critical-root-displacement} and~\eqref{eq:residue-defect-bound}.
\end{proof}

The norm hypothesis in Eq.~\eqref{eq:kernel-convergence-proof} can be
certified either by an explicit remainder estimate
$\sup_{z\in K}\|\cB(z)-\widetilde\cB_N(z)\|\leq r_N(K)$ with
$r_N(K)\to0$ or by the following a posteriori construction when all
approximants act on the same Banach space.

\paragraph{A computable summable-increment criterion.}
Let $U\subset\mathbb C$ be open and let
$\widetilde\cB_N:U\to\mathcal K(\mathcal X_{\cC})$ be analytic compact-operator
families.  Suppose that for every $K\Subset U$ there are certified bounds
$\varepsilon_N(K)\geq0$ and explicitly computable numbers $r_N(K)\geq0$
with certified decay to zero such that
\begin{equation}
 \begin{aligned}
 \sup_{z\in K}\|\widetilde\cB_{N+1}(z)-\widetilde\cB_N(z)\|
 &\leq\varepsilon_N(K),\\
 \sum_{n=N}^{\infty}\varepsilon_n(K)
 &\leq r_N(K),\qquad r_N(K)\longrightarrow0.
 \end{aligned}
 \label{eq:summable-increment-bound}
\end{equation}
Then the sequence converges locally uniformly to an analytic compact family
$\cB_\infty$ with the computable tail estimate
\begin{equation}
 \sup_{z\in K}\|\cB_\infty(z)-\widetilde\cB_N(z)\|
 \leq\sum_{n=N}^{\infty}\varepsilon_n(K)\leq r_N(K).
 \label{eq:summable-increment-tail}
\end{equation}
If a consistency argument for the truncation identifies
$\cB_\infty=D_X\mathfrak F(X_*;z,\eta)$ this tail bounds the truncation error in
Eq.~\eqref{eq:kernel-convergence-proof}.  The same construction applies to
the source and readout maps when their increments have corresponding
effective tail bounds.  Local uniform convergence of the analytic
kernel families also yields convergence of their $z$ derivatives on smaller
compact subsets of $U$.
The Banach-valued Weierstrass argument applies. Equation~\eqref{eq:summable-increment-bound} makes the sequence Cauchy in $C(K,\mathcal L(\mathcal X_{\cC}))$. Summation gives Eq.~\eqref{eq:summable-increment-tail} and norm closure preserves compactness.  Local uniform convergence preserves analyticity while the
Cauchy formula on a slightly larger compact set gives derivative convergence.
The explicit majorants $r_N(K)$ make the tail effective. Summability by
itself gives convergence but not a computable error bound.

For example certified computable constants $A(K)\geq0$ and
$0<\theta(K)<1$ with $\varepsilon_n(K)\leq A(K)\theta(K)^n$ give
$r_N(K)=A(K)\theta(K)^N/[1-\theta(K)]$.  This geometric majorant supplies
both the tail bound and its rate of decay.

If a threshold enters $U$ the meromorphic decomposition
Eq.~\eqref{eq:fredholm-resolvent-decomposition} may fail.  When operator-norm
convergence is unavailable stable-looking eigenvalues of finite nonnormal
matrices also require a spectral-pollution analysis.  These cases call for an
additional spectral argument beyond the isolated-pole theorem.

%% file: sections-en/declarations.tex
\section*{Acknowledgments}

AI-assisted tools were used for language editing and for drafting and
checking the numerical verification and formalization code. They also
assisted selected numerical consistency checks. The author is responsible
for the scientific content and conclusions.

\section*{Declarations}

Conflict of interest. The author declares no competing interests.

Data and code availability. This work uses no empirical dataset.
The analytic examples are specified by the equations and conventions in
the manuscript and can be reproduced directly from them.
Related finite-dimensional calculations use the
Coarse-Graining Genesis Framework~\cite{Guo2026CoarseGrainingFramework}
available at \url{https://doi.org/10.5281/zenodo.22067005}.

Funding. No funding was received for this work.